\documentclass[12pt]{article}
\RequirePackage{amsthm,amsmath}
\RequirePackage{natbib}

\usepackage{nicematrix}
\usepackage{mathrsfs}
\usepackage{mathtools}
\usepackage{amssymb}
\usepackage{amsbsy}
\usepackage{enumerate}
\usepackage{enumitem}
\usepackage{bm}
\usepackage{color}
\usepackage{booktabs}
\usepackage{hyperref} % not crucial - just used below for the URL
\usepackage{fullpage}
\usepackage{setspace}
\usepackage{graphicx}
\usepackage{stackrel}
\usepackage{caption}
\usepackage{mathabx}
\usepackage{subcaption}
\usepackage{multirow}
\usepackage{color}
\usepackage{verbatim}
\usepackage{bm}
\usepackage{graphics}
\usepackage{color}  % NOTE: Entering colored text - {\color{red}followed by a red 
\usepackage{float} %required for the placement specifier H
\theoremstyle{plain}
\newtheorem{theorem}{Theorem}
\newtheorem{lemma}{Lemma}
\newtheorem{proposition}{Proposition}
\newtheorem{corollary}{Corollary}

\graphicspath{{./Figures/}}
\allowdisplaybreaks

\graphicspath{{./figs/}}
\DeclareGraphicsExtensions{.eps,.ps,.pdf}

\newcommand{\be}{\begin{eqnarray}}
\newcommand{\ee}{\end{eqnarray}}
\newcommand{\ben}{\begin{eqnarray*}}
\newcommand{\een}{\end{eqnarray*}}

\usepackage{wrapfig}
\usepackage{float}

\usepackage{diagbox}

\usepackage{graphicx,enumerate}
\usepackage{delarray}
\usepackage{hyperref, subcaption}
\usepackage{rotating}
\usepackage{url}
\usepackage{colortbl,color}
\usepackage{multirow,amssymb,amsmath, amsfonts, bm}
\usepackage[title]{appendix}
\usepackage{mathrsfs}
\usepackage{dsfont}
\usepackage{stmaryrd}

\usepackage{lastpage}
\usepackage{makecell}
\usepackage{algorithm}
\usepackage{algorithmic}

\numberwithin{equation}{section}
\usepackage[title]{appendix}

\usepackage[utf8]{inputenc}
\usepackage{chngcntr}
\usepackage{apptools}
\AtAppendix{\counterwithin{lem}{section}}
\AtAppendix{\counterwithin{prop}{section}}

\newtheorem{myRem}{Remark}
\newtheorem{case}{Special case}
\newtheorem{example}{Example}

\newtheorem{defin}{Definition}
\newtheorem{assump}{Assumption}

\newtheorem{thm}{Theorem}
\newtheorem{prop}{Proposition}
\newtheorem{cor}{Corollary}

\newcommand{\R}{\mathbb R}
\newcommand{\p}{{\rm I}\kern-0.18em{\rm P}}
\newcommand{\1}{{\rm 1}\kern-0.24em{\rm I}}
\newcommand{\E}{{\rm I}\kern-0.18em{\rm E}}

\newcommand{\cov}{\text{Cov}}

\newcommand{\uderbar}[1]{\underset{\raise0.3em\hbox{$\smash{\scriptscriptstyle-}$}}{#1}}
\def\boxit#1{\vbox{\hrule\hbox{\vrule\kern6pt\vbox{\kern6pt#1\kern6pt}\kern6pt\vrule}\hrule}}

\begin{document}

\title{  Covariate Selection for Optimizing Balance with an Innovative Adaptive Randomization Approach}
%\author{Ziqing~Guo,~Yang~Liu\thanks{Yang Liu is an assistant professor at Renmin University of China. Liu's research is supported by the National Natural Science Foundation of China grant nos.~12301324.}, ~Lucy~Xia\thanks{Corresponding author: lucyxia@ust.hk. Lucy Xia is an assistant professor at the Hong Kong University of Science and Technology.  Xia's research is supported by the Research Grants Council ECS grant 26305120. }}

\author{
    Ziqing Guo 
    \\     Department of ISOM,  HKUST, Hong Kong \\
    Yang Liu
      \\ 
Institute of Statistics and Big Data, Renmin University of China,  China\\
     Lucy Xia\thanks{Corresponding author (Email: lucyxia@ust.hk)}
     \\    Department of ISOM, HKUST, Hong Kong \\
    }

% \author{     Lucy Xia\thanks{Corresponding author (Email: lucyxia@ust.hk)}
%      \\Department of ISOM, HKUST, Hong Kong \\
%     Yang Liu
%       \\ 
% Institute of Statistics and Big Data, Renmin University of China,  China\\
% Ziqing Guo 
%     \\     Department of ISOM,  HKUST, Hong Kong \\
%     }    
\date{}
\maketitle
\setcounter{footnote}{1}

\begin{abstract}
Balancing influential covariates is crucial for valid treatment comparisons in clinical studies. While covariate-adaptive randomization is commonly used to achieve balance, its performance can be inadequate when the number of baseline covariates is large. It is therefore essential to identify the influential factors associated with the outcome and ensure balance among these critical covariates.  In this article, we propose a novel adaptive randomization approach that integrates the patients' responses and covariates information to select sequentially significant covariates and maintain their balance. We establish theoretically the consistency of our covariate selection method and demonstrate that the improved covariate balancing, as evidenced by a faster convergence rate of the imbalance measure, leads to higher efficiency in estimating treatment effects. Furthermore, we provide extensive numerical and empirical studies to illustrate the benefits of our proposed method across various settings.


\end{abstract}

\section{Introduction}
\noindent

In clinical trials, ensuring the comparability of treatment groups with respect to crucial baseline covariates is critical for maintaining the credibility of the results. Covariate-adaptive randomization (CAR) methods \citep{rosenberger2008handling,hu2014adaptive} are extensively used to balance these influential covariates. Typically, permuted block randomization or biased coin design can be implemented within the strata generated from baseline covariates \citep{zelen1974randomization,baldi2011covariate} to balance these baseline factors. However, this approach may result in inadequate balance when the number of strata is relatively large compared to the sample size \citep{toorawa2009use}. To address this issue, Minimization \citep{taves1974minimization}, and its randomized counterpart, the \cite{pocock1975sequential} procedure are considered to ensure the balance of the covariates’ margins. \cite{hu2012asymptotic} and \cite{hu2023multi} further extend these methods to ensure multi-dimensional covariate balancing simultaneously. Compared to other adaptive design methods, CAR procedures generally do not affect the consistency of the treatment effect estimate and do not lead to increased Type I error rates when analyzed with
appropriate primary analysis method such as the permutation test \citep{FDA2019}.  Review articles such as \cite{taves2010use}, \cite{kahan2012reporting}, \cite{lin2015pursuit}, and \cite{ciolino2019ideal} have reported the prevalent applications of these methods in leading medical journals. For a detailed exploration of the development of CAR procedures, we direct readers to foundational works by \cite{taves1974minimization}, \cite{pocock1975sequential}, \cite{baldi2011covariate}, \cite{hu2012asymptotic}, \cite{hu2023multi}, and \cite{ma2022new}.

Despite the widespread applications of CAR procedures, their finite-sample performance may be inadequate when dealing with a large set of baseline covariates. This issue can be particularly severe when prior knowledge for selecting relevant covariates is insufficient and the sample size is not proportionate to or even less than the number of covariates.  
 For example, in Phase III trials of newly developed treatments for emerging diseases, there is often limited information available about the treatment and the disease itself. In these situations, although a large pool of candidate covariates may be collected, there may be insufficient understanding of which covariates should be incorporated into the trial design. Consequently, selecting the most relevant factors becomes essential for effectively designing the trial.
  Typically, CAR procedures seek to sequentially minimize a weighted imbalance score that incorporates various covariates \citep{hu2012asymptotic,ma2022new,hu2023multi}. The over-inclusion of covariates in CAR might lead to the under-weighting of those truly influential covariates, thereby reducing the procedures' effectiveness on the key covariates. Furthermore, when dealing with continuous covariates, Mahalanobis distance has been adopted commonly as an imbalance measure \citep{qin2022adaptive,morgan2012rerandomization}. On the other hand, optimal designs can also be used in clinical trials \citep{atkinson1982optimum, atkinson2002comparison} to balance continuous covariates and achieve higher efficiency for estimating treatment effects. When the dimensionality is larger than the sample size, all such methods will require the estimation of the precision matrix and thus are either not implementable or highly unstable.  As a result, accurately identifying and selecting influential covariates is vital for ensuring the effectiveness of covariate balancing in CAR implementation. This becomes particularly important in practical scenarios, such as the development of drugs for new and urgent diseases, where practitioners lack prior knowledge in choosing appropriate baseline factors.  In turn, there is a pressing need to develop procedures that integrate covariate selection with covariate balancing within adaptive randomization for phase III trials.

The randomization procedure utilizing both patients' covariates and response information falls under the framework of covariate-adjusted response-adaptive randomization (CARA) \citep{atkinson2013randomised,hu2006theory,rosenberger2015randomization}. This class of procedures is proposed to allocate a patient to the beneficial treatment arm with a higher probability based on the previous patients' covariates, treatment assignments, response profiles, and the current patient's covariates
\citep{rosenberger2001covariate}. \cite{zhang2007asymptotic} provides a general framework for CARA and derives the asymptotic properties of the treatment effect estimate and the allocation ratios.  Works including, but not limited to \cite{sverdlov2013utility}, \cite{huang2013longitudinal}, \cite{cheung2014covariate}, \cite{biswas2016class}, and~\cite{mukherjee2023covariate,mukherjee2024covariate}, have further demonstrated the use of CARA for survival outcomes, generalized linear model, and longitudinal data.
\cite{hu2015unified} and \cite{Baldi2012CARA} consider extending this class of CARA to address concerns such as cost, estimation efficiency, ethics, and other related factors.  In a recent work, \cite{ZhuCARA2023} proposes incorporating semi-parametric estimation techniques within CARA designs to address the concern of model misspecification. In another line of research,  \cite{villar2018covariate} introduces a new class of CARA methods based on the forward-looking Gittins index rule to achieve the goal of non-myopic response-adaptive randomization.

Nevertheless, very few works in the adaptive randomization literature tackle the simultaneous selection and balancing of influential covariates. For discrete covariates, \cite{zhang2022covariate} proposes an adaptive randomization that sequentially selects influential covariates via group Lasso and then balances them utilizing \cite{hu2012asymptotic}'s procedure. When dealing with continuous covariates, one can discretize them and apply \cite{zhang2022covariate}, but this approach may result in efficiency loss \citep{LiContinuousCAR2019,ma2020statistical}. In practice, while the CAR procedures for discrete covariates are the most widely used \citep{zelen1974randomization,pocock1975sequential,hu2012asymptotic,baldi2011covariate}, there has been recent development of designs that can directly balance continuous covariates \citep{hu2012balancing,ma2013balancing,zhou2018sequential,qin2022adaptive}. Incorporating covariate selection \citep{fan2009ultrahigh,shah2013variable,muller2024isotonic,reeve2023optimal} into these existing methods can effectively address the issue of over-incorporation and improve the estimation efficiency of treatment effects. Notably, \cite{ma2022new} introduces a novel framework that balances not only individual continuous covariates but also functions of them, showcasing its potential for enhancing estimation efficiency. The number of functions can be significantly larger than the original number of covariates; consequently, there is an urgent need to adopt simultaneous covariate selection when applying methods like \cite{ma2022new}.

In this article, we propose an adaptive randomization procedure, named {\sf ARCS} (Adaptive Randomization with Covariate Selection), that aims to achieve the dual goals of covariate selection and covariate balancing. {\sf ARCS} is an iterative procedure by assigning the recruited patients in batches. At each iteration, we first select important covariates by employing either simple $l_1$ regularization techniques like Lasso or incorporating covariate selection using an additive outcome model. In the second step, we minimize the imbalance measures associated with various functional forms for the selected covariates, as proposed in \cite{ma2022new}. This ensures that the {\sf ARCS} procedure not only selects the important covariates, but also adapts to newly developed CAR procedures with flexible covariate balancing. This adaptive randomization procedure enables a significant reduction in the dimension of covariates through the selection step, greatly enhances balancing efficiency, and improves the accuracy of inference results. In particular, we consider two special cases of \cite{ma2022new} as the choices of imbalance measure: 1) Mahalanobis distance, and we denote the corresponding procedure as {\sf ARCS-M}; 2) the weighted sum of differences-in-covariates means and covariances, and we denote the corresponding procedure as {\sf ARCS-COV}. Extensive numerical studies are conducted to comprehensively evaluate the superior performance of the proposed {\sf ARCS} procedures.

Our theoretical results demonstrate the advantages of {\sf ARCS} by simultaneously selecting and balancing covariates. As we allow for an increasing number of covariates, our analyses reveal that the rate of imbalance depends on the covariate dimension. Specifically, the performance of traditional CAR procedures without covariate selection suffers from a slow convergence rate of imbalance when the dimension of covariates is large. Furthermore, if significant covariates associated with the response are not adequately balanced, the efficiency gains from using CAR in estimating treatment effects may be compromised. On the other hand, the appropriate selection of influential covariates significantly mitigates the impact of covariate dimension on the performance of CAR, particularly when the number of influential covariates associated with the response is finite. To this end, we derive the consistency of covariate selection and demonstrate the improved convergence rates of imbalance for selected covariates under {\sf ARCS}. Furthermore, we derive the asymptotic distribution of difference-in-means estimate of the treatment effect under the {\sf ARCS-M} or {\sf ARCS-COV} based on a  linear outcome model.  We show that the efficiency gain achieved through {\sf ARCS-M} or {\sf ARCS-COV} is substantial, as the difference-in-means estimate attains optimal efficiency for treatment effect estimation.   Note that the asymptotic properties of the treatment effect estimate are often of great importance, as they can be used to construct subsequent statistical inference methods. \cite{FDA2023} has indicated that ignoring the covariates used in CAR may lead to overestimated standard errors, resulting in undesirable properties such as reduced Type I error rates. Consequently, these findings underscore the necessity of achieving covariate selection and balance through {\sf ARCS}, evaluating the property of CAR in high-dimensional settings, and provide the basis for developing subsequent valid statistical inference methods.

%A key contribution of this work is that we provide a detailed analysis of the balance property of CAR procedure in the high-dimensional setting, which emphasize 

The rest of this article is organized as follows. In Section~\ref{sec:framework}, we introduce the essential notations and lay out the details of {\sf ARCS}. In Section~\ref{sec:theo}, we present the mild assumptions and the theoretical properties of {\sf ARCS}. Extensive simulation studies and a well-calibrated clinical trial example are provided in Sections~\ref{Section: Numerical Studies} and~\ref{sec:real data} respectively. We conclude the article with a summary and short discussion in Section~\ref{sec:conclusion}. The technical proofs and additional numerical results are relegated to the Appendix.

\subsection{Notations}

For a set $\mathcal{A}$, let $|\mathcal{A}|$ be the cardinality of $\mathcal{A}$, and $\mathcal{A}^c$ be the complement of $\mathcal{A}$.  For a vector $\beta = (\beta_1,\dots,\beta_p)^\top$, denote the set of indices corresponding to the nonzero elements of $\beta$ as $J(\beta) := \{j\in\{1,\dots,p\}:\beta_j \neq 0\}$. Further, let $\lVert \beta \rVert_1$ and $\lVert \beta \rVert_2$ represent the $L_1$ norm and $L_2$ norm of $\beta$, respectively. In particular, $\lVert \beta \rVert_1 = \sum_{i=1}^p |\beta_i|$, and $\lVert \beta \rVert_2 = (\sum_{i=1}^p \beta_i^2)^{1/2}$. For any $q\times p$ -dimensional matrix $\mathbb{M}\in\mathbb{R}^{q\times p}$, with $m_{ij}$ denoting its element in the $i$-th row and $j$-th column, we denote its Frobenius norm as $\lVert \mathbb{M} \rVert_F = (\sum_{i=1}^q \sum_{j=1}^p m_{ij}^2)^{1/2}$, and $\lVert \mathbb{M} \rVert_{1,2} = \max_{1\leq j \leq p} \sqrt{\sum_{i=1}^q m_{ij}^2}$. Lastly, $\stackrel{d}{\to}$ means converge in distribution. 

\section{Framework}\label{sec:framework}

Consider an experiment with two treatments and $n$ patients.  Let $T_i$ denote the treatment assignment for the $i$-th patient, where $T_i = 1$ if the patient is assigned to the treatment group and $T_i = 0$ if the patient is assigned to the control group.  Then, the numbers of patients in the treatment and control groups are  $n_1 = \sum_{i=1}^n T_i$ and $n_0 = \sum_{i=1}^n (1 - T_i)$, respectively.  Let $X_i = (x_{i1},\dots,x_{ip})^\top \in\mathbb{R}^p$ be a $p$-dimensional vector representing the baseline  covariates of the $i$-th patient and $\mathbb{X} = (X_1, \dots, X_n)^\top \in \mathbb{R}^{n\times p}$ be an $n\times p$-dimensional matrix. Moreover, let $\mathbb{X}(0)$ be the submatrix of $\mathbb{X}$ containing covariates for all patients in the control group and $\mathbb{X}(1)$ is defined similarly. Furthermore, we define $X_{i,J}$ as the subvector of $X_{i}$ constrained to set $J$, i.e., the vector containing all $\{x_{ik}: k\in{J}\}$.

After the assignment of the $i$-th unit, we assume the observed outcome satisfies the following linear model:
\begin{equation}
Y_i = \mu(1) T_i + \mu(0) (1-T_i) + X_i^\top\beta^* + \varepsilon_i,\label{eq:outcome}
\end{equation}
where $\mu(1)$, $\mu(0)$ are the main effects for the treatment group and control group, and $\tau = \mu(1)- \mu(0)$ represents the treatment effect. Let $\beta^* = (\beta^*_1,\dots,\beta^*_p)^\top \in\mathbb{R}^p$, $J^* = J(\beta^*)$ denote the set of indices corresponding to the non-zero coordinates in $\beta^*$ and $s = |J(\beta^*)|$ represent the sparsity. Furthermore,  $\{\varepsilon_i\}_{i=1}^n$ are independent and identically distributed random errors with mean zero and variance $\sigma^2_\varepsilon$.

%Incorporating unimportant covariates in CAR will lead to less weight on balancing important covariates, thus resulting in worse precision of the estimated treatment effect. To address this issue, we propose the {\sf ARCS} framework that selects influential covariates while simultaneously balancing them during the randomization. {\sf ARCS} splits the data into batches with a batch size of $N$. Within each batch, we balance the covariates selected in the previous batch using a pre-specified CAR procedure. Once all units in the batch are assigned, {\sf ARCS} selects influential covariates through a variable selection method based on all the available data. Additionally, we introduce an initial batch with a batch size of $N_0$ to initialize the set of important covariates. For simplicity, let's assume that $n - N_0$ is a multiple of $N$.

%``As an example for balancing continuous and nonlinear covariates, we apply the {\sf ARCS} framework to balance the general covariates $\phi(X_i)$ proposed in \cite{ma2022new}, where $\phi(X_i):\mathbb{R}^p\to\mathbb{R}^q$ is a functional with $q\ge p$.  Here, $\phi(X_i)$ usually contains much more covariates than $X_i$, thus the variable selection step of {\sf ARCS} is more crucial.''

To incorporate the concern of covariates balancing from multiple perspectives, we adopt  $\phi(X_i)$ proposed in \cite{ma2022new} as a $q$-dimensional function of covariate
for the construction of the imbalance measure. Here, $\phi(X_i):\mathbb{R}^p\to\mathbb{R}^q$ is a function where $q \geq p$, and is used to incorporate additional functionals beyond individual covariates of $X_i$ to achieve finer covariate balance. The corresponding imbalance measure is  
$$
\text{Imb}_{n}^\phi =\lVert \Lambda_{n}
 \rVert^2_2=\left\lVert  \sum_{i=1}^n (2T_i-1)\phi(X_i)\right\rVert^2_2,
$$
 where $\Lambda_{n} = \sum_{i=1}^n (2T_i-1) \phi(X_{i})$. Given that $q\ge p$, $\phi(X_i)$ may have a much higher dimension than $X_i$. This could potentially lead to down-weighting of the influential covariates, which highlights the need for a design that incorporates covariate selection. The following {\sf ARCS} is proposed to address this concern.  Our procedure assigns the patients in batches, with $b$ denoting the batch number.  We use the first $N_0$ patients for the initial stage and proceed with our sequential randomization procedure for the remaining  $(n-N_0)/N$  batches of patients, with $N$ representing the batch size so that $(n-N_0)/N$ is an integer. When $N = 1$, this goes back to the traditional individual assignment.

\begin{itemize}
\item[(1)] Given an even integer $N_0$ (size of the initial batch $b=0$), for  $1\le i\le N_0/2$,  assign $(1,0)$ or $ (0,1)$ to  $(T_{2i-1}, T_{2i})$   with probability 1/2. After observing $\lbrace Y_i, X_i, T_i \rbrace_{i=1}^{N_0}$, apply  Lasso to obtain $\hat{\mu}^{(0)}(a)$, $\hat{\beta}^{(0)}(a)$ for $a \in \{0,1\}$ from the following optimization problem:
 $$
(\hat{\mu}^{(0)}(a), (\hat{\beta}^{(0)}(a))^\top)^{\top} = \mathop{\arg\min}\limits_{(\mu, \beta)} \left\{\frac{1}{n_a^{(0)}} \sum_{1\leq i\leq N_0: T_i=a} (Y_i - \mu(a) - X_i^\top\beta)^2 + \lambda_a \sum_{j=1}^p |\beta_j| \right\}\,, 
$$
where $\beta = (\beta_1,\dots,\beta_p)^\top$, $\hat{\beta}^{(0)}(a) = (\hat{\beta}_{1}^{(0)}(a),\dots,\hat{\beta}_{p}^{(0)}(a))^\top$, $n_0^{(0)} = N_0 -\sum_{i=1}^{N_0} T_i$, $n_1^{(0)} = \sum_{i=1}^{N_0} T_i$, and $\lambda_0$ and $\lambda_1$ are tuning parameters. Let $\widehat{J}^{(0)}$ denote the set of selected covariates from the initial stage, defined as the intersection of the supports of $\hat{\beta}^{(0)}(0)$ and $\hat{\beta}^{(0)}(1)$ : $\widehat{J}^{(0)} = J(\hat{\beta}^{(0)}(0)) \cap J(\hat{\beta}^{(0)}(1)).$

\item[(2)] For $1 \leq b \leq (n-N_0)/N$, repeat steps (3)-(4) given below until all $(n-N_0)/N $ batches of patients are assigned.

\item[(3)] Suppose the first  $b-1$ batches of subjects (with a total of $N_0 + (b-1)N$ patients) have already been assigned,   repeat the following until the $N$ patients in the $b$-th batch is assigned. For $1\le j \le N$,  suppose the first  $j-1$ patients in the $b$-th batch have been assigned and the $j$-th patient is waiting for the assignment. In other words, it is the $i= (N_0+ (b-1)N+j)$-th patient in the whole sample about to be assigned. Define the imbalance measure for the selected covariates evaluated from the first $i$ patients  by
$$
\text{Imb}^{\phi}_{i, \widehat{J}^{(b-1)}} =\lVert \Lambda_{i,\widehat{J}^{(b-1)}}
 \rVert^2_2= \left\lVert \sum_{k=1}^{i} (2T_k-1) \phi(X_{k, \widehat{J}^{(b-1)}}) \right\rVert_2^2,
$$
with the last assignment $T_i$ unknown and to be determined.  More specifically, Let   $\text{Imb}^{\phi}_{i, \widehat{J}^{(b-1)}}(a)$ denote the imbalance measure $\text{Imb}^{\phi}_{i, \widehat{J}^{(b-1)}}$ if $T_i = a$, for $a \in \{0, 1\}$. Assign the $i$-th patient with probability

$$
\p \left(T_i = 1 \mid \left\lbrace X_j\right\rbrace_{j=1}^i, \left\lbrace T_j\right\rbrace_{j=1}^{i-1}\right) = \left\{\begin{array}{ll}
\rho, & \text{Imb}^{\phi}_{i, \widehat{J}^{(b-1)}}(1) < \text{Imb}^{\phi}_{i, \widehat{J}^{(b-1)}}(0), \\
1-\rho, & \text{Imb}^{\phi}_{i, \widehat{J}^{(b-1)}}(1) > \text{Imb}^{\phi}_{i, \widehat{J}^{(b-1)}}(0), \\
0.5, & \text{Imb}^{\phi}_{i, \widehat{J}^{(b-1)}}(1) = \text{Imb}^{\phi}_{i, \widehat{J}^{(b-1)}}(0),
\end{array}\right.
$$
where $\rho\in (0.5,1)$  is the probability of a biased coin \citep{efron1971forcing}.   Intuitively, with $\rho>0.5$, we assign the $i$-th patient to the under-represented treatment arm defined in terms of covariates, leading to a smaller imbalance measure.

\item[(4)] After the first $b$ batches with $N_0+ bN$ patients have been assigned,  perform covariate selection with Lasso for $a \in \{0,1\}$ as:
$$
(\hat{\mu}^{(b)}(a), (\hat{\beta}^{(b)}(a))^\top) = \mathop{\arg\min}\limits_{(\mu, \beta)} \left\{\frac{1}{n_a^{(b)}} \sum_{1\leq i\leq N_0+bN: T_i=a} (Y_i - \mu - X_i^\top\beta)^2 + \lambda_a \sum_{j=1}^p |\beta_j| \right\}\,,
$$
where $\hat{\beta}^{(b)}(a) = (\hat{\beta}^{(b)}_1(a),\dots,\hat{\beta}^{(b)}_p(a))^\top$, $n_0^{(b)} = N_0+bN -\sum_{i=1}^{N_0+bN} T_i$, and $n_1^{(b)} = \sum_{i=1}^{N_0+bN} T_i$. Then update the set of selected covariates $\widehat{J}^{(b)}$ to be the intersection of the supports of $\hat{\beta}^{(b)}(0)$ and $\hat{\beta}^{(b)}(1)$:  $\widehat{J}^{(b)} = J(\hat{\beta}^{(b)}(0)) \cap J(\hat{\beta}^{(b)}(1))$.

\end{itemize}

\begin{myRem}

In general, the values of $N_0$ and $N$ may not affect the asymptotic properties of {\sf ARCS}. For the implementation, we may choose a relatively large number of $N_0$ to ensure the initial estimation of $\widehat{J}^{(0)}$, e.g., $N_0 \in \lbrace 30, 40,50\rbrace$. Furthermore, the value of the batch size improves the flexibility of {\sf ARCS}.    When $N=1$, the corresponding procedure is fully adaptive and may exhibit finer finite-sample performance, though at a higher computational cost. For numerical evidence, see Section~\ref{Section: Numerical Studies} for details.

Steps (3) and (4) are the key components of {\sf ARCS}. Step (3) ensures the covariate balancing through sequential adaptive randomization.  It assigns the $i$-th patient to the treatment arm leading to a smaller imbalance measure, with a higher probability $(\rho>0.5)$. Typically, a large value of $\rho$, such as 0.85 and 0.9,  is suggested for practical implementations \citep{hu2012asymptotic}. Furthermore, $\phi(\cdot)$  in the imbalance measure can address different concerns regarding covariate balancing and improve the estimation of treatment effect \citep{ma2022new}. Typically, the first element of $\phi(\cdot)$ is chosen as a positive constant, so that the difference of the numbers of patients in the two treatment arms are considered for balancing. We demonstrate the usage of $\phi(\cdot)$ with two special cases, as given below this remark.

Step (4) selects important covariates by sequentially applying regularization methods such as Lasso. Combining steps (3) and (4), the main contribution of the {\sf ARCS} procedures is that it works with the imbalance measure $\text{Imb}_{i, \widehat{J}}^\phi$, which is defined only on the sequentially selected set of covariates $\widehat{J}$, rather than utilizing $\text{Imb}_{i}^\phi$ defined on all the $p$ covariates. Through iteratively selecting and balancing the influential covariates, theoretically we guarantee that $\widehat{J}$ is a consistent estimate of the true relevant set of covariates $J^*$ (see Theorem \ref{prop:consistency}). Thus, the {\sf ARCS} procedure asymptotically and adaptively solves the issue of over-incorporation that may arise when using all the $p$ covariates, as well as the problem of under-incorporation that could occur if we artificially decide on a few covariates to use in the procedure, as this may cause us to miss important covariates.

\end{myRem}

While {\sf ARCS} can adapt to various choices of  $\phi(\cdot)$, we illustrate its theoretical properties and empirical performance through two common examples from the literature.

\begin{case}[Balancing the Mahalanobis distance]\label{ex:M}
Suppose $\Sigma_{J^*} = \cov (X_{k,J^*})$ is known and the eigen-decomposition of the precision matrix is $\Sigma_{J^*}^{-1} = \Gamma D \Gamma^\top$. Further, let us assume $\phi(X_{k,J^*}) = (\sqrt{w_0},\sqrt{w_1} X_{k,J^*}^\top \Gamma D^{1/2})^\top$, then the corresponding imbalance measure for the first $i$ patients is 
\begin{align*}
\text{Imb}^{\phi}_{i,J^*} &= w_0(n_1-n_0)^2 \cr 
&+ w_1 \left( \sum_{1\leq k\leq i:T_k=1} X_{k,J^*} - \sum_{1\leq k\leq i:T_k=0} X_{k,J^*} \right)^\top \Sigma_{J^*}^{-1} \left( \sum_{1\leq k\leq i:T_k=1} X_{k,J^*} - \sum_{1\leq k\leq i:T_k=0} X_{k,J^*} \right),
\end{align*}
where $w_0$, $w_1$ are positive constants and $w_0+w_1=1$. The first term in $\text{Imb}^{\phi}_{i,J^*}$ retains the magnitude of the overall difference $n_1- n_0$, and thus the second term in $\text{Imb}^{\phi}_{i,J^*}$ is asymptotically proportional to the Mahalanobis distance. We refer to {\sf ARCS} with the choice of $\phi(X_{k,J^*})$ in this special case as {\sf ARCS-M}\footnote{Since $\text{Imb}^{\phi}_{i,J^*}$ is proportional to the Mahalanobis distance and $\Sigma_{J^*}$ is unknown in practice, we implement {\sf ARCS-M} by replacing step (3) of {\sf ARCS} by the {\sf ARM} procedure in \citep{qin2022adaptive}. For details, please refer to Section \ref{append:ARCS} in the Appendix.}.

\end{case}

 Mahalanobis distance is often considered in the literature for several reasons~\citep{qin2022adaptive,morgan2012rerandomization}. It is an affinely invariant measure, so the scaling of the covariates will not affect the value of the imbalance. Further, minimizing the Mahalanobis distance ensures the balance for each component in $X_{i,J^*}$. Additionally, improving balance in terms of the Mahalanobis distance usually leads to efficiency gains in estimating the treatment effect ~\citep{qin2022adaptive,morgan2012rerandomization}.  {\sf ARCS-M}  is thus considered as an example of $\phi(X_{k,J^*})$ for achieving balancing through Mahalanobis distance.

\begin{case}[Balancing the mean and the covariance]\label{ex:cov}
Let $\phi(X_{k,J^*}) = (\sqrt{w_0},  \sqrt{w_1} X_{k,J^*}^\top,\\ \sqrt{w_2} \text{vec} (X_{k,J^*}X_{k,J^*}^\top)^\top)^\top$, then the corresponding imbalance measure for the first $i$ patients is 
\begin{align*}
\text{Imb}_{i,J^*}^\phi &= w_0(n_1 - n_0)^2 + w_1 \left\lVert \sum_{1\leq k\leq i:T_k=1} X_{k,J^*} - \sum_{1\leq k\leq i:T_k=0} X_{k,J^*}\right\rVert_2^2 \\
& + w_2 \left\lVert \sum_{1\leq k\leq i:T_k=1} X_{k,J^*} X_{k,J^*}^\top - \sum_{1\leq k\leq i:T_k=0} X_{k,J^*} X_{k,J^*}^\top\right\rVert_F^2,
\end{align*}
where $w_0$, $w_1$ and $w_2$ are nonnegative constants with $w_0 + w_1 + w_2 = 1$. We refer to {\sf ARCS} with the choice of $\phi(X_{k,J^*})$ in this special case as {\sf ARCS-COV}.
\end{case}

{\sf ARCS-COV} measures imbalance through the means and covariances of the covariates. Similar to the Mahalanobis distance, including covariate means can also enhance the efficiency of estimating treatment effects \citep{LiContinuousCAR2019,ma2022new}. Practically, $Y$ may be influenced by higher-order covariate effects, and $\text{vec} (X_{k,J^*}X_{k,J^*}^\top)$ is incorporated into $\phi(X_{k,J^*})$ to account for quadratic effects. When more complex relationships between $X_{i,J^*}$ and $Y_i$ are of interest, more intricate forms of $\phi(X_{i,J^*})$, such as the kernelized imbalance measure discussed in \cite{ma2022new}, may be considered. In cases where additive relationships are assumed for $\lbrace X_{i,J^*}, Y_i \rbrace$, we illustrate an extension of {\sf ARCS} in the following remark.

\begin{myRem}\label{remark2}
An extension of the {\sf ARCS} framework is to select covariates through methods designed for sparse additive models \citep{ravikumar2009sparse}. This extension is especially beneficial if we consider a nonlinear outcome model as follows, where the performance of Lasso may not be desirable:
\begin{equation}\label{eq:additive}
Y_i = \mu(1) T_i +\mu(0) (1-T_i) + \sum_{j=1}^p \beta_j^* g_j(x_{ij}) + \varepsilon_i,
\end{equation}
where $\{g_j\}_{j=1}^p$ are component functions. Correspondingly, we modify step (4) in {\sf ARCS} to (4') as follows.
\end{myRem}

\begin{itemize}
\item[(4')] After the first $b$ batches with $N_0+ bN$ patients have been assigned, we perform covariate selection via a sparse additive model such as \cite{ravikumar2009sparse} for $a \in \{0,1\}$ as:
$$
(\hat{\mu}^{(b)}(a), (\hat{\beta}^{(b)}(a))^\top) = \mathop{\arg\min}\limits_{(\mu, \beta)} \left\{\frac{1}{n_a^{(b)}} \sum_{1\leq i\leq N_0+bN: T_i=a} (Y_i - \mu - \sum_{j=1}^p \beta_j g_j(x_{ij}))^2 + \lambda_a \sum_{j=1}^p |\beta_j| \right\}\,.
$$
where $\hat{\beta}^{(b)}(a) = (\hat{\beta}^{(b)}_1(a),\dots,\hat{\beta}^{(b)}_p(a))^\top$, $n_0^{(b)} = N_0+bN -\sum_{i=1}^{N_0+bN} T_i$, and $n_1^{(b)} = \sum_{i=1}^{N_0+bN} T_i$. Then update the set of selected covariates $\widehat{J}^{(b)}$ to be the intersection of the supports of $\hat{\beta}^{(b)}(0)$ and $\hat{\beta}^{(b)}(1)$:  $\widehat{J}^{(b)} = J(\hat{\beta}^{(b)}(0)) \cap J(\hat{\beta}^{(b)}(1))$. 
\end{itemize}

We name the procedures with step (4') {\sf ARCS-M-add} and {\sf ARCS-COV-add}. Thus far, we have finished constructing the {\sf ARCS} framework illustrated with four specific methods: {\sf ARCS-M}, {\sf ARCS-COV}, {\sf ARCS-M-add} and {\sf ARCS-COV-add}. We summarize the {\sf ARCS} procedure in Algorithm \ref{alg:ARCS} in the Appendix.

\section{Theoretical results}
\label{sec:theo}

\subsection{Properties of {\sf ARCS}}

In this section, we evaluate the properties of {\sf ARCS} (with $\phi(\cdot)$ in the two special cases) for covariate selection and covariate balancing. Recall $\Lambda_{n,J} = \sum_{i=1}^n (2T_i-1) \phi(X_{i,J})$ is the vector of imbalance measure evaluated for the selected set of covariates $J$. For $i>N_0$, let $b_i = \lfloor (i-N_0)/N \rfloor$ be the batch index before the batch in which patient $i$ is located. Because the following results are valid for any batch $b=0,1,\dots,(n-N_0)/N$, we write $\hat{\beta}^{(b)}(a)$ in short as $\hat{\beta}(a)$ and $n_a^{(b)}$ as $n_a$ for simplicity, when there is no ambiguity.

% \begin{defin}[Small-ball condition over $(B_0(s'),u,v)$ \citep{mendelson2015learning}]
% $X$ is said to satisfy the small-ball condition over $(B_0(s'),u,v)$ if there exists positive numbers $u$ and $v$ such that
% $$
% \p( |\delta^\top X| \geq u \lVert \delta\rVert_2 )\geq v\,,\quad \forall \delta\in B_0(s'),
% $$
% where $B_0(s') = \{\delta \in \R^p: \lVert \delta\rVert_0 \leq s'\}$.
% \end{defin}

\begin{defin}[Restricted Eigenvalue (RE) condition \citep{bickel2009simultaneous}]\label{def:RE}
We say matrix $A\in \mathbb{R}^{q\times p}$ satisfies $\text{RE} (\tilde{s}, \tilde{c})$ condition with $\tilde{c}>0$ and $\tilde{s}$ an integer such that $1\leq \tilde{s} \leq p$, if the following condition holds:
$$
\kappa(\tilde{s},\tilde{c}) \coloneq \min_{\substack{J\subset\{1,\dots,p\}, \\ |J|\leq \tilde{s}}} \min_{\substack{\delta\neq 0, \\ \lVert \delta_{J^{\sf c}}\rVert_1 \leq \tilde{c}\lVert \delta_{J}\rVert_1}} \frac{\lVert A \delta\rVert_2}{\sqrt{q}\lVert \delta_{J}\rVert_2} > 0 \,.
$$
\end{defin}

\begin{defin}[Isotropic $\psi_2$ random vector with constant $\iota$ \citep{rudelson2012reconstruction}]\label{def:psi}
A random vector $u\in\mathbb{R}^p$ is called isotropic $\psi_2$ with 
 constant $\iota$ if for every $v\in\mathbb{R}^p$, $\E|u^\top v|^2 = \lVert v\rVert_2^2$, and 
$$
\lVert u^\top v \rVert_{\psi_2} \coloneq \inf\{t: \E(\exp((u^\top v)^2/t^2)) \leq 2\} \leq \iota \lVert v \rVert_2 . 
$$ 
\end{defin}

\begin{assump}
Under the outcome model (\ref{eq:outcome}),  the following conditions hold. 
\begin{enumerate}[label=(\roman*)]\label{assum}
\item\label{assum:subgaussian} 
For $a=\{0,1\}$, $\mathbb{X}(a)$ can be represented as $\mathbb{X}(a) =  \Psi_a A_a$, for some $\Psi_a \in \mathbb{R}^{n_a\times p}$, $A_a \in\mathbb{R}^{p\times p}$, where the rows of $\Psi_a$ are all isotropic $\psi_2$ random vectors with constant $\iota$, and $p^{1/2}A_a$ satisfies $\text{RE}(s, 9)$ condition. Additionally, $\lVert A_a \rVert_{1,2}$ is upper bounded by a constant.

\item\label{assum:noise} $\varepsilon_1,\dots,\varepsilon_n$ are i.i.d. copies of random variable $\varepsilon$ with mean 0 and variance $\sigma_\varepsilon^2$. Furthermore, $\lVert \varepsilon \rVert_{\psi_2}\leq L\sigma_{\varepsilon}$ holds for some positive constant $L$. 
% \item\label{assum:minimal_signal}
% \begin{equation}
% \min_{j \in J^*} |\beta_j^*| \geq \frac{256\sqrt{2}}{u^2 v}\sigma_\varepsilon s\max_{a=0,1}\left\{ \sqrt{\frac{(1+\alpha)\log n_a + \log p}{n_a}} \right\}\,.
% \end{equation}
\item\label{assum:minimal_signal} Let $\kappa(s,3)$ be the value associated with $p^{1/2}A_a$ under $\text{RE}(s,3)$. We will assume $\min_{j \in J^*} |\beta_j^*| \geq l_\beta$, where
$$ 
l_\beta \geq  \frac{8c \sigma_\varepsilon s^{1/2}}{(1-r_a)^2 \kappa^2(s,3)}\max\left\lbrace \sqrt{ \frac{(1+ \alpha) \log n_a + \log 2p }{\xi_a n_a}}, 
\frac{(1+ \alpha) \log n_a + \log 2p }{\xi_a n_a} \right\rbrace ,   \ \  a \in \lbrace 0,1 \rbrace
$$
for  some $r_0, r_1\in (0,1)$, and positive constants $c$, $\xi_0$, $\xi_1$ and $\alpha$.

\item\label{assum:moment} $\E (\lVert \phi(X_{i,J^*})\rVert_2^\gamma) = O(s_\phi^{\gamma/2})$ for some $\gamma > 2$, where $s_\phi$ is the number of dimensions of vector $\phi(X_{i,J^*})$. The nonzero eigenvalues of $\E(\phi(X_{i,J^*})\phi(X_{i,J^*})^\top)$ are bounded below by a positive constant.
\end{enumerate}
\end{assump}

% irrepresentable condition
% \begin{assumption}\label{assum:2}
% There exists some $\eta>0$ such that 
% $$
% \lVert  \Sigma_{\s^c \s} \Sigma^{-1}_{\s\s} \rVert_\infty \leq 1-\gamma\,,
% $$
% where $\s^c = \{1,\dots,p\}/\s$, and for any $n\times m $ matrix $M$, $\lVert M\rVert_\infty := \max_{1\leq i\leq n} \sum_{j=1}^m |M_{ij}|$. There exists constants $C_{\min},C_{\max}>0$ such that
% $$
% \lambda_{\min} (\Sigma_{\s\s}) \geq C_{\min}\,.
% $$
% where $\lambda_{\min}(M)$ is the minimal eigenvalue of matrix $M$.
% \end{assumption}

\begin{myRem}
In Assumption \ref{assum}, \ref{assum:subgaussian}- \ref{assum:minimal_signal} are commonly assumed in the Lasso literature  \citep{bickel2009simultaneous,rudelson2012reconstruction}. Assumption \ref{assum:noise} roughly assumes that $\varepsilon$ follows sub-Gaussian distributions.

Assumption \ref{assum:moment} is similar to Assumptions 2 and 3 in \cite{ma2022new}, but with two key differences. First, we consider $\phi(X_{i,J^*})$ defined on an important subset $J^*$, while \cite{ma2022new} considers $\phi(X)$ defined on all covariates. In high-dimensional settings where the number of covariates $p$ can greatly exceed the sample size $n$, the dimension of $\phi(X)$ is on the order of $p$ or powers of $p$. Directly assuming $\E (\lVert \phi(X)\rVert_2^\gamma)$ to be finite, the nonzero eigenvalue of $\E(\phi(X)\phi(X)^\top)$ to be lower-bounded by a positive constant, and balancing all dimensions of $\phi(X)$ may not be reasonable. In contrast, by focusing on the important subset $J^*$, we effectively reduce the dimension of $\phi(X_{i,J^*})$ to $s_\phi$, which is at the same order of $s$ or powers of $s$, making the corresponding assumptions more reasonable. Second, instead of assuming a constant moment condition, we introduce $s_\phi$ that can slowly diverge depending on the choice of $\gamma$, allowing for greater flexibility. We discuss the implication of $s_\phi$ in Theorem \ref{thm:balance_phi}. Lastly, the minimum nonzero eigenvalue assumption could be relaxed to allow it to slowly diminish to 0, but then it would appear in the rate of $\Lambda_{n,J^*}$. For a cleaner presentation, we will maintain the current Assumption \ref{assum:moment}.
\end{myRem}

% \begin{prop}[Covariate selection consistency]\label{prop:consistency}
% Under Assumption \ref{assum} \ref{assum:small_ball}-\ref{assum:minimal_signal}, for $a=0,1$, suppose 
% $$
% \lambda_a = 2\sqrt{2}\sigma_\varepsilon\sqrt{\frac{(1+\alpha)\log n_a + \log p}{n_a}}\,,
% $$
% with $n_a>c_{1a} s'\log(ep/s')/v^2$ for some constants $c_{10}$ and $c_{11}$. There exist $c_{20}$, $c_{21}$ such that for any $\alpha>0$, with probability at least $1-n_1^{-1-\alpha} - n_0^{-1-\alpha} - \exp(-c_{21} n_1v^2) - \exp(-c_{20} n_0v^2)$,
% \begin{equation}\label{eq:consistency}
% \widehat{J} \coloneq \bigcap_{a\in 
% \lbrace0,1\rbrace} \left\{ j:| \hat{\beta}_j(a)|\geq \frac{128\sqrt{2}}{u^2 v}\sigma_\varepsilon s\sqrt{\frac{(1+\alpha)\log n_a + \log p}{n_a}} \right\} = J^* \,.
% \end{equation}
% \end{prop}

\begin{prop}[Covariate selection consistency]\label{prop:consistency}
Under Assumption \ref{assum} \ref{assum:subgaussian}-\ref{assum:minimal_signal} with the defined constants, for $a\in\lbrace0,1 \rbrace$, let the penalty parameter satisfy
$$
\lambda_a = 2c \sigma_\varepsilon \max\left( \sqrt{\frac{(1+\alpha)\log n_a + \log (2p)}{\xi_a n_a}}, \frac{(1+\alpha)\log n_a + \log (2p)}{\xi_a n_a} \right)\,,
$$
Then, when $n_a>C\log(p)$, with probability at least $1-n_0^{-1-\alpha} - n_1^{-1-\alpha} - \exp\lbrace-C_0 n_0\rbrace - \exp\lbrace-C_1 n_1  \rbrace$, 
\begin{equation}\label{eq:consistency}
\widehat{J} \coloneq \bigcap_{a\in \lbrace 0, 1 \rbrace} \left\{ j: |\hat{\beta}_j(a)| \geq l_\beta \right\} = J^* \, ,
\end{equation}
where $C$, $C_0$ and $C_1$ are some positive constants depending on $\iota$, $r_0$ and $r_1$ in the assumptions.
\end{prop}
Note that, (\ref{eq:consistency}) implies that if $p = o(e^{n_a})$ for $a \in 
\lbrace0,1 \rbrace$, then $\widehat{J}^{((n-N_0)/N)} = J^*$ almost surely as $n_0,n_1$ diverge, indicating the consistent selection of the significant covariates  from step (4) of {\sf ARCS}.

\begin{thm}\label{thm:balance_phi}

Suppose Assumptions \ref{assum} \ref{assum:subgaussian} - \ref{assum:moment} hold and  {\sf ARCS} is used, then 
$\E(\lVert \Lambda_{n,J^*}\rVert_2^2) = O(n^{\frac{1}{\gamma-1}}s_\phi^{\frac{2\gamma-3}{\gamma-1}}) $ and $\Lambda_{n,J^*} = O_p(n^{\frac{1}{2(\gamma-1)}}s_\phi^{\frac{2\gamma-3}{2\gamma-2}})$. Further, if $\E (\lVert \phi(X_{i,J^*})\rVert_2^\gamma) = O(s_\phi^{\frac{\gamma}{2}})$ for all $\gamma>2$, then $\Lambda_{n,J^*} = O_p(n^\epsilon s^{1-\epsilon}_\phi)$ for any small $\epsilon>0$.

\end{thm}

Since $\gamma>2$, Theorem \ref{thm:balance_phi} indicates that the convergence rate of the imbalance measure under {\sf ARCS} is faster than that of many other CAR procedures, such as $O_p(\sqrt{n})$ under complete randomization. Additionally, with simple algebraic calculations, $\Lambda_{n,J^*} = o_p(\sqrt{n})$ when $s_\phi = o(n^{\frac{\gamma-2}{2\gamma-3}})$. This allows $s_\phi$ to diverge with a rate defined by $\gamma$. Furthermore, if relaxing Assumption \ref{assum} \ref{assum:moment} is of concern, by allowing the minimum nonzero eigenvalue of $\E(\phi(X_{i,J^*})\phi(X_{i,J^*})^\top)$ to diminish with a rate $l_n$ to 0 (instead of lower bounded by a positive constant as in Assumption \ref{assum} \ref{assum:moment}), $\Lambda_{n,J^*} = O_p(n^{\frac{1}{2(\gamma-1)}}s_\phi^{\frac{2\gamma-3}{2\gamma-2}}l_n^{-\frac{\gamma-2}{2\gamma-2}} )$.   

\begin{cor}[{\sf ARCS-M}]\label{cor:M}
Suppose the conditions of Theorem~\ref{thm:balance_phi} hold. If {\sf ARCS-M} is used, then
\begin{itemize}
\item[(a).] $\sum_{i=1}^n (2T_i-1) = O_p(n^{\frac{1}{2(\gamma-1)}}s_\phi^{\frac{2\gamma-3}{2\gamma-2}})$,
\item[(b).] $(\bar{X}_{n,J^*}(1) - \bar{X}_{n,J^*}(0))^\top \Sigma_{J^*}^{-1} (\bar{X}_{n,J^*}(1) - \bar{X}_{n,J^*}(0))=O_p(n^{-\frac{2\gamma-3}{\gamma-1}}s_\phi^{\frac{3\gamma-4}{\gamma-1}})$, where $\bar{X}_{n,J^*}(a) = n_a^{-1} \sum_{1\leq i\leq n:T_i=a} X_{i,J^*}$ for $a\in \lbrace0,1 \rbrace$.
\end{itemize}
\end{cor}
In particular, with a diverging $s_\phi = o(n^{\frac{\gamma-2}{2\gamma-3}})$, $\sum_{i=1}^n (2T_i-1)=o_p(\sqrt{n})$ matching with the results in \cite{ma2022new}; and $(\bar{X}_{n,J^*}(1) - \bar{X}_{n,J^*}(0))^\top \Sigma_{J^*}^{-1} (\bar{X}_{n,J^*}(1) - \bar{X}_{n,J^*}(0)) = o_p(n^{-\frac{\gamma-1}{2\gamma-3}})$. If we assume $s_\phi=O(1)$, $\sum_{i=1}^n (2T_i-1)=O_p(n^{\frac{1}{2(\gamma-1)}})$ and $(\bar{X}_{n,J^*}(1) - \bar{X}_{n,J^*}(0))^\top \Sigma_{J^*}^{-1}(\bar{X}_{n,J^*}(1) - \bar{X}_{n,J^*}(0))=O_p(n^{-\frac{2\gamma-3}{\gamma-1}})$.  

Furthermore, since the imbalance measure considered in {\sf ARCS-M} is proportional to the Mahalanobis distance, we would like to provide a discussion about its comparison with \cite{qin2022adaptive}, where the {\sf ARM} method is developed to uniquely minimize the Mahalanobis distance, assuming $X$ follows a Gaussian distribution. Corollary \ref{cor:M} generalizes {\sf ARM} by extending it to sub-Gaussian distributions. Note that {\sf ARM} in \cite{qin2022adaptive} works with only finite dimension assuming $s_\phi=O(1)$. Thus, when $\E (\lVert \phi(X_{i,J^*})\rVert_2^\gamma) = O(s_\phi^{\frac{\gamma}{2}})$ for all $\gamma > 2$, the rate in Corollary \ref{cor:M}(b) degenerates to  that of \cite{qin2022adaptive},  that is, $O_p(n^{-2})$. This is much faster than the rate for complete randomization, i.e., $O_p(n^{-1})$, when $s_{\phi}=O(1)$.

The next corollary is an immediate consequence of Theorem \ref{thm:balance_phi}.

\begin{cor}[{\sf ARCS-COV}]\label{cor:cov} Assume the conditions of Theorem~\ref{thm:balance_phi} hold,  with {\sf ARCS-COV},
\begin{itemize}
\item[(a).] $\sum_{i=1}^n (2T_i-1) = O_p(n^{\frac{1}{2(\gamma-1)}}s_\phi^{\frac{2\gamma-3}{2\gamma-2}})$,
\item[(b).] $\sum_{i=1}^n (2T_i-1) X_{ij} = O_p(n^{\frac{1}{2(\gamma-1)}}s_\phi^{\frac{2\gamma-3}{2\gamma-2}})$ for $j\in J^*$,
\item[(c).] $\sum_{i=1}^n (2T_i-1) X_{ij} X_{ij'} = O_p(n^{\frac{1}{2(\gamma-1)}}s_\phi^{\frac{2\gamma-3}{2\gamma-2}})$ for $j,j' \in J^*$.
\end{itemize}

\end{cor}
This implies that the convergence rates in Corollary \ref{cor:cov} (a)-(c) are all $o_p(\sqrt{n})$ when $s_\phi = o(n^{\frac{\gamma-2}{2\gamma-3}})$. This matches the results in \cite{ma2022new} with a more reasonable moment condition in high-dimensional settings, and is better than the $O_p(\sqrt{n})$ under complete randomization.

\begin{prop}[Parallel results to those in \cite{ma2022new} without covariate selection]\label{prop:highd}
Suppose $X_1,\dots,X_n$ are i.i.d. copies of $X$ and $\E (\lVert \phi(X)\rVert_2^\gamma) = O(a_n^{\gamma/2})$ for some $\gamma>2$, and a rate $a_n$ to be discussed. Also, let us assume the nonzero eigenvalues of $\E(\phi(X) \phi(X)^\top)$ are lower-bounded by a positive constant. Then, utilizing methods in \cite{ma2022new}, $\Lambda_{n} = \sum_{i=1}^n (2T_i-1)\phi(X_i) = O_p(n^{\frac{1}{2(\gamma-1)}}a_n^\frac{2\gamma-3}{2\gamma-2})$.
\end{prop}

Proposition \ref{prop:highd} characterizes the convergence rate of the imbalance measure $\Lambda_n$ for the CAR without covariate selection~\citep{ma2022new}, under high-dimensional settings. \cite{ma2022new} assumes a finite moment condition, which implicitly suggests a form of sparsity constraint in  high-dimensional settings. However, this assumption may lose its validity as the dimension of baseline covariates, i.e., $p$,  increases significantly as a function of $n$ and covariate selection is not considered. Typically, in the absence of a sparsity assumption or the covariate selection, the parameter $a_n$ should depend on $p$ and be considerably larger than $s_\phi$. Thus, employing CAR without covariates selection could lead to  inadequate  balance, particularly as $a_n$ diverges.

While our analysis primarily evaluates the properties of covariate balance for continuous covariates, the phenomenon outlined in Proposition~\ref{prop:highd} is also applicable to procedures designed for discrete covariates \citep{taves1974minimization,pocock1975sequential,atkinson1982optimum}. These procedures generally perform well when the dimension of the covariates is fixed. However, they may face challenges similar to those identified by \cite{ma2022new} as the dimension of covariates increases, potentially leading to divergence in the imbalance measure. In contrast, our proposed {\sf ARCS} method ensures a stable convergence rate for the imbalance measure, effectively addressing this challenge.

%\begin{myRem}

%The existing CAR methods do not take into account the divergence in the dimensions of covariates. Therefore, only {\sf ARCS} is asymptotically efficient in this scenario. In the following comparison, we assume the dimension of covariates is finite. The asymptotic property of {\sf ARCS} matches the results in \cite{ma2022new}. When balancing discrete covariates, the asymptotic property of {\sf ARCS} is the same as that in \cite{taves1974minimization} and \cite{pocock1975sequential} under Assumption 1 (iv). When balancing continuous covariates, The convergence rate of the imbalance measure in \cite{taves1974minimization}, \cite{pocock1975sequential}, and \cite{atkinson1982optimum} are all $\sum_{i=1}^n (2T_i-1) X_i = O_p(n^{1/2})$, while the imbalance measure in {\sf ARCS} exhibits a better rate of $O_p(n^{\frac{1}{2(\gamma-1)}})$, where $\gamma>2$. Under Special case \ref{ex:M} and Special case \ref{ex:cov}, the rate of the average allocation proportion of {\sf ARCS} is the same as that in \cite{efron1971forcing} biased coin design.
%\end{myRem}

By effectively  selecting covariates, {\sf ARCS}  improves the convergence rate of the imbalance measure on  the important  covariates,  from $O_p(n^{\frac{1}{2(\gamma-1)}}a_n^\frac{2\gamma-3}{2\gamma-2})$ to $O_p(n^{\frac{1}{2(\gamma-1)}}s_\phi^\frac{2\gamma-3}{2\gamma-2})$. The  significant improvement can be illustrated by considering a simple case where  $\gamma=3$, $a_n = O(n)$, and $s_{\phi} = O(n^{1/3-\epsilon})$ for a given $\epsilon > 0$. In this setting, CAR without covariate selection results in $\Lambda_n = O_p(n)$, whereas under {\sf ARCS}, $\Lambda_{n,J^*} = O_p(n^{1/2 - 3\epsilon/4})$. Given that $n^{-1/2}\Lambda_n = O_p(n^{1/2})$ is diverging, CAR without covariate selection becomes ineffective in balancing each covariate component. In contrast, {\sf ARCS} ensures that $n^{-1/2}\Lambda_{n,J^*} = O_p(n^{-3\epsilon/4}) = o_p(1)$. In accordance with the outcome model (\ref{eq:outcome}) and as delineated in Theorem~\ref{thm:treatment_phi}, this improved balance translates into substantial efficiency gain for estimating treatment effects, thus underscoring the advantages of {\sf ARCS}.

% \begin{thm}
% When $p\geq n$, the value of $\text{Imb}_{n, \{1,\dots,p\}}^M$ remain the same for any assignments with equal-sized units in two groups.
% \end{thm}

% \begin{proof}
% $$
% \widehat{\Sigma} = \frac{1}{n-1}\sum_{i=1}^n (X_i - \bar{X})(X_i - \bar{X})^\top.
% $$
% When $p\geq n$, for any centralized $n\times p$ data matrix $\mathbb{X}$ and any $n$-dimensional assignment vector $T=(T_1,\dots,T_n)^\top$, 
% $$
% \text{Imb}_{n,\{1,\dots,p\}}^M = \frac{2(n-1)}{n} (2T-1)^\top \mathbb{X} (\mathbb{X} \mathbb{X}^\top)^- \mathbb{X}^\top (2T-1) = (n-1)(2T-1)^\top (2T-1) = 2n-2.
% $$
% \end{proof}

\subsection{Treatment Effect Estimation}

In this section, we demonstrate the impact of {\sf ARCS} on the subsequent estimation of treatment effect.  For simplicity, we consider the difference-in-means estimator:
$$
\hat{\tau} = \frac{1}{n_1}\sum_{i=1}^n T_i Y_i - \frac{1}{n_0}\sum_{i=1}^n (1-T_i) Y_i.
$$

\begin{thm}(Optimal precision)\label{thm:treatment_phi}
Under the conditions of Theorem \ref{thm:balance_phi}, and let us assume $s = O(1)$, $s_\phi = o(n^{\frac{\gamma-2}{2\gamma-3}})$. With both {\sf ARCS-M} and {\sf ARCS-COV}, we have
$$
\sqrt{n}(\hat{\tau} - \tau) \stackrel{d}{\to} \mathcal{N} (0,4\sigma_\varepsilon^2).
$$
\end{thm}

It is worth noting that with complete randomization, $\sqrt{n}(\hat{\tau} - \tau) $ converges to $\mathcal{N} (0,4\lbrace\sigma_\varepsilon^2+\text{Var}[X^{\top}\beta^*] \rbrace)$ in distribution, thus the difference of $4\text{Var}[X^{\top}\beta^*]$ demonstrates the benefit of covariate balancing  \citep{ma2020statistical,ma2022new}

%\begin{prop}
%Assume the conditions of Proposition \ref{prop:highd} hold, $\sqrt{n}(\hat{\tau} - \tau) = O_p(n^{-\frac{\gamma-2}{2\gamma-2}}a_n^{\frac{2\gamma-3}{2\gamma-2}}s)$.
%\end{prop}

%\begin{myRem}
%The original procedure in \cite{ma2022new} balance $\phi(X_i)$ based on all the dimensions of $X_i$. $\E (\lVert \phi(X)\rVert_2^\gamma) = O(p_\phi^{\gamma/2})$, where $p_\phi$ is the number of dimensions of the vector $\phi(X)$. Following the same proof in Theorem \ref{thm:balance_phi}, we get $\Lambda_n := \sum_{i=1}^n (2T_i-1) \phi(X_i) = o_p(\sqrt{np_\phi})$, which means $\Lambda_{n,J^*} = o_p(\sqrt{np_\phi})$. Under high-dimensional case that $p \gg s$, our method has smaller balance rate $o_p(\sqrt{ns_\phi})$. Also, the estimated treatment effect of {\sf COV} procedure does not have the asymptotically normality because of the growing $p$, and $n^{-1/2}\sum_{i=1}^n (2T_i-1)(\beta^*)^\top \phi(X_i) = o_p(ps)$.
%\end{myRem}
\begin{myRem}
The asymptotic variance of {\sf ARCS-M} and {\sf ARCS-COV} in Theorem~\ref{thm:treatment_phi} achieve the same optimal rates because both methods are based on the same linear generative model (\ref{eq:outcome}). Therefore, if the important covariates in $J^*$ are asymptotically balanced, which is guaranteed by Theorem \ref{thm:balance_phi}, both methods will achieve the optimal asymptotic variance.

This phenomenon might change if we relax the model (\ref{eq:outcome}) to a general outcome model since {\sf ARCS-M} can only balance the mean of covariates, while {\sf ARCS-COV} balances higher-order moments or more sophisticated functions of covariates. In such situations, we generally anticipate $\sigma_{\sf COV}^2 \leq \sigma_{\sf M}^2$, where $\sigma_{\sf COV}^2$ and $\sigma_{\sf M}^2$ are the corresponding counterparts of $\sigma_\varepsilon^2$ for the asymptotic variance of $\hat{\tau}$ under {\sf ARCS-COV} and {\sf ARCS-M}. Lastly, for the {\sf ARCS-M} procedure, when $p > n$, it solves the inverse of $\widehat{\Sigma}_{J^*}$ instead of the singular matrix $\widehat{\Sigma}$ in the Mahalanobis distance, thus significantly improving the results.
\end{myRem}

\section{Numerical Studies}
\label{Section: Numerical Studies}

In this section, we demonstrate the superiority of {\sf ARCS} through four simulation studies. In Example \ref{ex:1} and Example \ref{ex:discrete}, we compare our method {\sf ARCS-M} with complete randomization ({\sf CR}), rerandomization ({\sf RR}) in \cite{morgan2012rerandomization}, adaptive randomization via Mahalanobis distance ({\sf ARM}) in \cite{qin2022adaptive}, and compare {\sf ARCS-COV} with the original {\sf COV} procedure without covariate selection in \cite{ma2022new} and {\sf CR}. In terms of simulation models, we consider both the low-dimensional ($n\geq p$) and the high-dimensional settings ($n<p$)\footnote{When $n<p$, we implement {\sf ARM} utilizing a generalized inverse since the sample covariance matrix is not invertible.}. It is worth noting that {\sf RR} may be problematic when $n<p$ \footnote{When $n<p$, it may not be feasible to implement {\sf RR} because we find the Mahalanobis distance remains unchanged for any assignment. Thus, its value may always exceed the threshold defined in {\sf RR}, making it impossible to find assignments that satisfy the stopping rule. }, thus it is not considered for high-dimensional settings.

More specifically, in Example \ref{ex:1} we illustrate the performance of various methods where all the covariates are continuous. While the current theoretical framework only accommodates continuous covariates, {\sf ARCS} procedures perform well numerically when the covariates are discrete or mixed. To further demonstrate the robustness of {\sf ARCS} against variable types, in Example \ref{ex:discrete} we implement a model with a combination of discrete and continuous covariates. Example \ref{ex:additive} considers a non-linear outcome model and we compare the performances of {\sf CR}, {\sf COV}, with that of {\sf ARCS-M-add} and {\sf ARCS-COV-add} discussed in Remark \ref{remark2}. We aim to illustrate the advantages of utilizing sparse additive models in terms of covariate balancing in such situations. In Example \ref{ex:phi_outcome}, we consider a nonlinear outcome model that matches the selection of $\phi$ in {\sf COV}, which is similar to the simulation settings in \cite{ma2022new}. We compare the performance of {\sf COV} with that of {\sf ARCS-COV} and {\sf ARCS-COV-add}. This special setting will demonstrate the superior performance of the {\sf ARCS} methods, not only in terms of covariate balance but also in terms of the standard error of the estimated treatment effect.

We will compare different methods based on the following measures. For {\sf RR}, {\sf ARM}, {\sf ARCS-M}, {\sf ARCS-M-add}, we compare $$\text{Imb}_{n,J^*}^M = (n/2) (\bar{X}_{n,J^{*}}(1) - \bar{X}_{n,J^{*}}(0))^\top \widehat{\Sigma}_{J^{*}}^{-1} (\bar{X}_{n,J^{*}}(1) - \bar{X}_{n,J^{*}}(0)),$$  for which these methods are designed, and it is the Mahalanobis distance with sample covariance matrix.  For {\sf COV}, {\sf ARCS-COV} and {\sf ARCS-COV-add}, we report the results of Differences-in-Normed Covariates Means (DNCM), i.e., $n^2\lVert \bar{X}_{n,J^*}(1) - \bar{X}_{n,J^*}(0)\rVert^2_2$,  Differences-in-Normed Covariances (DNC), i.e., $n^2\lVert \widehat{\Sigma}_{J^*}(1) - \widehat{\Sigma}_{J^*}(0)\rVert_F^2$, and $\text{Imb}_{n,J^*}^\phi$ with $\phi$ in the Special case \ref{ex:cov}.  For these three {\sf COV}-related methods, the randomization is based on the imbalance measure $\text{Imb}_{n,J^*}^\phi$, where DNCM and DNC are components of $\text{Imb}_{n,J^*}^\phi$.

In all of the examples, we report results based on 5,000 repetitions. For the covariate selection step in {\sf ARCS}, we apply 5-fold cross-validation to choose the penalty parameters.  We set initial batch size $N_0=30$, and to obtain comprehensive and reliable results, we systematically vary the incremental batch size $N=1,2,10$. The treatment effect is estimated by the difference-in-means estimator $\hat{\tau}$.\footnote{All numerical analyses in this paper are conducted using R. We implement Lasso regression with cross-validation utilizing the R package \texttt{glmnet}.}

\begin{example}[Continuous covariates only]\label{ex:1}
Let us assume the following outcome model (\ref{sim1}), with $\beta^* = (3, 1.5, 0, 0, 2, \mathbf{0}_{p-5}^\top)^\top$, $\mu(0)=0$ and $\mu(1)=1$,
  \begin{equation}\label{sim1}
      Y_i = \mu(1) T_i + \mu(0)(1-T_i) + X_i^\top\beta^* + \varepsilon_i.
  \end{equation} Furthermore, we generate the i.i.d. covariates $\{X_i\}_{i=1}^n\sim\mathcal{N}(0,\Sigma)$, with  $\Sigma_{ij} = 0.5^{|i-j|}$, $i,j \in \{1,\cdots,p\}$, and $\{\varepsilon_i\}_{i=1}^n$ are i.i.d. $\mathcal{N}(0,1)$. 
\begin{itemize}
\item[(a)] Low-dimensional setting with $s=3$, $p=10$, and $n\in\{60,120\}$.
\item[(b)] High-dimensional setting with $s=3$, $p=150$, and $n\in\{60,120\}$.
\end{itemize}
\end{example}

We summarize the performance of {\sf ARCS-M} and its competitors under settings (a) and (b) in Table \ref{tb:1_ARM_low} and \ref{tb:1_ARM_high}, respectively. Similarly, we report the performance of {\sf ARCS-COV} and its competitors under settings (a) and (b) in Table \ref{tb:1_COV_low} and \ref{tb:1_COV_high}, respectively. Figures \ref{fig:distance-M-Ex1} and \ref{fig:distance-phi-Ex1} visualize the advantage of {\sf ARCS-M} and {\sf ARCS-COV} over their competitors in terms of the distributions of the corresponding imbalance metrics. Figure~\ref{fig:fpr-Ex1} reports the false positive rate and true positive rate for selecting the significant covariates in $J(\beta^*)$ with {\sf ARCS}. It indicates that {\sf ARCS} algorithms successfully solve both the issues of over-incorporation and under-incorporation of covariates as the sample size grows.

\begin{figure}[h!]
\centering
\begin{subfigure}[t]{0.35\textwidth}
    \centering
    \includegraphics[scale=0.55]{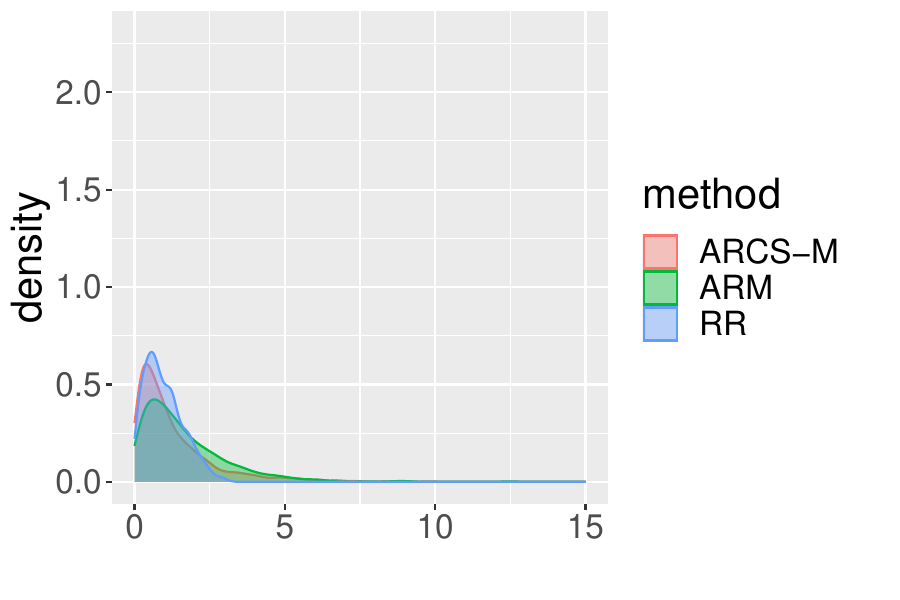}
        \vspace{-0.4in}
    \caption{$n=60$, $p=10$.}
        \vspace{-0.1in}
\end{subfigure}%
\begin{subfigure}[t]{0.43\textwidth}
    \centering
    \includegraphics[scale=0.55]{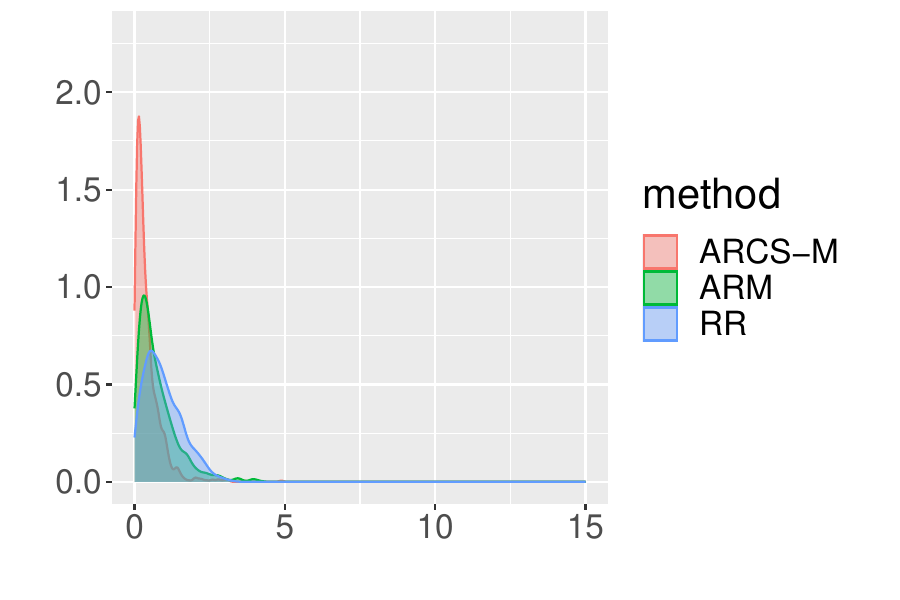}
      \vspace{-0.4in}
    \caption{$n=120$, $p=10$.}
      \vspace{-0.1in}
\end{subfigure}

\vspace{0.25cm}

\begin{subfigure}[t]{0.35\textwidth}
    \centering
    \includegraphics[scale=0.55]{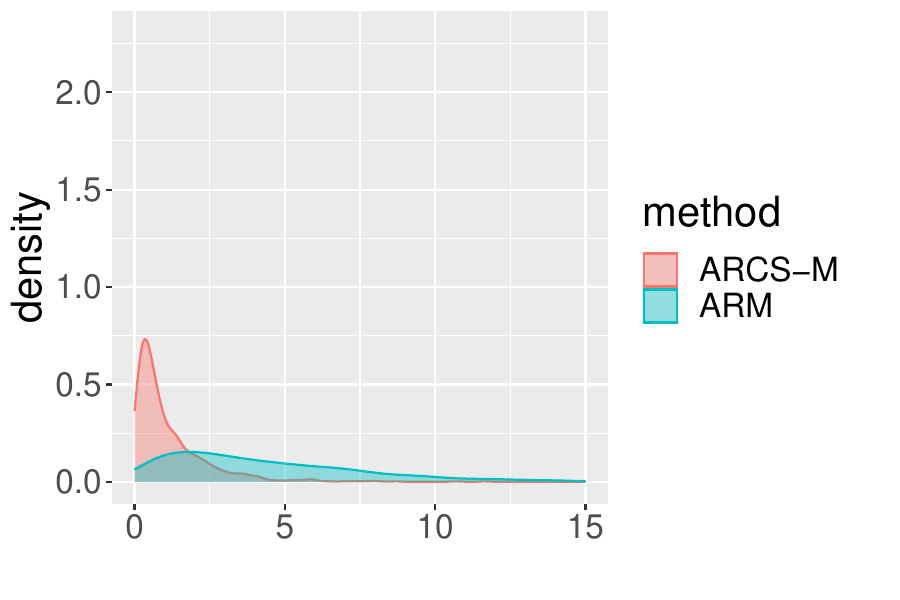}
      \vspace{-0.4in}
    \caption{$n=60$, $p=150$.}
      \vspace{-0.1in}
\end{subfigure}
\begin{subfigure}[t]{0.43\textwidth}
    \centering
    \includegraphics[scale=0.55]{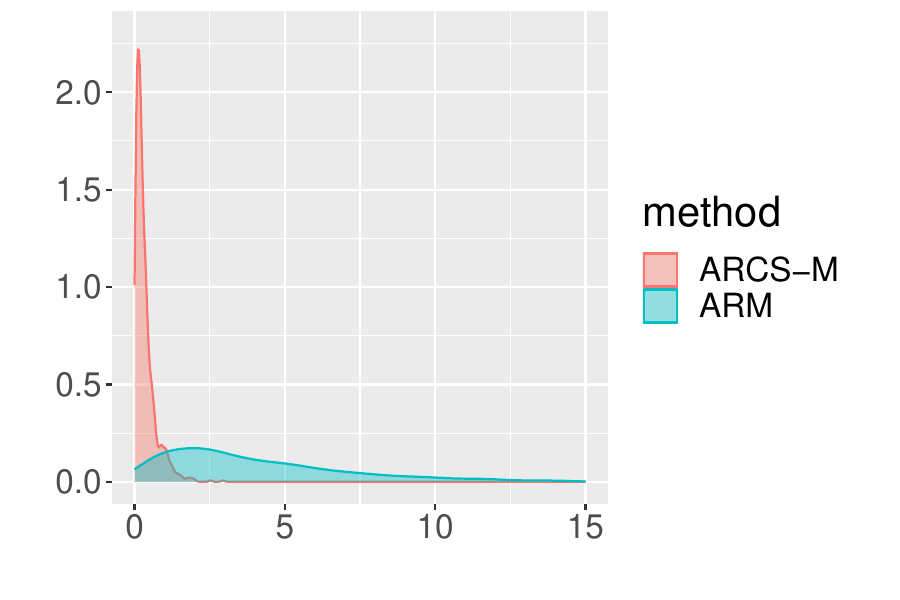}
      \vspace{-0.4in}
    \caption{$n=120$, $p=150$.}
      \vspace{-0.1in}
\end{subfigure}
\caption{The simulated distribution of $\text{Imb}_{i,J^*}^M$ in Example \ref{ex:1}.}
\label{fig:distance-M-Ex1}
\end{figure}

\begin{table}[h!] 
\caption{Simulation results for  Example \ref{ex:1} (a) with $p=10$: (1) $\text{Imb}_{n,J^*}^M$; (2) $\tau$: mean and  $\sqrt{n}\times$standard deviation ($\sqrt{n}$s.d.) and (3) computation time (in seconds). }
\label{tb:1_ARM_low}
\centering
\small
%\footnotesize
\resizebox{0.95\linewidth}{!}{\begin{tabular}{c|c|ccccccc}
\hline
& \multirow{2}{*}{$n$} & \multirow{2}{*}{{\sf CR}} & \multirow{2}{*}{{\sf RR}} & \multirow{2}{*}{{\sf ARM}} &  \multicolumn{2}{c}{{\sf ARCS-M}}  \\
\cline{6-7}
& & & & & $N=2$ & $N=10$ \\
\hline
\multirow{2}{*}{$\text{Imb}_{n,J^*}^M$ }  & 60 & 5.82 & 0.97 & 1.63 & 1.27 & 1.10 \\ 
 & 120 & 5.86 & 0.96 & 0.76 & 0.41 & 0.39 \\ 
 \hline
 \multirow{2}{*}{ $\tau$: mean~($\sqrt{n}$s.d.)} & 60 & 1.06~(9.34) & 1.02~(4.09) & 0.98~(5.20) & 0.98~(4.53) & 0.99~(4.60) \\  
&  120 & 0.98~(9.20) & 0.98~(4.19) & 0.99~(3.82) & 1.00~(3.21) & 1.00~(3.02) \\ 
\hline
\multirow{2}{*}{time (s) }  & 60 & 0.00 & 0.11 & 0.13 & 1.97 & 0.50 \\ 
  & 120 & 0.00 & 0.08 & 0.18 & 5.22 & 1.20 \\ 
\hline
\end{tabular}}
\end{table}

\begin{table}[h!]
\caption{Simulation results Example \ref{ex:1} (b) with for $p=150$: (1) $\text{Imb}_{n,J^*}^M$; (2) $\tau$: mean and  $\sqrt{n}\times$standard deviation ($\sqrt{n}$s.d.) and (3) computation time (in seconds).}
\label{tb:1_ARM_high}
\centering
%\footnotesize
\begin{tabular}{c|c|cccccc}
\hline
& \multirow{2}{*}{$n$} & \multirow{2}{*}{{\sf CR}} & \multirow{2}{*}{{\sf ARM}} &  \multicolumn{2}{c}{{\sf ARCS-M}}  \\
\cline{5-6}
& & & & $N=2$ & $N=10$ \\
\hline

\multirow{2}{*}{$\text{Imb}_{n,J^*}^M$ }  & 60 & 5.94 & 4.47 & 1.12 & 1.17 \\ 
 & 120 & 6.21 & 4.16 & 0.34 & 0.34 \\ 
 \hline
 \multirow{2}{*}{ $\tau$: mean~($\sqrt{n}$s.d.)} & 60 & 1.00~(9.47) & 1.04~(7.95) & 0.98~(4.42) & 0.98~(4.34) \\ 
&  120 & 1.01~(9.72) & 1.02~(7.69) & 0.99~(2.95) & 1.00~(2.94) \\ 
\hline
\multirow{2}{*}{time (s) }  & 60 & 0.00 & 1.62 & 2.39 & 0.55 \\ 
  & 120 & 0.00 & 3.23 & 7.63 & 1.77 \\ 
\hline
\end{tabular}
\end{table}

The superiority of {\sf ARCS-M} is evident in three key aspects, as shown in Tables \ref{tb:1_ARM_low} and \ref{tb:1_ARM_high}. Firstly, in terms of covariate imbalance measured by $\text{Imb}_{n, J^*}^M$, {\sf ARCS-M} consistently outperforms its competitors (in both low-dimensional and high-dimensional settings, with $n=60$ or $n=120$), with the only exception in $n=60,p=10$ where {\sf ARCS-M} is slightly worse than {\sf RR}. The advantage is more evident in the high-dimensional case where $p=150$, as {\sf ARCS-M} reduces the imbalance to only 25\% and 10\% of that of {\sf ARM} (and {\sf CR}) when $n=60$ and $n=120,$ respectively. Secondly, as shown in panel (a) of Figure \ref{fig:fpr-Ex1}, the variable selection step successfully identifies all three influential covariates ($J^*$) even with a small sample size of $30$ (after 1st batch), resulting in a true positive rate of 1 on average. On the other hand, with $s=3$ influential covariates and $p=10$, a false-positive rate of 0.25 indicates an average selection of 1.75 non-influential covariates. Thus, on average, {\sf ARCS-M} balances 4.75 covariates while its competitors balance all 10 covariates. When $p=150$, panel (b) of Figure \ref{fig:fpr-Ex1} reports a true positive rate of 1 after around 10 batches, and a 0 false positive rate consistently. This means {\sf ARCS-M} can balance on exactly the 3 truly important covariates rather than 150 covariates as done in other competing methods. By focusing on a smaller set of covariates covering $J^*$, thanks to the additional covariate selection step, {\sf ARCS-M} mitigates the issue of over-incorporation of covariates, resulting in superior performance in terms of balancing the underlying influential set $J^*$. Thirdly, {\sf ARCS-M} achieves smaller standard errors of the estimated treatment effect compared to its competitors in general, and the improvement is more pronounced with increasing sample size ($n$). While {\sf ARCS-M} with $N=2$ may be slower than {\sf ARM} due to the additional Lasso step, increasing the batch size to $N=10$ significantly reduces computational time. Therefore, for users with limited computational resources, choosing a slightly larger batch size can substantially reduce computation time without sacrificing efficiency, as demonstrated in both Tables \ref{tb:1_ARM_low} and \ref{tb:1_ARM_high}. 

\begin{figure}[h!]
\centering
\begin{subfigure}[t]{0.35\textwidth}
    \centering
    \includegraphics[scale=0.55]{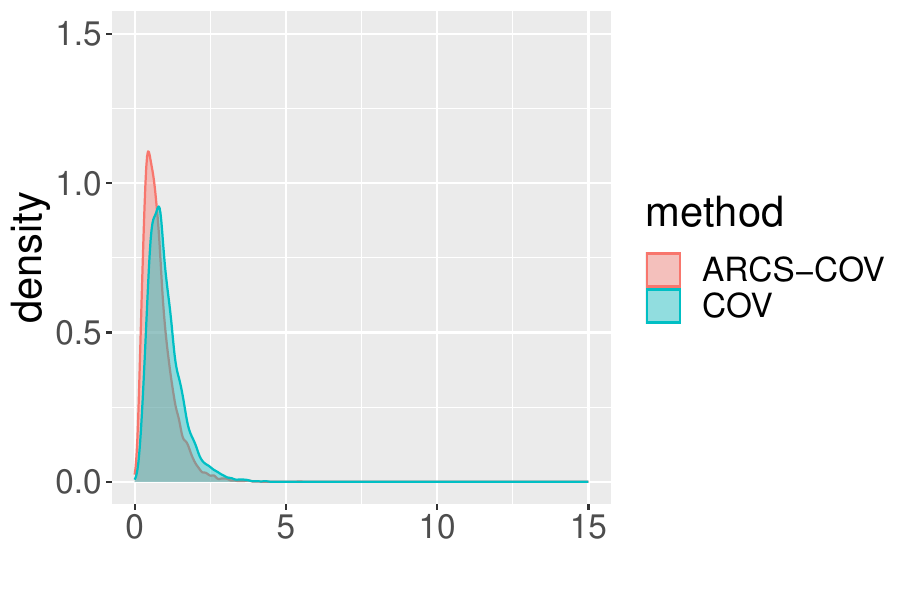}
    \vspace{-0.4in}
    \caption{$n=60$, $p=10$.}
    \vspace{-0.1in}
\end{subfigure}%
\begin{subfigure}[t]{0.43\textwidth}
    \centering
    \includegraphics[scale=0.55]{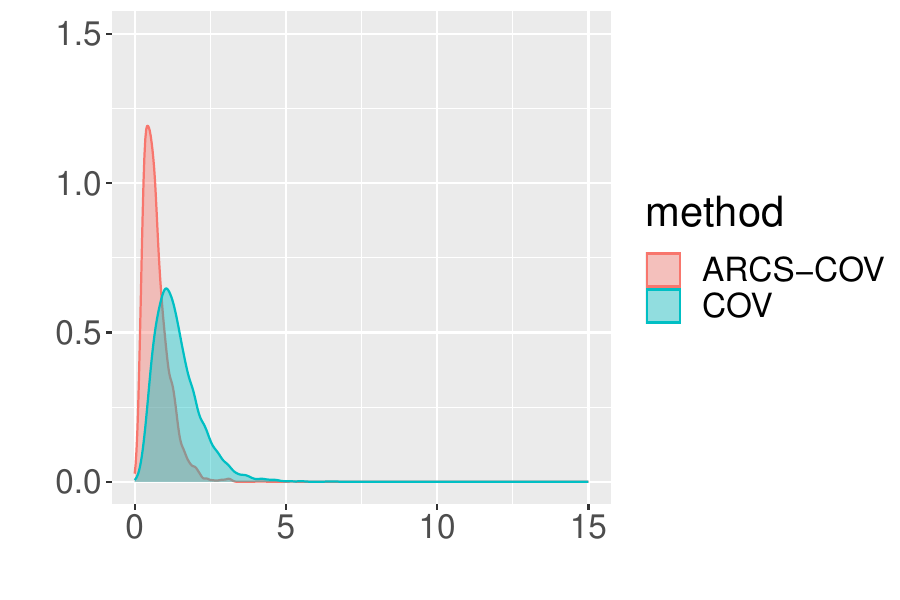}
    \vspace{-0.4in}
    \caption{$n=120$, $p=10$.}
     \vspace{-0.1in}
\end{subfigure}

\vspace{0.4cm}

\begin{subfigure}[t]{0.35\textwidth}
    \centering
    \includegraphics[scale=0.55]{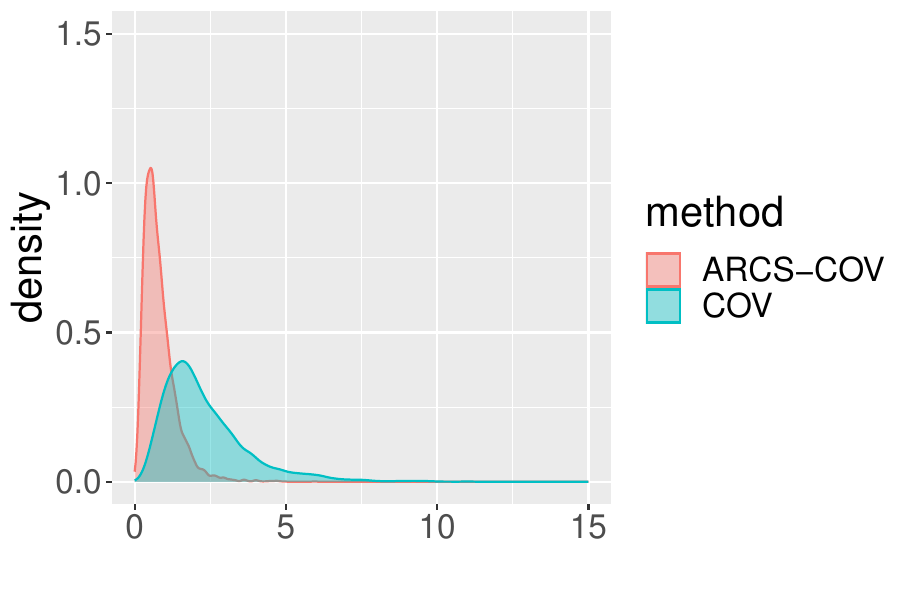}
        \vspace{-0.4in}
    \caption{$n=60$, $p=150$.}
        \vspace{-0.1in}
\end{subfigure}
\begin{subfigure}[t]{0.43\textwidth}
    \centering
    \includegraphics[scale=0.55]{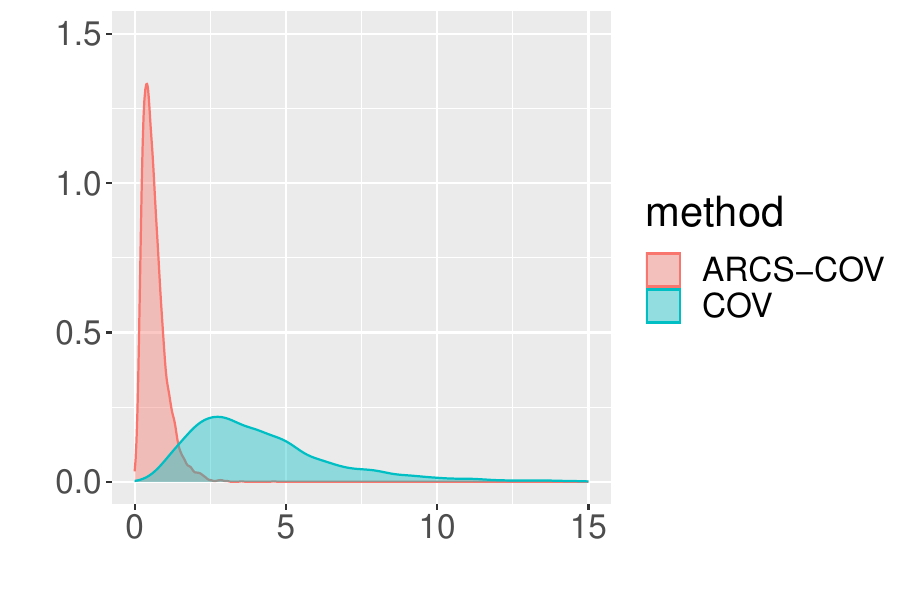}
    \vspace{-0.4in}
    \caption{$n=120$, $p=150$.}
    \vspace{-0.1in}
\end{subfigure}
\caption{The simulated distribution of $\text{Imb}_{i,J^*}^\phi/100$ in Example \ref{ex:1}.}
\label{fig:distance-phi-Ex1}
\end{figure}

\begin{table}[h!]
\caption{Simulation results for  Example \ref{ex:1} (a) with $p=10$ : (1) DNCM; (2) DNC; (3) $\text{Imb}_{n,J^*}^\phi$;  (4) $\tau$: mean and $\sqrt{n}\times$standard deviation in Special case \ref{ex:cov} and (5) computation time (in seconds).}
\label{tb:1_COV_low}
\centering
%\footnotesize
\resizebox{0.9\linewidth}{!}{\begin{tabular}{c|c|cccccc}
\hline
& \multirow{2}{*}{$n$} & \multirow{2}{*}{{\sf CR}} & \multirow{2}{*}{{\sf COV}} &  \multicolumn{2}{c}{{\sf ARCS-COV}}  \\
\cline{5-6}
& & & & $N=1$ & $N=10$ \\
\hline

\multirow{2}{*}{DNCM}  & 60 & 718.78 & 133.09 & 132.09 & 131.84 \\ 
 & 120 & 1453.19 & 147.20 & 108.92 & 106.56 \\ 
 \hline
 \multirow{2}{*}{DNC}  & 60 & 3023.38 & 1678.61 & 1365.60 & 1387.96 \\ 
 & 120 & 5999.52 & 2522.13 & 1286.70 & 1261.87 \\ 
 \hline
\multirow{2}{*}{$\text{Imb}_{n,J^*}^\phi$ }  & 60 & 308.76 & 101.92 & 79.20 & 80.44 \\ 
 & 120 & 620.76 & 141.22 & 72.02 & 71.70 \\ 
 \hline

\multirow{2}{*}{ $\tau$: mean~($\sqrt{n}$s.d.)} & 60 & 0.99~(9.32) & 1.00~(4.07) & 1.00~(4.09) & 1.00~(4.05) \\ 
&  120 & 0.99~(9.56) & 1.00~(3.23) & 1.00~(3.04) & 1.00~(3.02) \\ 
\hline
\multirow{2}{*}{time (s) }  & 60 & 0.00 & 0.00 & 3.53 & 0.46 \\ 
  & 120 & 0.00 & 0.01 & 7.83 & 0.81 \\ 
\hline
\end{tabular}}
\end{table}

\begin{table}[h!]
\caption{Simulation results for   Example \ref{ex:1}(b) with $p=150$: (1) DNCM; (2) DNC; (3) $\text{Imb}_{n,J^*}^\phi$;  (4) $\tau$: mean and $\sqrt{n}\times$standard deviation in Special case \ref{ex:cov} and (5) computation time (in seconds).}
\label{tb:1_COV_high}
\centering
%\footnotesize
\begin{tabular}{c|c|cccccc}
\hline
& \multirow{2}{*}{$n$} & \multirow{2}{*}{{\sf CR}} & \multirow{2}{*}{{\sf COV}} &  \multicolumn{2}{c}{{\sf ARCS-COV}}  \\
\cline{5-6}
& & & & $N=1$ & $N=10$ \\
\hline

\multirow{2}{*}{DNCM}  & 60 & 726.54 & 388.97 & 139.97 & 155.55 \\ 
 & 120 & 1430.37 & 600.13 & 103.25 & 104.17 \\ 
 \hline
 \multirow{2}{*}{DNC}  & 60 & 3043.94 & 3002.91 & 1412.86 & 1453.88 \\ 
 & 120 & 6064.75 & 5968.93 & 1153.76 & 1203.63 \\ 
 \hline
\multirow{2}{*}{$\text{Imb}_{n,J^*}^\phi$ }  & 60 & 310.87 & 231.08 & 82.08 & 87.90 \\ 
 & 120 & 622.20 & 422.87 & 65.14 & 67.20 \\ 
 \hline
\multirow{2}{*}{ $\tau$: mean~($\sqrt{n}$s.d.)} & 60 & 1.01~(9.37) & 0.98~(6.64) & 1.01~(4.16) & 1.00~(4.27) \\ 
&  120 & 1.01~(9.35) & 1.01~(5.77) & 1.00~(3.04) & 0.99~(2.98) \\ 
\hline
\multirow{2}{*}{time (s) }  & 60 & 0.00 & 0.04 & 5.98 & 0.47 \\ 
  & 120 & 0.00 & 0.09 & 13.47 & 1.69 \\ 
\hline
\end{tabular}
\end{table}

\begin{figure}[h!]
 \centering
\begin{subfigure}[t]{0.37\textwidth}
    \centering
    \includegraphics[scale=0.5]{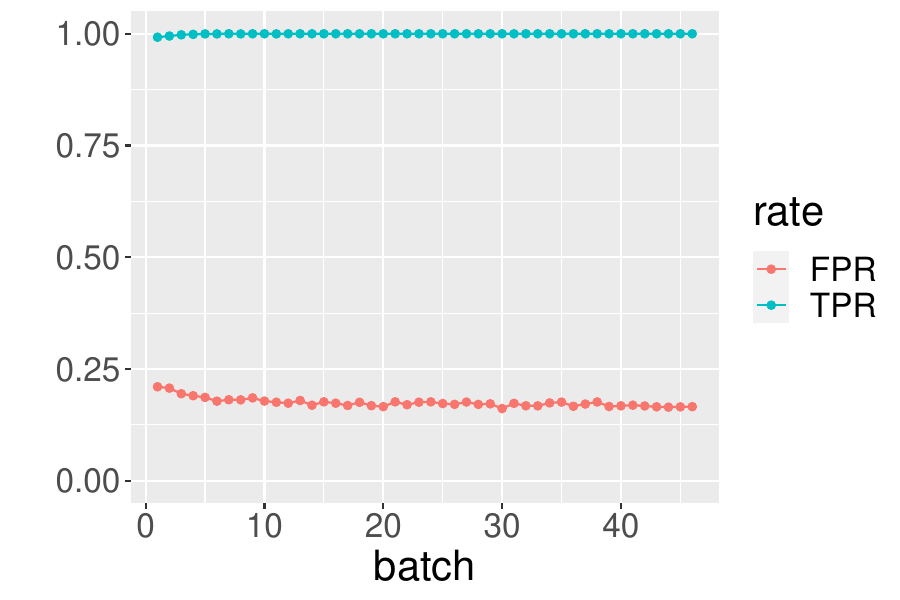}
     \vspace{-0.2in}
    \caption{{\sf ARCS-M}, $n=120$, $p=10$.}
      \vspace{-0.1in}
\end{subfigure}%
\begin{subfigure}[t]{0.4\textwidth}
    \centering
    \includegraphics[scale=0.5]{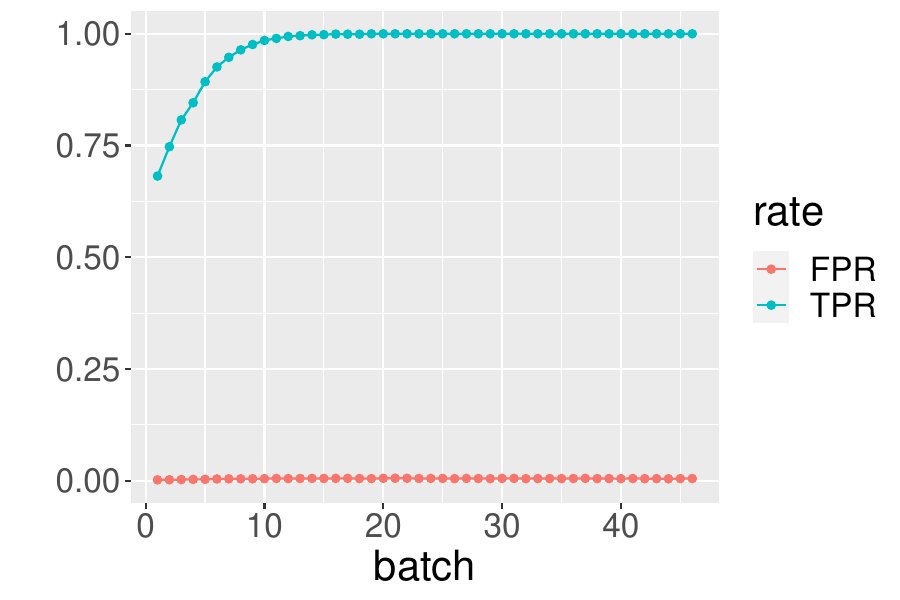}
      \vspace{-0.2in}
    \caption{{\sf ARCS-M}, $n=120$, $p=150$.}
      \vspace{-0.1in}
\end{subfigure}

\vspace{0.5cm}

\begin{subfigure}[t]{0.37\textwidth}
    \centering
    \includegraphics[scale=0.5]{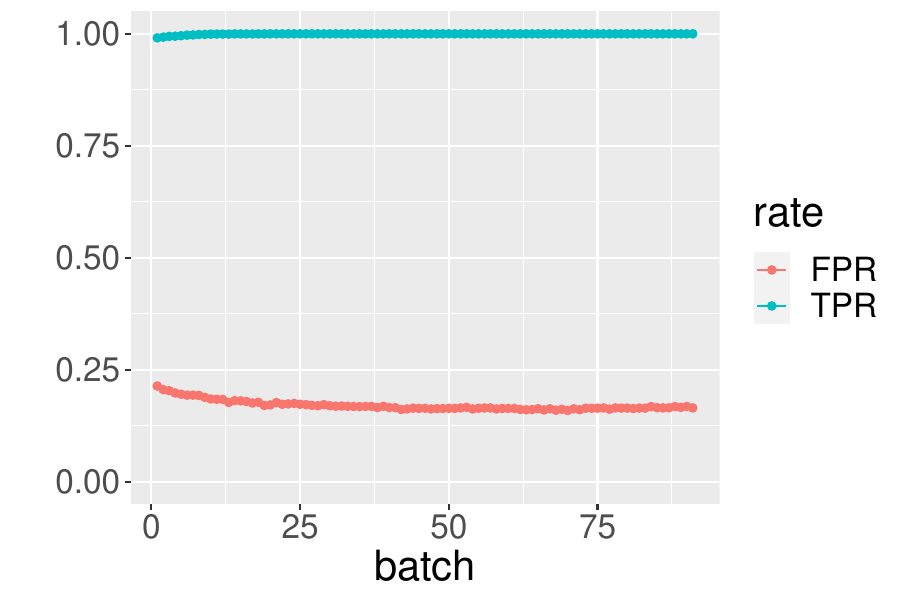}
      \vspace{-0.2in}
    \caption{{\sf ARCS-COV}, $n=120$, $p=10$.}   
\end{subfigure}
\begin{subfigure}[t]{0.4\textwidth}
    \centering
    \includegraphics[scale=0.5]{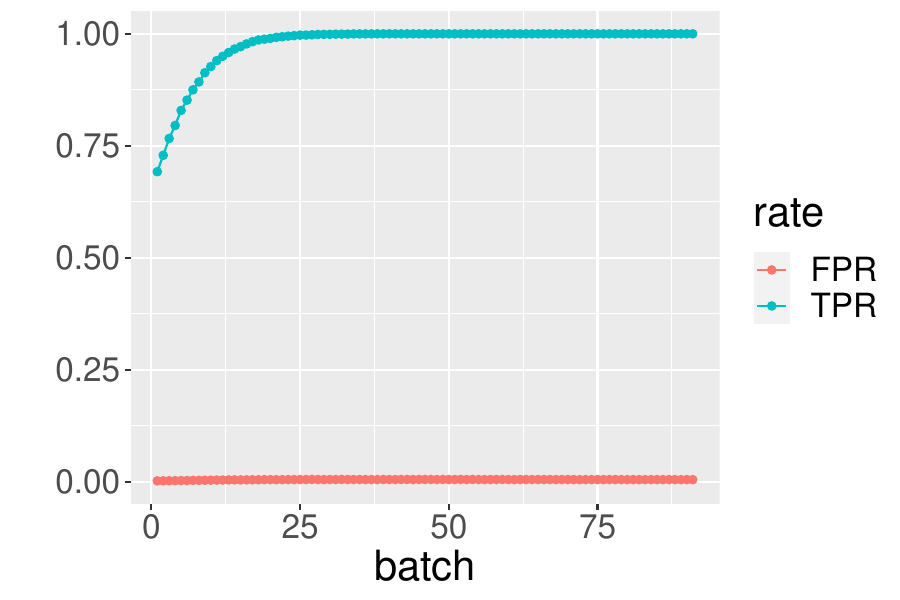}
      \vspace{-0.2in}
    \caption{{\sf ARCS-COV}, $n=120$, $p=150$.}
 
\end{subfigure}
\caption{True positive rate and false positive rate for covariates selection in Example \ref{ex:1}.}
\label{fig:fpr-Ex1}
\end{figure}

Similar messages for {\sf ARCS-COV} have been conveyed in Tables \ref{tb:1_COV_low} and \ref{tb:1_COV_high}. Additionally, we present results on DNCM and DNC (defined at the beginning of this section) following the methodology outlined in \cite{ma2022new}. Although these quantities are not the primary focus of {\sf ARCS-COV}, the numerical results indicate that {\sf ARCS-COV} consistently outperforms {\sf COV} in all scenarios. Specifically, for cases with $n=120, p=10$ and $n=120, p=150$, the DNC achieved by {\sf ARCS-COV} is only approximately 50\% and 20\% of that attained by {\sf COV}, respectively. Similarly, for cases with $n=60, p=150$ and $n=120, p=150$ in Table \ref{tb:1_COV_high}, the DNCM achieved by {\sf ARCS-COV} is only about 35\% and 18\% of that obtained by {\sf COV}, indicating significant reductions in those metrics.

When comparing the results between low-dimensional settings with $p=10$ (Tables \ref{tb:1_ARM_low} and \ref{tb:1_COV_low}) and high-dimensional settings with $p=150$ (Tables \ref{tb:1_ARM_high} and \ref{tb:1_COV_high}), we observe that the performance of {\sf ARCS-M} and {\sf ARCS-COV} remains largely unchanged, while the performance of the competitors deteriorates rapidly. Figure \ref{fig:distance-M-Ex1} and \ref{fig:distance-phi-Ex1} deliver similar messages. There are two potential explanations for this observation: 1) The estimation of the precision matrix is crucial for the implementation of {\sf ARM}. When the number of covariates ($p$) exceeds the sample size ($n$), incorporating all covariates can lead to either an invalid inverse of a low-rank matrix or inaccurate results when using the generalized inverse of $\widehat{\Sigma}$. {\sf ARCS-M} is not affected since the number of selected covariates is smaller than $n$ in general; 2) Both {\sf ARM} and {\sf COV} tend to incorporate more unimportant covariates, which reduces the emphasis on balancing the important covariates. Specifically, {\sf COV} balances all $p=150$ covariates, while {\sf ARCS-COV} balances only the 3 true covariates. As a result, the issue of over-incorporation of non-influential covariates worsens as $p$ increases from 10 to 150.

\begin{example}[A combination of continuous and discrete covariates]\label{ex:discrete}

In accordance with the basic setups outlined in Example \ref{ex:1}, we introduce a modification by replacing the last four continuous covariates with discrete covariates. Specifically, we start by generating i.i.d. $Z_{i1}$, $Z_{i2} \sim \text{Bernoulli}(0.5)$ for $1\le i \le n$. Subsequently, we construct a set of discrete covariates by considering the four-way combinations of indicator functions involving $Z_{i1}$ and $Z_{i2}$, as described in the outcome model (\ref{sim2}). The remaining covariates and random noises are generated following Example \ref{ex:1}, while maintaining the dimension $p$ unchanged. The updated outcome model can be represented as
 \begin{equation}\label{sim2}
      Y_i = \mu(1) T_i + \mu(0)(1-T_i) + X_i^\top\beta^* + \varepsilon_i ,
  \end{equation} where the first $p-4$ coordinates of $X_i$ remain unchanged, $x_{i,p-3}=\1(Z_{i1} = 1, Z_{i2} = 1)$, $x_{i,p-2}=\1(Z_{i1} = 1, Z_{i2} = 0)$, $x_{i,p-1}=\1(Z_{i1} = 0, Z_{i2} = 1)$ and $x_{ip}=\1(Z_{i1} = 0, Z_{i2} = 0)$.
%\begin{align}
%Y_i &= \mu(1) T_i + \mu(0)(1-T_i) + X_{i,\{1,\dots,p-4\}}^\top \beta_{\{1,\dots,p-4\}}^* + \beta_{p-3}^* \1(Z_{i1} = 1, Z_{i2} = 1) \nonumber \\
%& + \beta_{p-2}^* \1(Z_{i1} = 1, Z_{i2} = 0) + \beta_{p-1}^* \1(Z_{i1} = 0, Z_{i2} = 1) + \beta_p^* \1(Z_{i1} = 0, Z_{i2} = 0) + \varepsilon_i.
%\end{align}\label{sim2}
Moreover, we choose $\beta^* = (3, 1.5, 0, 0, 2, \mathbf{0}_{p-9}^\top, 1, 0, 0, 0)^\top$, thus beyond the three influential continuous covariates, we have one more influential discrete covariate that contributes to the response $Y$. Since we have demonstrated the effect of varying sample size $n$ in Example \ref{ex:1}, we will examine only {\sf ARCS-COV}, keeping $n=120$ and varying $p=\{10, 150\}$, with $s=4$.
\end{example}

We present the results in Table \ref{tb:2_COV}. It is evident that when $p=150$, {\sf ARCS-COV} achieves a remarkable improvement over {\sf COV}, reducing all three metrics to at most 30\% of their corresponding values generated by {\sf COV}. Notably, {\sf CR} remains the poorest performing method across all approaches. On the other hand, when $p=10$, while the performance of {\sf ARCS-COV} in terms of the DNCM is comparable to that of {\sf COV}, {\sf ARCS-COV} still outperforms its competitors in terms of the other two metrics. Therefore, Example \ref{ex:discrete} demonstrates that {\sf ARCS} is applicable not only to the selection of continuous covariates but also to a mixture of continuous and discrete covariates. Furthermore, the average TPR and FPR of the last batch is $99.06\%$ and $22.94\%$ respectively when $n=120$, $p=10$; and the corresponding values are $91.71\%$ and $0.68\%$ when $n=120$, $p=150$. The results clearly indicate that, in both cases, we incorporate an average of only 1 non-influential covariate into our model ($22.94\% \times 6 = 1.38$ and $0.68\% \times 146 = 0.99$). This effectively addresses the problem of covariate over-incorporation. Moreover, the average true positive rates exceed 99\% and 91\%, demonstrating that {\sf ARCS} reliably captures all the influential covariates.

\begin{table}[h!]
\caption{Simulation results for Example \ref{ex:discrete}: (1) DNCM; (2) DNC; (3) $\text{Imb}_{n,J^*}^\phi$ in Special case \ref{ex:cov};  (4) $\tau$: mean and $\sqrt{n}\times$standard deviation and (5) computation time (in seconds).}
\label{tb:2_COV}
\centering
%\footnotesize
\begin{tabular}{c|c|cccccc}
\hline
& \multirow{2}{*}{$p$} & \multirow{2}{*}{{\sf CR}} & \multirow{2}{*}{{\sf COV}} &  \multicolumn{2}{c}{{\sf ARCS-COV}}  \\
\cline{5-6}
& & & & $N=1$ & $N=10$ \\
\hline
\multirow{2}{*}{DNCM} & 10 & 1572.69 & 111.10 & 113.21 & 117.20 \\ 
 & 150 & 1510.91 & 668.13 & 165.04 & 164.22 \\ 
 \hline
 \multirow{2}{*}{DNC} & 10 & 6909.01 & 2331.25 & 1736.38 & 1711.03 \\ 
 & 150 & 6963.59 & 6586.54 & 1854.04 & 1855.22 \\ 
 \hline
\multirow{2}{*}{$\text{Imb}_{n,J^*}^\phi$ } & 10 & 640.62 & 111.24 & 82.90 & 83.49 \\ 
 & 150 & 632.19 & 408.33 & 97.75 & 98.59 \\ 
 \hline
\multirow{2}{*}{ $\tau$: mean~($\sqrt{n}$s.d.)} & 10 & 0.99~(9.77) & 0.99~(2.91) & 1.01~(2.92) & 1.00~(2.95) \\ 
&  150 & 1.01~(9.30) & 1.00~(5.81) & 1.00~(3.10) & 1.00~(3.06) \\ 
\hline
\multirow{2}{*}{time (s) } & 10 & 0.00 & 0.00 & 7.11 & 0.78 \\ 
  & 150 & 0.00 & 0.11 & 9.91 & 1.09 \\ 
\hline
\end{tabular}
\end{table}

Since Example \ref{ex:1} has been thorough in exploring various settings, including small and large sample sizes, as well as low and high dimensions, in Examples \ref{ex:additive} and \ref{ex:phi_outcome}, we will focus on a single setting in terms of $n=300$ and $p=150$ to deliver a more focused message.\footnote{We implement the covariate selection step in the sparse additive model using the R package \texttt{sparseGAM}.} 
\begin{example}\label{ex:additive}
We consider a nonlinear outcome model, utilizing $f_1(x)$, $f_2(x)$, $f_3(x)$ and $f_4(x)$ which represent sine, polynomial and exponential functions. Specifically, 

\[
Y_i = \mu(1) T_i + \mu(0) (1 - T_i) + f_1(x_{i1}) + f_2(x_{i2}) + f_3(x_{i3}) + f_4(x_{i4}) + \varepsilon_i,
\]
with
$
f_1(x) = -2\sin(2x),\quad f_2(x) = x^2 - \frac{1}{3},\quad f_3(x) = x - \frac{1}{2},\quad f_4(x) = e^{-x} + e^{-1} -1\,.
$
The covariates and random noises are generated similarly as in Example \ref{ex:1}.
\end{example}

\begin{table}[h!]
\caption{Simulation results for Example \ref{ex:additive}: (1) $\text{Imb}_{n,J^*}^M$; (2) DNCM; (3) DNC; (4) $\text{Imb}_{n,J^*}^\phi$ in Special case \ref{ex:cov} and (5) $\tau$: mean and $\sqrt{n}\times$ standard deviation.}
\label{tb:add_ARM}
\centering
\resizebox{1\linewidth}{!}{\begin{tabular}{c|ccccc}
\hline
& {\sf CR} & {\sf ARM} & {\sf ARCS-M-add} & {\sf COV} & {\sf ARCS-COV-add} \\
\hline
$\text{Imb}_{n,J^*}^M$ &  8.00 &  5.16 & 1.23 & - & - \\ 
\hline
DNCM & 4686.11 & -& - &  1177.82 & 567.72 \\ 
DNC & 26754.32 & - & -  & 23811.09 & 11273.14 \\ 
$\text{Imb}_{n,J^*}^\phi$ & 2217.85 &  -& -& 1255.95 & 538.76 \\ 
 %\hline
%time (s) & 0.00 & 15.25 & 804.82 & 1.12 & 1228.76 \\
\hline
$\tau$: mean~($\sqrt{n}$s.d.) & 1.00~(6.20) & 1.00~(5.66) &  1.00~(5.15) & 1.01~(5.28) & 1.01~(4.66) \\ 
\hline
\end{tabular}}
\end{table}

From Table \ref{tb:add_ARM}, {\sf ARCS-M-add} reduces the Mahalanobis distance to 16\% and 24\% of the values achieved by {\sf CR} and {\sf ARM}, respectively. Further, {\sf ARCS-COV-add} method reduces all the three measures DNCM, DNC, and $\text{Imb}_{n,J^*}^\phi$  to at most 50\% of the values achieved by {\sf COV}.  Thus overall, {\sf ARCS-M-add} and {\sf ARCS-COV-add} are effective in balancing the important covariates when the outcome model is nonlinear.
\begin{example}\label{ex:phi_outcome}
We choose the following outcome model such that it matches the choice of $\phi$ in {\sf COV}.The covariates and random noises are generated similarly as in Example \ref{ex:1}.

\[
Y_i = \mu(1) T_i + \mu(0) (1 - T_i) + 3x_{i1} + 3x_{i2} + 3x_{i1}^2 + 3x_{i2}^2 + 3x_{i1}x_{i2} + \varepsilon_i.
\]

\end{example}

Table \ref{tb:add_COV} presents a three-fold message. First, when the underlying outcome model matches the choice of $\phi$ that defines the imbalance measure in the CAR procedure, both {\sf ARCS-COV} and {\sf ARCS-COV-add} significantly improve the inference result by reducing the standard error of the estimated treatment effect to approximately half. {\sf COV} is expected to perform well in this setting, as stated in \cite{ma2022new}, but it appears to underperform because the outcome model (defined with only $x_{i1}$ and $x_{i2}$) does not match the $\phi$ that is defined on all $p=150$ covariates without selection. Second, in terms of DNCM, {\sf ARCS-COV} mildly improves the results  compared to {\sf COV}, while {\sf ARCS-COV-add} directly reduces the DNCM measure to less than 27\% of the value obtained with {\sf COV}. This illustrates the advantage of utilizing the sparse additive models. Third, in terms of the DNC and $\text{Imb}_{n,J^*}^\phi$ metrics,  both {\sf ARCS-COV} and {\sf ARCS-COV-add} outperforms {\sf COV} and {\sf CR} by reducing the metrics to less than 28\% of the values obtained with their competitors, with {\sf ARCS-COV-add} performs comparatively better. This suggests that the key improvement in reducing DNC and $\text{Imb}_{n,J^*}^\phi$ is primarily due to the covariate selection step.

\begin{table}[tb]
\caption{Simulation results over 5000 repetitions in Example \ref{ex:phi_outcome}, comparing the performances of {\sf CR}, {\sf COV}, {\sf ARCS-COV} and {\sf ARCS-COV-add} under a nonlinear outcome model matching $\phi$ in {\sf COV}.}
\label{tb:add_COV}
\centering
\begin{tabular}{c|cccccc}
\hline
& {\sf CR} & {\sf COV} &  {\sf ARCS-COV} & {\sf ARCS-COV-add} \\
\hline
DNCM & 2508.21 & 602.80 & 565.41 & 160.24 \\ 
 \hline
DNC & 8625.14 & 8290.49 & 1926.34 & 1740.01 \\ 
 \hline
$\text{Imb}_{n,J^*}^\phi$ & 1182.76 & 617.95 & 178.73 & 112.53 \\ 
 %\hline
%time (s) & 0.00 & 2.78 & 137.14 & 1792.28 \\ 
\hline
$\tau$: mean~($\sqrt{n}$s.d.) & 0.95~(7.79) & 0.98~(7.41) & 0.99~(3.43) & 1.03~(3.63) \\ 
\hline
\end{tabular}
\end{table}

\section{An Example based on Clinical Trial Study}
\label{sec:real data}
We present an extensive analysis of an illustrative clinical trial study to highlight the advantages of {\sf ARCS} over other existing methods without a covariate selection step, in terms of covariate balancing and treatment effect estimation.   The concerned randomized clinical trial was originally conducted to compare the effectiveness of three treatments for chronic depression: nefazodone, the cognitive behavioral-analysis system of psychotherapy, and their combination \citep{keller2000comparison}. As we aim to address the scenario involving two treatments in this paper, we will focus only on the data related to the nefazodone treatment and the combination treatment following \cite{ma2022new}. In this study, the outcome variable $Y$ of interest is the 24-item Hamilton Rating Scale for Depression post-treatment, specifically the last observed score (referred to as FinalHAMD). %Different from \cite{ma2022new} which considered two covariates pre-selected subjectively, age and the 24-item Hamilton Rating Scale for Depression (HAMD24), 
We consider all 57 informative covariates in the dataset, including both discrete and continuous ones, but excluding any meaningless covariates such as the indices. Thus, we work with a dataset of size $n=376$ and  $p=57$. We apply 5-fold cross-validated Lasso, and select $s=2$ covariates, the indicator of white people and  the baseline 24-item Hamilton Rating Scale for Depression, namely RACE and HAMD24. We consider the following two outcome models, one is linear and the other is quadratic.
\begin{align}\label{eq:real_linear}
\text{{\sf Linear: }} \text{FinalHAMD}_i = \mu(1)T_i + \mu(0)(1-T_i) + \beta_1^* \text{RACE}_i + \beta_2^* \text{HAMD24}_i + \varepsilon_i.
\end{align}
\vspace{-2cm}
\begin{align}\label{eq:real_quadratic}
\text{{\sf Quadratic: }}\text{FinalHAMD}_i &= \mu(1)T_i + \mu(0)(1-T_i) + \beta_1^* \text{RACE}_i + \beta_2^* \text{HAMD24}_i \nonumber\\
&+ \beta_3^* \text{HAMD24}_i^2 + \beta_4^* \text{RACE}_i \times \text{HAMD24}_i + \varepsilon_i.
\end{align}

In the design stage, we assess the performance of  {\sf CR}, {\sf ARM}, and {\sf ARCS-M} in terms of the Mahalanobis distance for the RACE and HAMD24 covariates. Similarly, the efficacy of {\sf COV} and {\sf ARCS-COV} is evaluated based on DNCM and DNC metrics for these two covariates. DNCM  represents the sum of squared differences in RACE and HAMD24 between the treatment groups, which will have a small value if the means of the covariates are close in the two treatment arms.  DNC further accounts for the second moment of HAMD24 and the interaction between RACE and HAMD24, providing a more comprehensive measure of covariate imbalance that captures higher-order relationships.  This can be particularly useful when the association between the outcome and the covariates extends beyond simple linear relationships.

We first compute the least-square estimators $(\hat{\mu}(1), \hat{\mu}(0), \hat{\beta})$ based on the original dataset and further get $\hat{\mu}(1) - \hat{\mu}(0)$. Then in each replication, we generate a pseudo dataset based on these estimated parameters and $\varepsilon_i \overset{\text{i.i.d.}}{\sim}\mathcal{N}(0, 1)$. Numerical results are summarized in Table \ref{tb:real_arm} and Figure \ref{fig:fpr-real-data} based on 5,000 replications.  

%\begin{table}[h!]
%\caption{Numerical results for the clinical trial example: comparison across different methods based on both linear and quadratic outcome models.}
%\label{tb:real_arm}
%\centering
%\resizebox{1\linewidth}{!}{\begin{tabular}{c|c|ccccc}
%\hline
% & & {\sf CR} & {\sf ARM} & {\sf ARCS-M} & {\sf COV} & {\sf ARCS-COV} \\
%\hline
%\multirow{3}{*}{{\sf Linear model}} & $\text{Imb}_{n,J^*}^M$ & 4.16 & 1.71 & 0.07 & - & - \\ 
%\cline{2-7}
% & DNCM$(\times 10^3)$ & 38.41 &  -& -& 3.38 & 0.95 \\ 
% & DNC$(\times 10^6)$ & 129.88 &  -& -& 11.46 & 3.16 \\ 
% & $\text{Imb}_{n,J^*}^\phi(\times 10^6)$ & 75.84 &  -& -& 6.69 & 1.34 \\ 
 %\hline
%time (s) & 0.00 & 15.25 & 804.82 & 1.12 & 1228.76 \\
%\cline{2-7}
% & $\tau$: mean~($\sqrt{n}$s.d.) & -4.34~(5.12) & -4.36~(3.68) & -4.33~(2.19) & -4.32~(2.62) & -4.34~(2.07) \\ 
%\hline
%\multirow{3}{*}{{\sf Quadratic model}} & $\text{Imb}_{n,J^*}^M$ & 4.01 & 1.73 & 0.18  & - & - \\ 
%\cline{2-7}
% & DNCM$(\times 10^3)$ & 35.16 &  -& -& 3.43 & 1.52 \\ 
% & DNC$(\times 10^6)$ & 116.39 &  -& -& 11.01 & 5.23 \\ 
% & $\text{Imb}_{n,J^*}^\phi(\times 10^6)$ & 52.70 &  -& -& 6.71 & 4.05 \\ 
 %\hline
%time (s) & 0.00 & 15.25 & 804.82 & 1.12 & 1228.76 \\
%\cline{2-7}
% & $\tau$: mean~($\sqrt{n}$s.d.) & -4.47~(5.62) & -4.47~(4.45) & -4.47~(3.43) & -4.46~(2.92) & -4.47~(2.35) \\ 
%\hline
%\end{tabular}}
%\end{table}

\begin{table}[h!]
\caption{Numerical results for the clinical trial example: comparison across different methods based on both linear and quadratic outcome models.}
\label{tb:real_arm}
\centering
\resizebox{1\linewidth}{!}{\begin{tabular}{c|c|ccccc}
\hline
 & & {\sf CR} & {\sf ARM} & {\sf ARCS-M} & {\sf COV} & {\sf ARCS-COV} \\
\hline
\multirow{3}{*}{{\sf Linear model}} & $\text{Imb}_{n,J^*}^M$ & 3.93 & 0.88 & 0.08 & - & - \\ 
\cline{2-7}
 & DNCM$(\times 10^3)$ & 38.34 & - & - & 3.02 & 0.94 \\ 
 & DNC$(\times 10^6)$ & 127.82 & - & - & 10.02 & 3.13 \\ 
 & $\text{Imb}_{n,J^*}^\phi(\times 10^6)$ & 64.86 & - & - & 6.84 & 1.42 \\ 
 %\hline
%time (s) & 0.00 & 15.25 & 804.82 & 1.12 & 1228.76 \\
\cline{2-7}
 & $\tau$: mean~($\sqrt{n}$s.d.) & -4.33~(5.02) & -4.34~(2.91) & -4.34~(2.06) & -4.34~(2.44) & -4.34~(2.11) \\ 
\hline
\multirow{3}{*}{{\sf Quadratic model}} & $\text{Imb}_{n,J^*}^M(\times 10^6)$ & 4.05 & 0.91 & 0.15 & - & - \\ 
\cline{2-7}
 & DNCM$(\times 10^3)$ & 38.8 & - & - & 3.18 & 1.47 \\ 
 & DNC$(\times 10^6)$ & 129.79 & - & - & 10.47 & 4.94 \\ 
 & $\text{Imb}_{n,J^*}^\phi(\times 10^6)$ & 63.35 & - & - & 6.81 & 3.43 \\ 
 %\hline
%time (s) & 0.00 & 15.25 & 804.82 & 1.12 & 1228.76 \\
\cline{2-7}
 & $\tau$: mean~($\sqrt{n}$s.d.) & -4.48~(5.61) & -4.47~(3.62) & -4.48~(3.17) & -4.47~(2.87) & -4.48~(2.30) \\ 
\hline
\end{tabular}}
\end{table}

\begin{figure}[h!]
 \centering
\begin{subfigure}[t]{0.37\textwidth}
    \centering
    \includegraphics[scale=0.5]{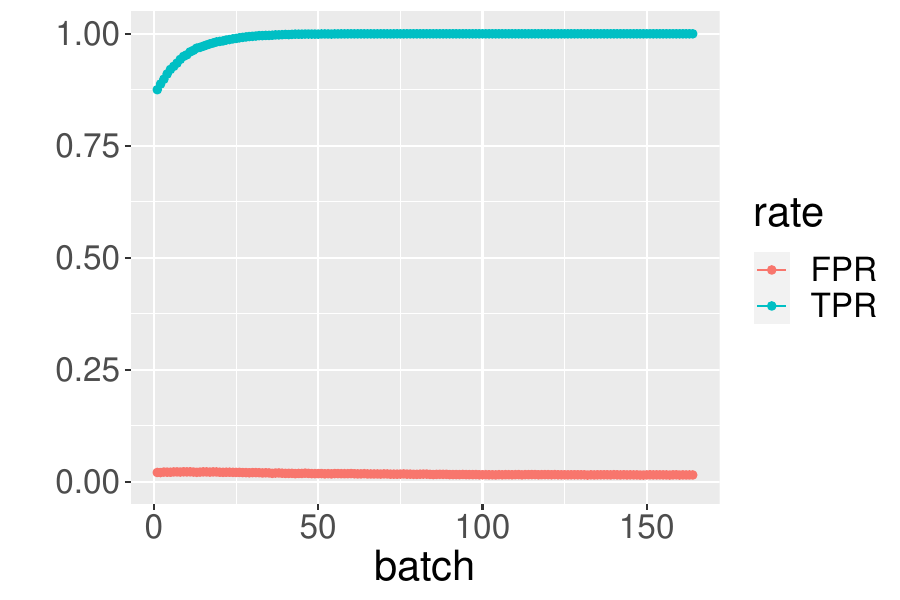}
     \vspace{-0.2in}
    \caption{{\sf ARCS-M}, {\sf Linear model}.}
      \vspace{-0.1in}
\end{subfigure}%
\begin{subfigure}[t]{0.4\textwidth}
    \centering
    \includegraphics[scale=0.5]{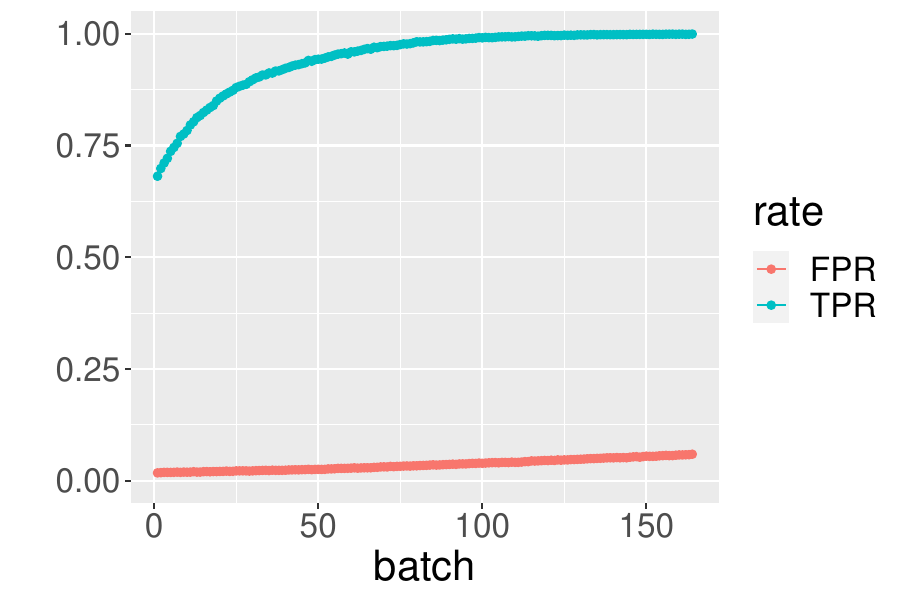}
      \vspace{-0.2in}
    \caption{{\sf ARCS-M}, {\sf Quadratic model}.}
      \vspace{-0.1in}
\end{subfigure}

\vspace{0.5cm}

\begin{subfigure}[t]{0.37\textwidth}
    \centering
    \includegraphics[scale=0.5]{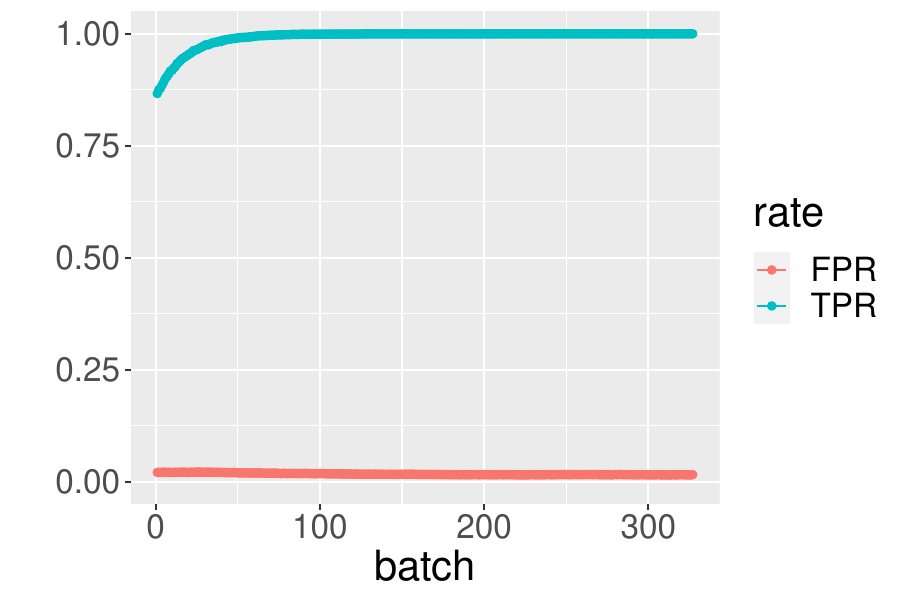}
      \vspace{-0.2in}
    \caption{{\sf ARCS-COV}, {\sf Linear model}.}   
\end{subfigure}
\begin{subfigure}[t]{0.4\textwidth}
    \centering
    \includegraphics[scale=0.5]{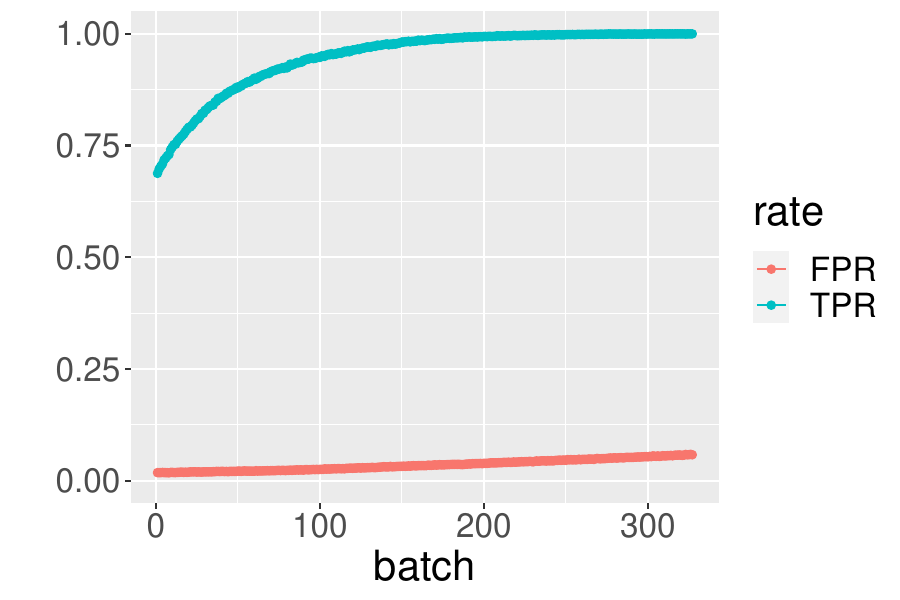}
      \vspace{-0.2in}
    \caption{{\sf ARCS-COV}, {\sf Quadratic model}.}
 
\end{subfigure}
\caption{True positive rate and false positive rate for covariates selection in the clinical trial example.}
\label{fig:fpr-real-data}
\end{figure}

Figure \ref{fig:fpr-real-data} demonstrates the advanced performance of the covariate selection step used in {\sf ARCS}.  For both the linear and the quadratic models, Figure \ref{fig:fpr-real-data} indicates that  {\sf ARCS-M} and {\sf ARCS-COV} control the false positive rate under 10\% or less, and their true positive rates reach 1 once the accumulated sample size exceeds $n/2$ or even earlier. Consequently, it ensures the significant covariates are selected and used for balancing in {\sf ARCS-M} or {\sf ARCS-COV}. 

Thanks to the reduced dimension of covariates, both {\sf ARCS-M} and {\sf ARCS-COV} can improve the balance of the significant covariates and further enhance the efficiency of the treatment effect estimation. Examining first the results from {\sf `Linear model'} in Table \ref{tb:real_arm}, we observe that both {\sf ARCS-M} and {\sf ARCS-COV} significantly reduce the covariate imbalance and the standard error of $\hat{\tau}$, compared to their counterparts {\sf ARM} and {\sf COV}, not to mention {\sf CR} which has the poorest performance consistently. This demonstrates the effect of the extra covariate selection step in {\sf ARCS}. The similar standard errors of $\hat{\tau}$ by {\sf ARCS-M} and {\sf ARCS-COV} indicate that when a linear model is used, balancing the covariate means might be sufficient. When we further look at the  {\sf `Quadratic model'} results, where the outcome model matches the selection of $\phi$ in {\sf COV}, {\sf ARCS-COV} achieves a scaled standard error of $\hat{\tau}$ of 2.30, while {\sf ARCS-M} is 3.17. This shows that when the model becomes more sophisticated, balancing higher-order functions of covariates can further improve the estimation of the treatment effect, validating the development of {\sf ARCS-COV}.

\section{Conclusion}
\label{sec:conclusion}

In this article, {\sf ARCS} procedures are proposed to achieve the simultaneous goals of covariates selection and balancing. Importantly, we extend the finite-dimensional covariates framework to compare the balance properties of the CAR  without covariate selection and {\sf ARCS}. Under a high-dimensional covariates setting, we show that the CAR without covariates selection suffers from a slow convergence rate of its imbalance measure, whereas  {\sf ARCS} achieves a much faster convergence rate for the imbalance over the set of true influential covariates.  The benefit of {\sf ARCS} is further demonstrated by showing improved efficiency in estimating the treatment effect from a linear outcome model.  We conduct extensive numerical studies to demonstrate the appealing properties of  {\sf ARCS} under various settings.

 While our primary focus has been on demonstrating the impact of selecting and balancing covariates on treatment effect estimation, the results presented also provide significant insights into the development of subsequent inference methods for controlling Type I errors. Assuming the linear outcome model (\ref{eq:outcome}), inference for the treatment effect can be conducted based on Theorem~\ref{thm:treatment_phi} and the estimates derived from Lasso regression. However, the model-based estimand depends on the performance of Lasso, which may raise regulatory concerns. As a result, the associated inference might only be approved for use as a secondary analysis.  Consequently, there is a pressing need to expand the presented framework to develop model-robust inference methods following the application of {\sf ARCS} that can potentially be approved as the primary analysis.

As balancing covariate with CAR has become widely accepted in clinical practice, {\sf ARCS} has the potential to further enhance design performance by incorporating a variable selection process. Recent examples of approved trials utilizing CAR techniques include the Exblifep study~\citep{kaye2022effect}, the ZELSUVMI study~\citep{browning2022efficacy}, and the STAMPEDE study~\citep{james2017abiraterone}. However, the effectiveness of CAR can be challenged when many covariates are involved.  For instance, in the STAMPEDE study, 1,917 patients were enrolled and stratified across eight covariates, leading to fewer than two patients per stratum. If additional covariates were also considered, the efficiency of the trial design could be further compromised. In such scenarios, employing {\sf ARCS} to select the most important covariates can improve covariate balance and enhance the accuracy of treatment effect estimation.

To understand the balance properties of the CAR procedure, it is desirable to examine the finite-sample performance of the imbalance measure in more detail. In typical clinical trials, it is assumed that baseline covariates are predefined prior to the experiment, with the dimension considered finite as patient enrollment increases  \citep{hu2012asymptotic, qin2022adaptive, ma2022new}. However, integrating electronic health records and other technological advancements can result in a substantial increase in the dimension of covariates available for the design of the trial. When the number of covariates is large, the features used for achieving balance may have an even higher dimension, potentially leading to worse finite-sample performance on covariate balancing.  Although Proposition \ref{prop:highd}  partially addresses this issue by introducing a simplified assumption, more comprehensive and relaxed assumptions are necessary to adequately delineate the structure of the features utilized in the design. Subsequently,  detailed evaluations of the finite sample properties should be undertaken to understand better the impacts and efficacy of the balance achieved through CAR in high-dimensional settings.

Although Proposition \ref{prop:highd}  partially addresses this issue by introducing a simplified assumption, more comprehensive and relaxed assumptions are necessary to adequately delineate the structure of the features utilized in the design. Subsequently,  detailed evaluations of the finite sample properties should be undertaken to understand better the impacts and efficacy of the balance achieved through CAR in high-dimensional settings. 

The results presented in this paper can be extended in several directions. First, the fast diverging dimension of the functions of covariates used in CAR demonstrates the need to reduce their dimension in order to improve the performance of covariate balancing.  It is thus advantageous to directly select the basis of the functions used in the design with {\sf ARCS}. However, correlations among the basis of these functionals can introduce complexities in the selection process. To mitigate this challenge, one might consider employing the group lasso technique on blocks of functional elements, particularly when supplementary information facilitates the segmentation of these functionals into independent blocks. It is also of interest to understand the impact of the selection of the functionals with {\sf ARCS} on the estimation of treatment effect.  
Second, it is often critical to balance covariates for multi-arm trials.  It would be a worthwhile attempt to extend our framework to settings with multi-treatment arms based on \cite{hu2023multi} and further investigate its theoretical properties. Third, addressing ethical concerns by assigning more patients to the more effective treatment arm is a vital task in clinical trials. Typically, CARA procedures are employed to allocate patients to the beneficial treatment arm based on their baseline covariates, treatment assignments, and responses. While ARCS primarily focuses on balancing important covariates, it is desirable to integrate both covariate balancing
and ethical considerations into the design of the trial. To achieve this dual objective, we
could modify the imbalance measure to allow for unequal allocation, and combine {\sf ARCS}
with existing CARA methods such as those outlined by \cite{rosenberger2001covariate} and \cite{zhang2007asymptotic}.  Developing methodologies in the aforementioned directions would be pressing and important challenges for future research.

\section{Acknowledgements}

Yang Liu's research is supported by the National Natural Science Foundation of China grant 12301324. Lucy Xia's research is supported by the Hong Kong Research Grants Council ECS grant 26305120.

\vspace{-1cm}

\begin{appendices}

\section{Supplement for {\sf ARCS} procedure }

\subsection{Implementation of of {\sf ARCS-M}}
\label{append:ARCS}

Since $\text{Imb}^{\phi}_{i,J^*}$ is proportional to the Mahalanobis distance and $\Sigma_{J^*}$ is unknown in practice, we alternatively implement step (3) of {\sf ARCS} following the procedure in \citep{qin2022adaptive}, described as following step (3'). 
\vspace{-0.3cm}
\begin{itemize}
\item[(3')] Suppose we have assigned $(2i-2)$ units. For $a=0,1$, the imbalance measure via Mahalanobis distance is 
\vspace{-0.3cm}
$$
\text{Imb}^{M}_{2i,\widehat{J}^{(b-1)}} = (2i/2) (\bar{X}_{2i,\widehat{J}^{(b-1)}}(1) - \bar{X}_{2i,\widehat{J}^{(b-1)}}(0))^\top \widehat{\Sigma}_{\widehat{J}^{(b-1)}}^{-1} (\bar{X}_{2i,\widehat{J}^{(b-1)}}(1) - \bar{X}_{2i,\widehat{J}^{(b-1)}}(0))\,,
$$
where $\widehat{\Sigma}_{\widehat{J}^{(b-1)}}$ is the sample covariance matrix calculated by the sample set $\{X_{k,\widehat{J}^{(b-1)}}\}_{k=1}^{2i}$. Let $\text{Imb}^{M}_{2i,\widehat{J}^{(b-1)}}(a)$ be the imbalance measure with $T_{2i-1} = a$, $T_{2i} = 1-a$ for $a=0,1$. Let us assign the $(2i-1)$-th unit with probability
\vspace{-0.1cm}
$$
\p \left(T_{2i-1} = 1 \mid  X_1 , \ldots,X_{2i}, T_1,  \ldots,, T_{2i-2}\right) = \left\{\begin{array}{ll}
\rho, & \text{Imb}^{M}_{2i,\widehat{J}^{(b-1)}}(1) < \text{Imb}^{M}_{2i,\widehat{J}^{(b-1)}}(0), \\
1-\rho, & \text{Imb}^{M}_{2i,\widehat{J}^{(b-1)}}(1) > \text{Imb}^{M}_{2i,\widehat{J}^{(b-1)}}(0), \\
0.5, & \text{Imb}^{M}_{2i,\widehat{J}^{(b-1)}}(1) = \text{Imb}^{M}_{2i,\widehat{J}^{(b-1)}}(0),
\end{array}\right.
$$
where $\rho\in (0.5,1)$, and define $T_{2i} = 1- T_{2i-1}$.
\end{itemize}

\newpage

\subsection{Summary of the {\sf ARCS} Procedures}

\begin{algorithm}[h!] 
\caption{{\sf ARCS} Procedure}
\label{alg:ARCS}
\begin{algorithmic}[1]
\REQUIRE ~~\\ %Input
$X_1,\dots,X_n$: Covariates; \\
$N_0$: Initial batch size; $N$: Batch size; \\
 $\rho$: The probability of the biased coin.
\ENSURE ~~\\ %Output
$T_1,\dots,T_n$: Assignments.
\FOR{$i=1$ to $N_0$}
\STATE Let $T_i=1$ with probability 0.5.
\ENDFOR
\STATE For $a=0,1$, obtain the Lasso estimator
$$
(\hat{\mu}^{(0)}(a), (\hat{\beta}^{(0)}(a))^\top) = \mathop{\arg\min}\limits_{(\mu, \beta)} \left\{\frac{1}{n_a^{(0)}} \sum_{1\leq i\leq N_0: T_i=a} (Y_i - \mu - X_i^\top\beta)^2 + \lambda_a \sum_{j=1}^p |\beta_j| \right\}\,.
$$
\STATE Set $\widehat{J}^{(0)} = J(\hat{\beta}^{(0)}(0)) \cap J(\hat{\beta}^{(0)}(1))$.
\FOR{$b=1$ to $(n-N_0)/N$}
\FOR{$i=N_0+(b-1)N+1$ to $N_0+bN$}
\STATE Calculate $\text{Imb}_{i,\widehat{J}^{(b-1)}}(1)$,  $\text{Imb}_{i,\widehat{J}^{(b-1)}}(0)$.
\IF{$\text{Imb}_{i,\widehat{J}^{(b-1)}}(1) < \text{Imb}_{i,\widehat{J}^{(b-1)}}(0)$}
\STATE Set $T_i=1$ with probability $\rho$.
\ELSIF{$\text{Imb}_{i,\widehat{J}^{(b-1)}}(1) > \text{Imb}_{i,\widehat{J}^{(b-1)}}(0)$}
\STATE Set $T_i=1$ with probability $1-\rho$.
\ELSE
\STATE Set $T_i=1$ with probability 0.5.
\ENDIF
\ENDFOR
\STATE For $a=0,1$, obtain the Lasso estimator
$$
(\hat{\mu}^{(b)}(a), (\hat{\beta}^{(b)}(a))^\top) = \mathop{\arg\min}\limits_{(\mu, \beta)} \left\{\frac{1}{n_a^{(b)}} \sum_{1\leq i\leq N_0+bN: T_i=a} (Y_i - \mu - X_i^\top\beta)^2 + \lambda_a \sum_{j=1}^p |\beta_j| \right\}\,.
$$
\STATE Set $\widehat{J}^{(b)} = J(\hat{\beta}^{(b)}(0)) \cap J(\hat{\beta}^{(b)}(1))$.
\ENDFOR
\RETURN $T_1,\dots,T_n$.
\end{algorithmic}
\end{algorithm}

\section{Proof of the Main Results}
\setcounter{thm}{0}

\subsection{Proof of Proposition \ref{prop:consistency_app}: the selection consistency of Lasso}

In this section, we prove the covariate selection consistency of Lasso in Proposition \ref{prop:consistency_app}. The proof contains two major steps. First, we prove that the Restricted Eigenvalue (RE) condition holds for $\mathbb{X}(a)$ with high probability under Assumption \ref{assum}\ref{assum:subgaussian}\ref{assum:minimal_signal} in Lemma \ref{lem:RE}. Second, we derive the covariate selection consistency based on the RE condition. The second part is based on  \cite{bickel2009simultaneous}.

Let $e_j \in \mathbb{R}^p$ be the vector with the $j$-th element equal to 1 and all other elements equal to 0. For simplicity of notation, $c'$, $C$,  $C'$  are generic constants that vary from line to line. Lemma \ref{lem:RE} provides results for a general $c_a$. Later we will apply the result to $c_a=3$ in our context.

\begin{lemma}\label{lem:RE}
Suppose Assumption \ref{assum} \ref{assum:subgaussian} and \ref{assum:minimal_signal} hold with $c_a=3$ replaced by a generic $c_a$, for $a \in  \lbrace 0,1 \rbrace$. Further, let $\kappa(s,3c_a)$ be the value associated with $p^{1/2}A_a$ under $\text{RE}(s, 3c_a)$, $\kappa(s,c_a)$ be the value associated with $p^{1/2}A_a$ under $\text{RE}(s, c_a)$, and set
$$
m = \min\left(s + s\max_{1\leq j \leq p}\lVert A_a e_j \rVert_2^2 \frac{16(3c_a)^2(3c_a+1)}{r_a^2\kappa^2(s,3c_a)}, p\right).
$$
Suppose 
$$
n_a \geq \frac{2000m\iota^4}{r_a^2} \log \left\lbrace  (mr_a)^{-1}60 pe\right\rbrace.
$$
Then with probability at least $1-2\exp\lbrace-r_a^2n_a/(2000\iota^4)\rbrace$, $\mathbb{X}(a)$ satisfies $\text{RE}(s, c_a)$ condition. Additionally, if we denote $\kappa_{\mathbb{X}(a)}(s,c_a)$ as the value of $\kappa(\cdot,\cdot)$ associated with $\mathbb{X}(a)$ under $\text{RE}(s, c_a)$, then
$$
\kappa_{\mathbb{X}(a)}(s,c_a) \geq (1-r_a)\kappa(s,c_a) > 0.
$$
\end{lemma}
 Lemma~\ref{lem:RE} is a direct consequence of  Theorem 6 in \cite{rudelson2012reconstruction}, we thus omit the proof. The following lemma provides a key step used in the proof of Proposition~\ref{prop:consistency_app}.

\begin{lemma}\label{lem:lambda}
Under Assumption \ref{assum} \ref{assum:subgaussian} and \ref{assum:noise}, for group $a\in\lbrace 0,1 \rbrace$, let
% $$
% \lambda_a = 2\sqrt{2} \sigma_\varepsilon \sqrt{\frac{(1+\alpha)\log n_a + \log p}{n_a}}\,,
% $$

$$
\lambda_a = 2c \sigma_\varepsilon \max\left( \sqrt{\frac{(1+\alpha)\log n_a + \log (2p)}{\xi_a n_a}}, \frac{(1+\alpha)\log n_a + \log (2p)}{\xi_a n_a} \right)\,,
$$

then with probability at least $1-n_a^{-(1+\alpha)}$, for all $\beta\in\R^p$,
% \begin{equation}\label{eq:lem}
% \begin{aligned}
% &\quad \frac{1}{n_a} \lVert \mathbb{X}(a) (\hat{\beta}(a) - \beta^*)\rVert_2^2 + \lambda_a \sum_{j=1}^p |\hat{\beta}_j(a) - \beta_j| \nonumber\\
% &\leq \frac{1}{n_a} \lVert \mathbb{X}(a) (\beta - \beta^*)\rVert_2^2 + 4 \lambda_a \sum_{j\in J(\beta)} |\hat{\beta}_j(a) - \beta_j| \nonumber\\
% &\leq \frac{1}{n_a} \lVert \mathbb{X}(a) (\beta - \beta^*)\rVert_2^2 + 4 \lambda_a \sqrt{|J(\beta)|} \sqrt{\sum_{j\in J(\beta)} |\hat{\beta}_j(a) - \beta_j|^2}
% \end{aligned}
% \end{equation}
\begin{align}
&\quad \frac{1}{n_a} \lVert \mathbb{X}(a) (\hat{\beta}(a) - \beta^*)\rVert_2^2 + \lambda_a \sum_{j=1}^p |\hat{\beta}_j(a) - \beta_j| \nonumber\\
&\leq \frac{1}{n_a} \lVert \mathbb{X}(a) (\beta - \beta^*)\rVert_2^2 + 4 \lambda_a \sum_{j\in J(\beta)} |\hat{\beta}_j(a) - \beta_j| \label{eq:lem11}\\
&\leq \frac{1}{n_a} \lVert \mathbb{X}(a) (\beta - \beta^*)\rVert_2^2 + 4 \lambda_a \sqrt{|J(\beta)|} \sqrt{\sum_{j\in J(\beta)} |\hat{\beta}_j(a) - \beta_j|^2}. \label{eq:lem12}
\end{align}

\end{lemma}

\begin{proof}[Proof of Lemma \ref{lem:lambda}]
Without loss of generality, suppose $Y$ is centralized. By the definition of Lasso, for all $\beta \in \R^p$ and group $a=0,1$, we have
\begin{align*}
\frac{1}{n_a} \lVert Y-\mathbb{X}(a)\hat{\beta}(a)\rVert_2^2 + 2\lambda_a\sum_{j=1}^p |\hat{\beta}_j(a)| \leq \frac{1}{n_a} \lVert Y-\mathbb{X}(a)\beta\rVert_2^2 + 2\lambda_a\sum_{j=1}^p |\beta_j|\,.
\end{align*}
Since $Y = \mathbb{X}\beta^* +\varepsilon$,  it follows that
\begin{align}\label{eq:lem1}
&\quad \frac{1}{n_a} \lVert \mathbb{X}(a)(\beta^* - \hat{\beta}(a))\rVert_2^2 + 2\lambda_a\sum_{j=1}^p |\hat{\beta}_j(a)| \nonumber \\
&\leq \frac{1}{n_a} \lVert \mathbb{X}(a)(\beta^* - \beta)\rVert_2^2 + 2\lambda_a\sum_{j=1}^p |\beta_j| + \frac{2}{n_a} \varepsilon^\top \mathbb{X}(a) (\hat{\beta}(a) - \beta) \,.
\end{align}
Define the event
$$
\mathcal{A}_a = \bigcap_{j=1}^p \left\{\frac{2}{n_a} \left| \sum_{i:T_i=a} \varepsilon_i x_{ij} \right| \leq \lambda_a \right\}\,,\quad a \in \lbrace 0,1 \rbrace,
$$
% then by Assumption \ref{assum} \ref{assum:scale} and \ref{assum:noise}, $\sum_{i:T_i=a} \varepsilon_i x_{ij}/\sqrt{n_a} \sim \mathcal{N}(0,\sigma^2_\varepsilon)$ almost surely. By the concentration inequality for Gaussian random variable,  we have
% $$
% \p(\mathcal{A}_a^c) \leq \sum_{j=1}^p \p\left(|Z| > \frac{\sqrt{n_a}\lambda_a}{2\sigma_\varepsilon}\right) \leq p \exp\left( - \frac{n_a\lambda_a^2}{8\sigma^2_\varepsilon} \right) = n_a^{-(1+\alpha)}\,,
% $$
% where $Z$ is a standard Gaussian random variable. 

By Assumption \ref{assum} \ref{assum:subgaussian} and \ref{assum:noise}, we obtain $\lVert x_{ij}\rVert_{\psi_2} \leq C'$ for some constant $C'$. It follows from Lemma 2.7.7. of \cite{vershynin2018high} that $\varepsilon_i x_{ij}$ is sub-exponential with 
$$
\lVert \varepsilon_i x_{ij} \rVert_{\psi_1} \coloneq \inf\{t>0: \E[\exp(|\varepsilon_i x_{ij}|/t) ]\leq 2\} \leq C'L\sigma_\varepsilon.
$$
Therefore, by Bernstein's inequality, i.e., Corollary 2.8.3 in \cite{vershynin2018high}, it follows that
 \begin{align}
 \p(\mathcal{A}_a^c) & \leq \sum_{j=1}^p \p\left(\frac{2}{n_a} \left| \sum_{i:T_i=a} \varepsilon_i x_{ij} \right| > \lambda_a \right) \nonumber \\ 
 &\leq 2 p\exp\left( -\xi_a n_a\min\left( \frac{\lambda_a^2}{4L^2(C')^2\sigma^2_\varepsilon}, \sqrt{\frac{\lambda_a^2}{4L^2(C')^2\sigma^2_\varepsilon}} \right) \right) \nonumber \\
 & = n_a^{-(1+\alpha)}. \nonumber
 \end{align}

Then, on the event $\mathcal{A}_a$, equation (\ref{eq:lem1}) can be rewritten as 
\begin{align*}
&\quad \frac{1}{n_a} \lVert \mathbb{X}(a)(\hat{\beta}(a) - \beta^*)\rVert_2^2 + \lambda_a \sum_{j=1}^p |\hat{\beta}_j(a) - \beta_j| \\
&\leq \frac{1}{n_a} \lVert \mathbb{X}(a)(\beta - \beta^*)\rVert_2^2 + 2 \lambda_a \sum_{j=1}^p (|\hat{\beta}_j(a) - \beta_j| + |\beta_j| - |\hat{\beta}_j(a)|) \\
&\leq \frac{1}{n_a} \lVert \mathbb{X}(a)(\beta - \beta^*)\rVert_2^2 + 4 \lambda_a \sum_{j\in J(\beta)} |\hat{\beta}_j(a) - \beta_j|,
\end{align*}
yielding the conclusion of (\ref{eq:lem12}). This completes the proof of this lemma. 
\end{proof}
 
Next we prove Proposition~1 based on Lemma \ref{lem:RE} with $c_a=3$ and Lemma \ref{lem:lambda}.

\begin{proposition}[Covariate selection consistency]\label{prop:consistency_app}
Under Assumption \ref{assum} \ref{assum:subgaussian}-\ref{assum:minimal_signal} with the defined constants, for $a\in\lbrace0,1 \rbrace$, let the penalty parameter satisfy
$$
\lambda_a = 2c \sigma_\varepsilon \max\left( \sqrt{\frac{(1+\alpha)\log n_a + \log (2p)}{\xi_a n_a}}, \frac{(1+\alpha)\log n_a + \log (2p)}{\xi_a n_a} \right)\,,
$$
Then, when $n_a>C\log(p)$, with probability at least $1-n_0^{-1-\alpha} - n_1^{-1-\alpha} - 2\exp\lbrace-C_0 n_0\rbrace - 2\exp\lbrace-C_1 n_1  \rbrace$, 
\begin{equation}\label{eq:consistency_app}
\widehat{J} \coloneq \bigcap_{a\in \lbrace 0, 1 \rbrace} \left\{ j: |\hat{\beta}_j(a)| \geq l_\beta \right\} = J^* \, ,
\end{equation}
where $C$, $C_0$ and $C_1$ are some positive constants depending on $\iota$, $r_0$ and $r_1$ in the assumptions.
\end{proposition}

\begin{proof}[Proof of Proposition \ref{prop:consistency_app}]
Let $\delta(a) = \hat{\beta}(a) - \beta^*$, and  set $\beta = \beta^*$ for (\ref{eq:lem11}) and (\ref{eq:lem12}). It follows from (\ref{eq:lem11}) that  $\lVert \delta_{(J^*)^c}(a)\rVert_1 \leq 3 \lVert \delta_{J^*}(a)\rVert_1$, and thus $\lVert \delta(a)\rVert_1 \leq 4\lVert \delta_{J^*}(a)\rVert_1 \leq 4\sqrt{s}\lVert \delta_{J^*}(a)\rVert_2$. On the event $\mathcal{A}_a$ in the proof of Lemma \ref{lem:lambda}, we have
$$
\frac{1}{n_a} \lVert \mathbb{X}(a) \delta(a)\rVert_2^2 \leq 4\lambda_a\sqrt{s}\lVert \delta_{J^*}(a)\rVert_2\,.
$$
Taking $c_a=3$ for  Lemma \ref{lem:RE},  it follows that  $\lVert \delta_{J^*}(a) \rVert_2 \leq 4\lambda_a\sqrt{s}/((1-r_a)\kappa(s,3))^2$,  and it further yields 
\begin{equation}\label{eq:l1_error}
\lVert \hat{\beta}(a) - \beta^*\rVert_1 \leq l_\beta\,.
\end{equation}
In addition, by Assumption~\ref{assum}~\ref{assum:minimal_signal} and the fact that 
$$
\lVert \hat{\beta}(a) - \beta^*\rVert_1 = \sum_{j\notin J^*} |\hat{\beta}(a)| + \sum_{j\in J^*} |\hat{\beta}(a) - \beta^*|,
$$ 
it follows that $\max_{j\notin J^*} |\hat{\beta}_j(a)| \leq l_\beta$, and $\min_{j\in J^*} |\hat{\beta}_j(a)| \geq l_\beta$. Thus, on the intersection of $\mathcal{A}_0$, $\mathcal{A}_1$, and the event that $\mathbb{X}(0)$ and $\mathbb{X}(1)$ satisfy $\text{RE}(s,3)$, we have (\ref{eq:consistency_app}). By Lemma \ref{lem:RE} and \ref{lem:lambda}, this event has probability at least $1-n_1^{-1-\alpha} - n_0^{-1-\alpha} - 2\exp\lbrace-r_0^2n_0/(2000\iota^4)\rbrace - 2\exp\lbrace-r_1^2n_1/(2000\iota^4)\rbrace$, which completes the proof for Proposition~\ref{prop:consistency_app}.
\end{proof}

\subsection{Proof of Theorem \ref{thm:balance_phi_app}: the convergence rate of imbalance measure}\label{sec:proof_2}

In this section, we present the proof of Theorem \ref{thm:balance_phi_app}, Corollary~\ref{cor:M_app}, and Proposition~\ref{prop:consistency_app}.

\begin{theorem}\label{thm:balance_phi_app}
Suppose Assumptions \ref{assum} \ref{assum:subgaussian} - \ref{assum:moment} hold and  {\sf ARCS} is used, then 
$\E(\lVert \Lambda_{n,J^*}\rVert_2^2) = O(n^{\frac{1}{\gamma-1}}s_\phi^{\frac{2\gamma-3}{\gamma-1}}) $ and $\Lambda_{n,J^*} = O_p(n^{\frac{1}{2(\gamma-1)}}s_\phi^{\frac{2\gamma-3}{2\gamma-2}})$. Further, if $\E (\lVert \phi(X_{i,J^*})\rVert_2^\gamma) = O(s_\phi^{\frac{\gamma}{2}})$ for all $\gamma>2$, then $\Lambda_{n,J^*} = O_p(n^\epsilon s^{1-\epsilon}_\phi)$ for any small $\epsilon>0$.
\end{theorem}

\begin{proof}[Proof of Theorem \ref{thm:balance_phi_app}]
The outline for the proof is as follows. By the selection consistency of Lasso, we can divide the patients who have already been assigned into two parts. The first part consists of finite number of patients at the beginning whose covariates are inadequately selected, resulting in over-incorporation or under-incorporation of balancing. The second part consists of those patients whose covariates are consistently selected and balanced by {\sf ARCS}. Then, for the first part of the patients, we derive the rate of the imbalance measure. Finally, based on \cite{ma2022new}, we establish another rate of the imbalance measure for the second part of the patients, depending on both $n$ and $s_\phi$.

Recall $b_n = \lfloor (n-N_0)/N \rfloor$. Define the event $\mathcal{B}_n := \{\widehat{J}^{(b_n)} = J^*\}$ and its complement $\mathcal{B}_n^c$. 

\textbf{Step 1:} By Proposition \ref{prop:consistency_app},  it follows that$$\p(\mathcal{B}_n^c)\leq n_0^{-(1+\alpha)} + n_1^{-(1+\alpha)} + 2\exp(-C_0n_0) + 2\exp(-C_1n_1),$$ which implies $\sum_{n=1}^\infty \p(\mathcal{B}_n^c) < \infty.$ By the Borel-Cantelli lemma, we have
$$
\p(\mathcal{B}_n^c \text{ infinitly often})=\p(\text{for any } m, \text{ there exists } n>m, \mathcal{B}_n^c)=0.
$$
In other words, with probability 1, there exists $m_0$, for all $n>m_0$, $\mathcal{B}_n$ holds. For this $m_0$, we decompose $\Lambda_{n,J^*}$ into two parts,
\begin{equation}\label{eq:Lambda_m}
\Lambda_{n,J^*} = \Lambda_{m_0,J^*} + \sum_{i=m_0+1}^n (2T_i-1) \phi(X_{i,J^*}).
\end{equation}
The remaining part of the proof is to derive the rates of the two terms in (\ref{eq:Lambda_m}). 

\textbf{Step 2:} For the first term in (\ref{eq:Lambda_m}), we have
\begin{align}\label{eq:bound_lambda_m0}
\E \lVert \Lambda_{m_0,J^*}\rVert_2^2 &= \E \left\{\sum_{1\leq i,k\leq m_0} (2T_i-1) (2T_k-1) (\phi(X_{i,J^*}))^\top\phi(X_{k,J^*}) \right\} \nonumber\\
& \leq \sum_{i=1}^{m_0} \E\lVert \phi(X_{i,J^*})\rVert_2^2 + \sum_{1\leq i,k\leq m_0: i\neq k} \E \{ (\phi(X_{i,J^*}))^\top\phi(X_{k,J^*}) \} \nonumber\\
& \leq C_1 s_\phi 
\end{align}
for some constant $C_1$ that does not depend on $n$. The last inequality in (\ref{eq:bound_lambda_m0}) follows from  Assumption \ref{assum} \ref{assum:moment}.

\textbf{Step 3:} Note that when $n>m_0$, by the law of total expectation, $$\E (\cdot) = \E(\cdot \mid \mathcal{B}_n) \p(\mathcal{B}_n) + \E(\cdot \mid \mathcal{B}_n^c) \p(\mathcal{B}_n^c) = \E (\cdot \mid \mathcal{B}_n).$$ For brevity, the expectations in the sequel are all conditional on the event $\mathcal{B}_n$. We will start with an elementary inequality by Taylor's expansion as follows. For $\gamma\geq 2$, and any vectors $u$ and $v$, we have
\begin{align}\label{eq:l2norm_inequality}
\lVert u+v\rVert_2 - \lVert u\rVert_2 \leq \gamma u^\top v \lVert u\rVert_2^{\gamma-2} + c_\gamma (\lVert v\rVert^\gamma_2 + \lVert v\rVert^2_2 \lVert u \rVert_2^{\gamma-2} )
\end{align}
for some constant $c_\gamma$. Denote the second term in (\ref{eq:Lambda_m}) as $\Lambda_{m_0:n,J^*}$. Substitute the decomposition $\Lambda_{m_0:(n+1),J^*} = \Lambda_{m_0:n,J^*} + (2T_{n+1} - 1)\phi(X_{n+1,J^*})$ into (\ref{eq:l2norm_inequality}), it follows that
\begin{align}\label{eq:Taylor}
\lVert \Lambda_{m_0:(n+1),J^*}\rVert_2^\gamma - \lVert \Lambda_{m_0:n,J^*}\rVert_2^\gamma &\leq \gamma (2T_{n+1}-1)\Lambda_{m_0:n,J^*}^\top \phi(X_{n+1,J^*})\lVert \Lambda_{m_0:n,J^*}\rVert_2^{\gamma-2} \nonumber\\
&\quad + c_\gamma (\lVert \phi(X_{n+1,J^*})\rVert_2^\gamma + \lVert \phi(X_{n+1,J^*})\rVert_2^2 \lVert \Lambda_{m_0:n,J^*}\rVert_2^{\gamma-2}).
\end{align}
By some calculations, we have $\text{Imb}_{n,J^*}^\phi (1) - \text{Imb}_{n,J^*}^\phi (0) = 4\Lambda_{n,J^*}^\top \phi(X_{n+1,J^*})$. Then, under {\sf ARCS}, we have
\begin{align}\label{eq:balance_inequality}
&\quad \E((2T_{n+1}-1)\Lambda_{m_0:n,J^*}^\top \phi(X_{n+1,J^*})\mid X_{n+1},\dots,X_1,T_n,\dots,T_1) \nonumber\\
&= \E((2T_{n+1}-1) (\Lambda_{n,J^*} - \Lambda_{m_0,J^*})^\top \phi(X_{n+1,J^*})\mid X_{n+1},\dots,X_1,T_n,\dots,T_1) \nonumber\\
&\leq -(2\rho-1) |\Lambda_{n,J^*}^\top \phi(X_{n+1,J^*})| + c s_\phi
\end{align}
for some constant $c$. Taking the conditional expectation $\E(\cdot\mid \Lambda_{n,J^*})$ on both sides of (\ref{eq:Taylor}), by inequality (\ref{eq:balance_inequality}) and Assumption \ref{assum} \ref{assum:moment}, we obtain
\begin{align*}
&\quad \E (\lVert \Lambda_{m_0:(n+1),J^*}\rVert_2^\gamma\mid \Lambda_{n,J^*}) - \lVert \Lambda_{m_0:n,J^*}\rVert_2^\gamma \\
&\leq (-\gamma(2\rho-1)\E(|\Lambda_{n,J^*}^\top \phi(X_{n+1,J^*})| \mid \Lambda_{n,J^*}) + c s_\phi) \lVert \Lambda_{m_0:n,J^*}\rVert_2^{\gamma-2} + c_\gamma (s_\phi^{\gamma/2}+s_\phi\lVert \Lambda_{m_0:n,J^*}\rVert_2^{\gamma-2}).
\end{align*}
Similar to the proof of that in Theorem 3.1, \cite{ma2022new}, with Assumption \ref{assum} \ref{assum:moment}, we can derive the lower bound of $\E(|\Lambda_{n,J^*}^\top \phi(X_{n+1,J^*})| \mid \Lambda_{n,J^*})$. To be specific, note that there exists an orthonormal matrix $U$ that satisfies
$$
U \E(\phi(X_{i,J^*})(\phi(X_{i,J^*}))^\top) U^\top = \text{diag}(\eta_1,\dots,\eta_d,0,\dots,0),
$$
where $\eta_1,\dots,\eta_d$ are positive eigenvalues of $\E(\phi(X_{i,J^*})(\phi(X_{i,J^*}))^\top)$, and $d\leq s_\phi$. Denote $\tilde{\Lambda}_{n,j}$ and $\tilde{\phi}_j(X_{n+1,J^*})$ as the $j$-th element of $U \Lambda_{n,J^*}$ and $U\phi(X_{n+1,J^*})$, respectively. Then
\begin{align*}
\E(|\Lambda_{n,J^*}^\top \phi(X_{n+1,J^*})| \mid \Lambda_{n,J^*}) &= \E(|\Lambda_{n,J^*}^\top U^\top U \phi(X_{n+1,J^*})| \mid \Lambda_{n,J^*}) \\
&= \E(|\sum_{j=1}^d \tilde{\Lambda}_{n,j}^\top \tilde{\phi}_j(X_{n+1,J^*})| \mid \Lambda_{n,J^*}) \\
&\geq c_0 \lVert\Lambda_{n,J^*}\rVert_2 \\
&\geq c_0\lVert \Lambda_{m_0:n,J^*}\rVert_2 - c_0\lVert \Lambda_{m_0,J^*}\rVert_2
\end{align*}
for some positive constant $c_0$, and the first inequality is by Assumption \ref{assum} \ref{assum:moment} and
$$
\inf\left\{\E\left((\sum_{j=1}^d a_j \tilde{\phi}_j(X_{n+1,J^*}))^2\right):\sum_{j=1}^d a_j^2=1\right\} = \inf\left\{ \sum_{j=1}^d\eta_j a_j^2 : \sum_{j=1}^d a_j^2=1 \right\} = \inf_{1\leq j \leq d}\eta_j > 0. 
$$
Then it follows that
\begin{align}\label{eq:before_function}
&\quad \E (\lVert \Lambda_{m_0:(n+1),J^*}\rVert_2^\gamma\mid \Lambda_{n,J^*}) - \lVert \Lambda_{m_0:n,J^*}\rVert_2^\gamma \nonumber\\
&\leq -\gamma(2\rho-1)c_0\lVert \Lambda_{m_0:n,J^*}\rVert_2^{\gamma-1} + c_\gamma' (s_\phi^{\gamma/2} + s_\phi \lVert\Lambda_{m_0:n,J^*}\rVert_2^{\gamma-2}),
\end{align}
for some constant $c_\gamma'$. Note that the function $h(x) = - c_1 x^{\gamma-1} + c_\gamma' s_\phi x^{\gamma-2}$ for $x\geq 0$ and $c_1>0$ has a maximum value that increases at a rate of $s_\phi^{\gamma-1}$. We also have $\gamma-1>\gamma/2$ by Assumption \ref{assum} \ref{assum:moment}, then for some constant $C_\gamma$,
\begin{align}\label{eq:bound_lambda_gamma}
\E(\lVert \Lambda_{m_0:n,J^*}\rVert_2^\gamma) \leq \E(\lVert \Lambda_{m_0:(n-1),J^*}\rVert_2^\gamma) + C_\gamma s_\phi^{\gamma-1} \leq \dots \leq C_\gamma n s_\phi^{\gamma-1} .
\end{align}
Let $\gamma=2$ in (\ref{eq:before_function}), it follows that
$$
\E (\lVert \Lambda_{m_0:(n+1),J^*}\rVert_2^2) - \E (\lVert \Lambda_{m_0:n,J^*}\rVert_2^2) \leq -2(2\rho-1)c_0\E(\lVert \Lambda_{m_0:n,J^*}\rVert_2) + 2c_\gamma' s_\phi.
$$
Let $m_0\leq m'\leq n$ be the largest integer such that $-2(2\rho-1)c_0\E(\lVert \Lambda_{m_0:m',J^*}\rVert_2) + 2c_\gamma' s_\phi\geq 0$. Then
\begin{align}\label{eq:bound_lambda_n}
\E(\lVert\Lambda_{m_0:n,J^*}\rVert_2^2) &= \E(\lVert\Lambda_{m_0:(m'+1),J^*}\rVert_2^2) + \sum_{i=m'+2}^n (\E(\lVert\Lambda_{m_0:i,J^*}\rVert^2_2) - \E(\lVert\Lambda_{m_0:(i-1),J^*}\rVert^2_2)) \nonumber\\
&\leq \E(\lVert\Lambda_{m_0:(m'+1),J^*}\rVert_2^2) \nonumber\\
&\leq \E(\lVert\Lambda_{m_0:m',J^*}\rVert_2^2) + 2c_\gamma' s_\phi \nonumber\\
&\leq C_2 n^{1/(\gamma-1)} s_\phi^{(2\gamma-3)/(\gamma-1)}
\end{align}
for some constant $C_2$. The last inequality of (\ref{eq:bound_lambda_n}) is by the H$\ddot{\text{o}}$lder's inequality for $t_1=\gamma-1$, $t_2 = \frac{\gamma-1}{\gamma-2}$ and (\ref{eq:bound_lambda_gamma}), 
\begin{align*}
\E(\lVert \Lambda_{m_0:m',J^*}\rVert_2^2) &\leq \{\E(\lVert \Lambda_{m_0:m',J^*}\rVert_2^{t_1\gamma/(\gamma-1)})\}^{1/t_1} \{\E(\lVert \Lambda_{m_0:m',J^*}\rVert_2^{t_2(\gamma-2)/(\gamma-1)})\}^{1/t_2}\\
&\leq (m')^{1/(\gamma-1)}s_\phi^{(2\gamma-3)/(\gamma-1)}C_\gamma'
\end{align*}
for some constant $C_\gamma'$. 

In conclusion, combining (\ref{eq:bound_lambda_m0}) and (\ref{eq:bound_lambda_n}), we obtain $\E\lVert \Lambda_{n,J^*} \rVert_2^2 = O(n^{1/(\gamma-1)}s_\phi^{(2\gamma-3)/(\gamma-1)})$.
\end{proof}

\begin{corollary}\label{cor:M_app}({\sf ARCS-M}) Assume the conditions of Theorem~\ref{thm:balance_phi_app} hold,  with {\sf ARCS-M},
\begin{itemize}
\item[(a).] $\sum_{i=1}^n (2T_i-1) = O_p(n^{\frac{1}{2(\gamma-1)}}s_\phi^{\frac{2\gamma-3}{2\gamma-2}})$,
\item[(b).] $(\bar{X}_{n,J^*}(1) - \bar{X}_{n,J^*}(0))^\top \Sigma_{J^*}^{-1} (\bar{X}_{n,J^*}(1) - \bar{X}_{n,J^*}(0))=O_p(n^{-\frac{2\gamma-3}{\gamma-1}}s_\phi^{\frac{3\gamma-4}{\gamma-1}})$, where $\bar{X}_{n,J^*}(a) = n_a^{-1} \sum_{1\leq i\leq n:T_i=a} X_{i,J^*}$ for $a\in 
\lbrace 0,1\rbrace$.
\end{itemize}

\end{corollary}

\begin{proof}[Proof of Corollary \ref{cor:M_app}]
(a) is a direct result of Theorem \ref{thm:balance_phi_app}. For (b), by Theorem \ref{thm:balance_phi_app} and noticing that
$$
\Lambda_{n,J^*}^\top \Lambda_{n,J^*} = (n_1-n_0)^2 + (n_1 \bar{X}_{n,J^*}(1) - n_0\bar{X}_{n,J^*}(0))^\top \Sigma_{J^*}^{-1} (n_1\bar{X}_{n,J^*}(1) - n_0\bar{X}_{n,J^*}(0)),
$$
we obtain the result.
\end{proof}

Next, we derive a similar result to Theorem \ref{thm:balance_phi_app} in the case where covariate selection is not performed.

\begin{proposition}\label{prop:highd_app}
Suppose $X_1,\dots,X_n$ are i.i.d. copies of $X$ and $\E (\lVert \phi(X)\rVert_2^\gamma) = O(a_n^{\gamma/2})$ for some $\gamma>2$, and a rate $a_n$ to be discussed. Also, let us assume the nonzero eigenvalues of $\E(\phi(X) \phi(X)^\top)$ are lower-bounded by a positive constant. Then, utilizing methods in \cite{ma2022new}, $\Lambda_{n} = \sum_{i=1}^n (2T_i-1)\phi(X_i) = O_p(n^{\frac{1}{2(\gamma-1)}}a_n^\frac{2\gamma-3}{2\gamma-2})$.
\end{proposition}

\begin{proof}[Proof of Proposition \ref{prop:highd_app}]
Similar to the proof in Theorem \ref{thm:balance_phi_app}, we divide the patients into two parts and obtain the increasing rate of the imbalance measure for each part. However, without the covariates selection step, the rate of both parts will include $a_n$, which is a much larger rate than $s_\phi$.

\textbf{Step 1:} Denote $\mathcal{B}_n = \{\widehat{J}^{(b_n)} = J^*\}$. By Proposition \ref{prop:consistency_app} and the Borel-Cantelli lemma, with probability 1, there exists $m_0$ such that for all $n>m_0$, $\mathcal{B}_n$ holds. Then we decompose $\Lambda_{n}$ as follows
$$
\Lambda_n = \Lambda_{m_0} + \sum_{m_0+1}^n (2T_i-1)\phi(X_i).
$$

\textbf{Step 2:} For the first part, we have
\begin{align}\label{eq:bound_lambda_m0_high}
\E \lVert \Lambda_{m_0}\rVert_2^2 &= \E \left\{\sum_{1\leq i,k\leq m_0} (2T_i-1) (2T_k-1) (\phi(X_{i}))^\top\phi(X_{k}) \right\} \nonumber\\
& \leq \sum_{i=1}^{m_0} \E\lVert \phi(X_{i})\rVert_2^2 + \sum_{1\leq i,k\leq m_0: i\neq k} \E \{ (\phi(X_{i}))^\top\phi(X_{k}) \} \nonumber\\
& \leq C_1 a_n
\end{align}
for some constant $C_1$.

\textbf{Step 3:} When $n>m_0$, $\mathcal{B}_n$ holds almost surely. For brevity, the expectation below are all conditional on the event $\mathcal{B}_n$. Denote $\Lambda_{m_0:n} = \sum_{i=m_0+1}^n(2T_i-1)\phi(X_i)$. Substitute $\Lambda_{m_0:(n+1)} = \Lambda_{m_0:n} + (2T_{n+1}-1)\phi(X_{n+1})$ into the inequality (\ref{eq:l2norm_inequality}), we have
\begin{align}\label{eq:Taylor_highd}
\lVert \Lambda_{m_0:(n+1)}\rVert_2^\gamma - \lVert \Lambda_{m_0:n}\rVert_2^\gamma &\leq \gamma (2T_{n+1}-1)\Lambda_{m_0:n}^\top \phi(X_{n+1})\lVert \Lambda_{m_0:n}\rVert_2^{\gamma-2} \nonumber\\
&\quad + c_\gamma (\lVert \phi(X_{n+1})\rVert_2^\gamma + \lVert \phi(X_{n+1})\rVert_2^2 \lVert \Lambda_{m_0:n}\rVert_2^{\gamma-2}).
\end{align}
Assuming $\E(\lVert \phi(X) \rVert_2^\gamma) = O(a_n^{\gamma/2})$, the procedure in \cite{ma2022new} yields
\begin{align*}
&\quad \E((2T_{n+1}-1)\Lambda_{m_0:n}^\top \phi(X_{n+1})\mid X_{n+1},\dots,X_1,T_n,\dots,T_1) \\
&\quad \E((2T_{n+1}-1)(\Lambda_{n} -\Lambda_{m_0})^\top \phi(X_{n+1})\mid X_{n+1},\dots,X_1,T_n,\dots,T_1) \\
&\leq -(2\rho-1) |\Lambda_n^\top \phi(X_{n+1})| + ca_n
\end{align*}
for some constant $c$. Then
\begin{align}\label{eq:before_function2}
&\quad \E (\lVert \Lambda_{m_0:(n+1)}\rVert_2^\gamma\mid \Lambda_{n}) - \lVert \Lambda_{m_0:n}\rVert_2^\gamma \nonumber\\
&\leq (-\gamma(2\rho-1)\E(|\Lambda_{n}^\top \phi(X_{n+1})| \mid \Lambda_{n}) + c a_n) \lVert \Lambda_{m_0:n}\rVert_2^{\gamma-2} + c_\gamma (a_n^{\gamma/2}+a_n\lVert \Lambda_{m_0:n}\rVert_2^{\gamma-2}) \nonumber \\
&\leq -\gamma(2\rho-1)c_0\lVert \Lambda_{m_0:n}\rVert_2^{\gamma-1} + c_\gamma' (a_n^{\gamma/2} + a_n \lVert\Lambda_{m_0:n}\rVert_2^{\gamma-2})
\end{align}
for some constants $c_0,c_\gamma'$, where the last inequality follows from
$$\E(|\Lambda_{n}^\top \phi(X_{n+1})| \mid \Lambda_{n})\geq c_0\lVert \Lambda_{n}\rVert_2\geq c_0\lVert \Lambda_{m_0:n}\rVert_2 - c_0\lVert \Lambda_{m_0}\rVert_2,$$
which is similar to the proof of that in Theorem 3.1 in \cite{ma2022new} and Theorem \ref{thm:balance_phi_app}. The function $h(x) = -c_1 x^{\gamma-1} + c'a_n x^{\gamma-2}$ has a maximum value that increases at a rate of $a_n^{\gamma-1}$. Thus,
$$
\E(\lVert \Lambda_{m_0:n}\rVert_2^\gamma) \leq \E(\lVert \Lambda_{m_0:(n-1)}\rVert_2^\gamma) + C_\gamma a_n^{\gamma-1} \leq \dots \leq C_\gamma n a_n^{\gamma-1}.
$$
Let $\gamma=2$ in (\ref{eq:before_function2}), we obtain
$$
\E (\lVert \Lambda_{m_0:(n+1)}\rVert_2^2) - \E (\lVert \Lambda_{m_0:n}\rVert_2^2) \leq -2(2\rho-1)c_0\E(\lVert \Lambda_{m_0:n}\rVert_2) + 2c_\gamma' a_n.
$$
Define $m_0\leq m'\leq m$ as the largest integer satisfies $-2(2\rho-1)c_0\E(\lVert \Lambda_{m_0:n}\rVert_2) + 2c_\gamma' a_n \geq0$. Following the proof of Theorem \ref{thm:balance_phi_app}, \begin{align}\label{eq:bound_lambda_n_high}
\E(\lVert\Lambda_{m_0:n}\rVert_2^2) &= \E(\lVert\Lambda_{m_0:(m'+1)}\rVert_2^2) + \sum_{i=m'+2}^n (\E(\lVert\Lambda_{m_0:i}\rVert^2_2) - \E(\lVert\Lambda_{m_0:(i-1)}\rVert^2_2)) \nonumber\\
&\leq \E(\lVert\Lambda_{m_0:(m'+1)}\rVert_2^2) \nonumber\\
&\leq \E(\lVert\Lambda_{m_0:m'}\rVert_2^2) + 2c_\gamma' a_n \nonumber\\
&\leq C_2 n^{1/(\gamma-1)} a_n^{(2\gamma-3)/(\gamma-1)}.
\end{align}

Finally, by (\ref{eq:bound_lambda_m0_high}) and (\ref{eq:bound_lambda_n_high}), we have $\E(\lVert \Lambda_{n}\rVert_2^2) = O_p(n^{\frac{1}{\gamma-1}} a_n^{\frac{2\gamma-3}{\gamma-1}})$, and it implies that $\Lambda_n = O_p(n^{\frac{1}{2(\gamma-1)}} a_n^{\frac{2\gamma-3}{2\gamma-2}})$, which completes the proof.
\end{proof}

\subsection{Proof of Theorem \ref{thm:treatment_phi_app}: optimal precision with {\sf ARCS-M} and {\sf ARCS-COV}}

Using the convergence rate of $\Lambda_{n,J^*}$ proved in section \ref{sec:proof_2}, we can establish the following optimal precision for $\hat{\tau}$ with {\sf ARCS-M} and {\sf ARCS-COV} procedures.
\begin{theorem}(Optimal precision)\label{thm:treatment_phi_app}
Under the conditions of Theorem \ref{thm:balance_phi_app}, and let us assume $s = O(1)$, $s_\phi = o(n^{\frac{\gamma-2}{2\gamma-3}})$. Then with both {\sf ARCS-M} and {\sf ARCS-COV}, we have
$$
\sqrt{n}(\hat{\tau} - \tau) \stackrel{d}{\to} \mathcal{N} (0,4\sigma_\varepsilon^2).
$$
\end{theorem}

\begin{proof}[Proof of Theorem \ref{thm:treatment_phi_app}]
First, it is worth noting  that  $s_\phi = o(n^{\frac{\gamma-2}{2\gamma-3}})$ implies that $\sum_{i=1}^n(2T_i-1) = o_p(\sqrt{n})$,
 then it follows from $s=O(1)$ that
\begin{align}
\sqrt{n}(\hat{\tau} - \tau) &=\frac{2}{\sqrt{n}}\sum_{i=1}^n(2T_i-1)(X_i^\top\beta^* + \varepsilon_i) +o_p(1).\label{eq.dist.tau}
\end{align}

In addition,  since $(2T_i-1)^2=1$ for $i=1,\dots,n$, the central limit theorem conditional on $T_1,\dots,T_n$ yields that 
\begin{align}
    \frac{2}{\sqrt{n}}\sum_{i=1}^n(2T_i-1)\varepsilon_i \stackrel{d}{\to} \mathcal{N}(0,4\sigma^2_\varepsilon). \label{eq.dist.error}
\end{align}
Since the asymptotic distribution  in (\ref{eq.dist.error}) does not depend on  $T_1,\dots,T_n$, it also holds unconditionally.

Next, we derive for {\sf ARCS-COV}.  Since  $s_\phi = o(n^{\frac{\gamma-2}{2\gamma-3}})$ and by Corollary \ref{cor:cov}, we have $\sum_{i=1}^n (2T_i-1) x_{ij} = o_p(\sqrt{n})$ for $j\in J^*$. Further with $s=O(1)$,  it follows that $n^{-1/2} \sum_{i=1}^n (2T_i-1) X_i^\top \beta^* = o_p(1)$. Then the asymptotic distribution of $\hat{\tau}$ follows from (\ref{eq.dist.tau}) and (\ref{eq.dist.error}).

For {\sf ARCS-M}, recall the eigen-decomposition $\Sigma_{J^*}^{-1} = \Gamma D \Gamma^\top$ and $\phi(X_{k,J^*}) = (\sqrt{w_0},\sqrt{w_1} \\ X_{k,J^*}^\top \Gamma D^{1/2})^\top$. Denote the element in the $j$-th row and $r$-th column of matrix $\Gamma$ as $\gamma_{jr}$, $j,r=1,\dots,s$, and $\eta_1,\dots,\eta_s$ are the eigenvalues of $\Sigma_{J^*}$. By Assumption \ref{assum} \ref{assum:subgaussian} and $s=O(1)$, $\eta_1,\dots,\eta_s$ are bounded above by a constant. Furthermore, with $s_\phi = o(n^{\frac{\gamma-2}{2\gamma-3}})$ and Theorem \ref{thm:balance_phi_app}, we have $\sum_{i=1}^n (2T_i-1) \sum_{j\in J^*}x_{ij} \gamma_{jr} / \sqrt{\eta_r} = o_p(\sqrt{n})$ for $r=1,\dots,s$, and thus $\sum_{i=1}^n (2T_i-1) x_{ij} = o_p(\sqrt{n})$. This together with  (\ref{eq.dist.tau}) and (\ref{eq.dist.error}) yields the  the asymptotic distribution of $\hat{\tau}$. The proof for this theorem is now complete. 
\end{proof}

\vspace{-1cm}
\section{Additional simulation results}

\begin{example}[$n=500$ version of Example \ref{ex:1}]\label{ex:n}
Consider the following outcome model (\ref{eq:largen_outcome}):
\begin{equation}\label{eq:largen_outcome}
  Y_i = \mu(1) T_i + \mu(0)(1-T_i) + X_i^\top\beta^* + \varepsilon_i,
\end{equation} 
where  $\beta^* = (3, 1.5, 0, 0, 2, \mathbf{0}_{p-5}^\top)^\top$, $\mu(0)=0$ and $\mu(1)=1$.
Furthermore, we assume $\{X_i\}_{i=1}^n$ are jointly  $\mathcal{N}(0,\Sigma)$ with the $(i,j)$-th element of $\Sigma$,  $\Sigma_{ij} = 0.5^{|i-j|}$ for $i,j \in \{1,\cdots,p\}$,  $\{\varepsilon_i\}_{i=1}^n$ are i.i.d. $\mathcal{N}(0,1)$, and $p\in\{10,150\}$.

\end{example}
\vspace{-0.3cm}
\begin{table}[h]
\caption{Simulation results Example \ref{ex:n}: (1)  $\text{Imb}_{n,J^*}^M$ and (2) $\tau$: mean and  $\sqrt{n}\times$standard deviation ($\sqrt{n}$s.d.).}
\label{tb:n_arm}
\centering
\footnotesize
\begin{tabular}{c|c|cccccc}
\hline
& \multirow{2}{*}{$p$} & \multirow{2}{*}{{\sf CR}} & \multirow{2}{*}{{\sf ARM}} &  \multicolumn{2}{c}{{\sf ARCS-M}}  \\
\cline{5-6}
& & & & $N=2$ & $N=10$ \\
\hline

\multirow{2}{*}{$\text{Imb}_{n,J^*}^M$ }  & 10 & 5.61 & 0.16 & 0.09 & 0.09 \\ 
 & 150 & 6.09 & 2.65 & 0.08 & 0.08 \\ 
 \hline
 \multirow{2}{*}{ $\tau$: mean~($\sqrt{n}$s.d.)} & 10 & 0.98~(9.12) & 1.00~(2.60) & 1.00~(2.29) & 1.00~(2.32) \\ 
&  150 & 1.01~(9.45) & 1.00~(6.25) & 1.00~(2.30) & 1.00~(2.30) \\ 
\hline

\end{tabular}
\end{table}
\vspace{-0.3cm}
\begin{table}[h]
\caption{Simulation results for Example \ref{ex:n}: (1) DNCM; (2) DNC; (3) $\text{Imb}_{n,J^*}^\phi$ in Special case \ref{ex:cov} and  (4) $\tau$: mean and $\sqrt{n}\times$standard deviation.}
\label{tb:n_cov}
\centering
\footnotesize
\begin{tabular}{c|c|cccccc}
\hline
& \multirow{2}{*}{$p$} & \multirow{2}{*}{{\sf CR}} & \multirow{2}{*}{{\sf COV}} &  \multicolumn{2}{c}{{\sf ARCS-COV}}  \\
\cline{5-6}
& & & & $N=1$ & $N=10$ \\
\hline
\multirow{2}{*}{DNCM} & 10 & 6044.10 & 166.22 & 112.58 & 116.89 \\ 
 & 150 & 5860.69 & 1054.43 & 104.49 & 106.40 \\ 
 \hline
 \multirow{2}{*}{DNC} & 10 & 24909.45 & 4277.54 & 1551.77 & 1546.48 \\ 
 & 150 & 25419.84 & 23706.31 & 1148.54 & 1157.76 \\ 
 \hline
\multirow{2}{*}{$\text{Imb}_{n,J^*}^\phi$ } & 10 & 2650.90 & 212.49 & 82.15 & 85.56 \\ 
 & 150 & 2579.55 & 1391.70 & 66.37 & 67.03 \\ 
 \hline
\multirow{2}{*}{ $\tau$: mean~($\sqrt{n}$s.d.)} & 10 & 1.01~(9.50) & 1.00~(2.40) & 0.99~(2.31) & 1.00~(2.41) \\ 
&  150 & 0.98~(9.36) & 1.00~(3.75) & 1.00~(2.21) & 0.99~(2.28) \\ 
\hline
\end{tabular}
\end{table}

\end{appendices}
\bibliographystyle{apalike}
\bibliography{reference}

\begin{thebibliography}{}

\bibitem[Antognini and Zagoraiou, 2012]{Baldi2012CARA}
Antognini, A.~B. and Zagoraiou, M. (2012).
\newblock {Multi-objective optimal designs in comparative clinical trials with
  covariates: The reinforced doubly adaptive biased coin design}.
\newblock {\em The Annals of Statistics}, 40(3):1315 -- 1345.

\bibitem[Atkinson, 1982]{atkinson1982optimum}
Atkinson, A.~C. (1982).
\newblock Optimum biased coin designs for sequential clinical trials with
  prognostic factors.
\newblock {\em Biometrika}, 69(1):61--67.

\bibitem[Atkinson, 2002]{atkinson2002comparison}
Atkinson, A.~C. (2002).
\newblock The comparison of designs for sequential clinical trials with
  covariate information.
\newblock {\em Journal of the Royal Statistical Society Series A: Statistics in
  Society}, 165(2):349--373.

\bibitem[Atkinson and Biswas, 2013]{atkinson2013randomised}
Atkinson, A.~C. and Biswas, A. (2013).
\newblock {\em Randomised Response-Adaptive Designs in Clinical Trials}.
\newblock CRC Press, Boca Raton, FL.

\bibitem[Baldi~Antognini and Zagoraiou, 2011]{baldi2011covariate}
Baldi~Antognini, A. and Zagoraiou, M. (2011).
\newblock The covariate-adaptive biased coin design for balancing clinical
  trials in the presence of prognostic factors.
\newblock {\em Biometrika}, 98(3):519--535.

\bibitem[Bickel et~al., 2009]{bickel2009simultaneous}
Bickel, P.~J., Ritov, Y., and Tsybakov, A.~B. (2009).
\newblock {Simultaneous analysis of Lasso and Dantzig selector}.
\newblock {\em The Annals of Statistics}, 37(4):1705 -- 1732.

\bibitem[Biswas et~al., 2016]{biswas2016class}
Biswas, A., Bhattacharya, R., and Park, E. (2016).
\newblock On a class of optimal covariate-adjusted response adaptive designs
  for survival outcomes.
\newblock {\em Statistical methods in medical research}, 25(6):2444--2456.

\bibitem[Browning et~al., 2022]{browning2022efficacy}
Browning, J.~C., Enloe, C., Cartwright, M., Hebert, A., Paller, A.~S., Hebert,
  D., Kowalewski, E.~K., and Maeda-Chubachi, T. (2022).
\newblock Efficacy and safety of topical nitric oxide-releasing berdazimer gel
  in patients with molluscum contagiosum: A phase 3 randomized clinical trial.
\newblock {\em JAMA Dermatology}, 158(8):871--878.

\bibitem[Cheung et~al., 2014]{cheung2014covariate}
Cheung, S.~H., Zhang, L.-X., Hu, F., and Chan, W.~S. (2014).
\newblock Covariate-adjusted response-adaptive designs for generalized linear
  models.
\newblock {\em Journal of Statistical Planning and Inference}, 149:152--161.

\bibitem[Ciolino et~al., 2019]{ciolino2019ideal}
Ciolino, J.~D., Palac, H.~L., Yang, A., Vaca, M., and Belli, H.~M. (2019).
\newblock Ideal vs. real: a systematic review on handling covariates in
  randomized controlled trials.
\newblock {\em BMC medical research methodology}, 19(136):1--11.

\bibitem[Efron, 1971]{efron1971forcing}
Efron, B. (1971).
\newblock Forcing a sequential experiment to be balanced.
\newblock {\em Biometrika}, 58(3):403--417.

\bibitem[Fan et~al., 2009]{fan2009ultrahigh}
Fan, J., Samworth, R., and Wu, Y. (2009).
\newblock Ultrahigh dimensional feature selection: beyond the linear model.
\newblock {\em The Journal of Machine Learning Research}, 10:2013--2038.

\bibitem[FDA, 2019]{FDA2019}
FDA (2019).
\newblock Adaptive design clinical trials for drugs and biologics guidance for
  industry.

\bibitem[FDA, 2023]{FDA2023}
FDA (2023).
\newblock Adjusting for covariates in randomized clinical trials for drugs and
  biological products.

\bibitem[Hu et~al., 2014]{hu2014adaptive}
Hu, F., Hu, Y., Ma, Z., and Rosenberger, W.~F. (2014).
\newblock Adaptive randomization for balancing over covariates.
\newblock {\em WIREs Computational Statistics}, 6(4):288--303.

\bibitem[Hu and Rosenberger, 2006]{hu2006theory}
Hu, F. and Rosenberger, W.~F. (2006).
\newblock {\em The Theory of Response-Adaptive Randomization in Clinical
  Trials}.
\newblock John Wiley \& Sons, Hoboken, New Jersey.

\bibitem[Hu et~al., 2023]{hu2023multi}
Hu, F., Ye, X., and Zhang, L.-X. (2023).
\newblock Multi-arm covariate-adaptive randomization.
\newblock {\em Science China Mathematics}, 66(1):163--190.

\bibitem[Hu et~al., 2015]{hu2015unified}
Hu, J., Zhu, H., and Hu, F. (2015).
\newblock A unified family of covariate-adjusted response-adaptive designs
  based on efficiency and ethics.
\newblock {\em Journal of the American Statistical Association},
  110(509):357--367.

\bibitem[Hu and Hu, 2012a]{hu2012asymptotic}
Hu, Y. and Hu, F. (2012a).
\newblock Asymptotic properties of covariate-adaptive randomization.
\newblock {\em The Annals of Statistics}, 40(3):1794 -- 1815.

\bibitem[Hu and Hu, 2012b]{hu2012balancing}
Hu, Y. and Hu, F. (2012b).
\newblock Balancing treatment allocation over continuous covariates: a new
  imbalance measure for minimization.
\newblock {\em Journal of Probability and Statistics}, 2012(1).

\bibitem[Huang et~al., 2013]{huang2013longitudinal}
Huang, T., Liu, Z., and Hu, F. (2013).
\newblock Longitudinal covariate-adjusted response-adaptive randomized designs.
\newblock {\em Journal of Statistical Planning and Inference},
  143(10):1816--1827.

\bibitem[James et~al., 2017]{james2017abiraterone}
James, N.~D., de~Bono, J.~S., Spears, M.~R., Clarke, N.~W., Mason, M.~D.,
  Dearnaley, D.~P., Ritchie, A.~W., Amos, C.~L., Gilson, C., Jones, R.~J.,
  et~al. (2017).
\newblock Abiraterone for prostate cancer not previously treated with hormone
  therapy.
\newblock {\em New England Journal of Medicine}, 377(4):338--351.

\bibitem[Kahan and Morris, 2012]{kahan2012reporting}
Kahan, B.~C. and Morris, T.~P. (2012).
\newblock Reporting and analysis of trials using stratified randomisation in
  leading medical journals: review and reanalysis.
\newblock {\em BMJ}, 345.

\bibitem[Kaye et~al., 2022]{kaye2022effect}
Kaye, K.~S., Belley, A., Barth, P., Lahlou, O., Knechtle, P., Motta, P., and
  Velicitat, P. (2022).
\newblock Effect of cefepime/enmetazobactam vs piperacillin/tazobactam on
  clinical cure and microbiological eradication in patients with complicated
  urinary tract infection or acute pyelonephritis: a randomized clinical trial.
\newblock {\em JAMA}, 328(13):1304--1314.

\bibitem[Keller et~al., 2000]{keller2000comparison}
Keller, M.~B., McCullough, J.~P., Klein, D.~N., Arnow, B., Dunner, D.~L.,
  Gelenberg, A.~J., Markowitz, J.~C., Nemeroff, C.~B., Russell, J.~M., Thase,
  M.~E., et~al. (2000).
\newblock A comparison of nefazodone, the cognitive behavioral-analysis system
  of psychotherapy, and their combination for the treatment of chronic
  depression.
\newblock {\em New England journal of medicine}, 342(20):1462--1470.

\bibitem[Li et~al., 2019]{LiContinuousCAR2019}
Li, X., Zhou, J., and Hu, F. (2019).
\newblock Testing hypotheses under adaptive randomization with continuous
  covariates in clinical trials.
\newblock {\em Statistical Methods in Medical Research}, 28(6):1609--1621.

\bibitem[Lin et~al., 2015]{lin2015pursuit}
Lin, Y., Zhu, M., and Su, Z. (2015).
\newblock The pursuit of balance: an overview of covariate-adaptive
  randomization techniques in clinical trials.
\newblock {\em Contemporary Clinical Trials}, 45(2015):21--25.

\bibitem[Ma et~al., 2024]{ma2022new}
Ma, W., Li, P., Zhang, L.-X., and Hu, F. (2024).
\newblock A new and unified family of covariate adaptive randomization
  procedures and their properties.
\newblock {\em Journal of the American Statistical Association},
  119(545):151--162.

\bibitem[Ma et~al., 2020]{ma2020statistical}
Ma, W., Qin, Y., Li, Y., and Hu, F. (2020).
\newblock Statistical inference for covariate-adaptive randomization
  procedures.
\newblock {\em Journal of the American Statistical Association},
  115(531):1488--1497.

\bibitem[Ma and Hu, 2013]{ma2013balancing}
Ma, Z. and Hu, F. (2013).
\newblock Balancing continuous covariates based on kernel densities.
\newblock {\em Contemporary Clinical Trials}, 34(2):262--269.

\bibitem[Morgan and Rubin, 2012]{morgan2012rerandomization}
Morgan, K.~L. and Rubin, D.~B. (2012).
\newblock {Rerandomization to improve covariate balance in experiments}.
\newblock {\em The Annals of Statistics}, 40(2):1263 -- 1282.

\bibitem[Mukherjee et~al., 2023]{mukherjee2023covariate}
Mukherjee, A., Coad, D.~S., and Jana, S. (2023).
\newblock Covariate-adjusted response-adaptive designs for censored survival
  responses.
\newblock {\em Journal of Statistical Planning and Inference}, 225:219--242.

\bibitem[Mukherjee et~al., 2024]{mukherjee2024covariate}
Mukherjee, A., Jana, S., and Coad, S. (2024).
\newblock Covariate-adjusted response-adaptive designs for semiparametric
  survival models.
\newblock {\em Statistical Methods in Medical Research}, page
  09622802241287704.

\bibitem[M{\"u}ller et~al., 2024]{muller2024isotonic}
M{\"u}ller, M.~M., Reeve, H.~W., Cannings, T.~I., and Samworth, R.~J. (2024).
\newblock Isotonic subgroup selection.
\newblock {\em Journal of the Royal Statistical Society Series B: Statistical
  Methodology, to appear}.

\bibitem[Pocock and Simon, 1975]{pocock1975sequential}
Pocock, S.~J. and Simon, R. (1975).
\newblock Sequential treatment assignment with balancing for prognostic factors
  in the controlled clinical trial.
\newblock {\em Biometrics}, (1):103--115.

\bibitem[Qin et~al., 2024]{qin2022adaptive}
Qin, Y., Li, Y., Ma, W., and Hu, F. (2024).
\newblock Adaptive randomization via mahalanobis distance.
\newblock {\em Statistica Sinica}, (1):353--375.

\bibitem[Ravikumar et~al., 2009]{ravikumar2009sparse}
Ravikumar, P., Lafferty, J., Liu, H., and Wasserman, L. (2009).
\newblock Sparse additive models.
\newblock {\em Journal of the Royal Statistical Society Series B: Statistical
  Methodology}, 71(5):1009--1030.

\bibitem[Reeve et~al., 2023]{reeve2023optimal}
Reeve, H.~W., Cannings, T.~I., and Samworth, R.~J. (2023).
\newblock Optimal subgroup selection.
\newblock {\em The Annals of Statistics}, 51(6):2342--2365.

\bibitem[Rosenberger and Lachin, 2015]{rosenberger2015randomization}
Rosenberger, W.~F. and Lachin, J.~M. (2015).
\newblock {\em Randomization in Clinical Trials: Theory and Practice}.
\newblock John Wiley \& Sons, Hoboken, New Jersey, 2nd edition.

\bibitem[Rosenberger and Sverdlov, 2008]{rosenberger2008handling}
Rosenberger, W.~F. and Sverdlov, O. (2008).
\newblock Handling covariates in the design of clinical trials.
\newblock {\em Statistical Science}, 23(3):404--419.

\bibitem[Rosenberger et~al., 2001]{rosenberger2001covariate}
Rosenberger, W.~F., Vidyashankar, A., and Agarwal, D.~K. (2001).
\newblock Covariate-adjusted response-adaptive designs for binary response.
\newblock {\em Journal of Biopharmaceutical Statistics}, 11(4):227--236.

\bibitem[Rudelson and Zhou, 2012]{rudelson2012reconstruction}
Rudelson, M. and Zhou, S. (2012).
\newblock Reconstruction from anisotropic random measurements.
\newblock In Mannor, S., Srebro, N., and Williamson, R.~C., editors, {\em
  Proceedings of the 25th Annual Conference on Learning Theory}, volume~23 of
  {\em Proceedings of Machine Learning Research}, pages 10.1--10.24, Edinburgh,
  Scotland. PMLR.

\bibitem[Shah and Samworth, 2013]{shah2013variable}
Shah, R.~D. and Samworth, R.~J. (2013).
\newblock Variable selection with error control: another look at stability
  selection.
\newblock {\em Journal of the Royal Statistical Society Series B: Statistical
  Methodology}, 75(1):55--80.

\bibitem[Sverdlov et~al., 2013]{sverdlov2013utility}
Sverdlov, O., Rosenberger, W.~F., and Ryeznik, Y. (2013).
\newblock Utility of covariate-adjusted response-adaptive randomization in
  survival trials.
\newblock {\em Statistics in Biopharmaceutical Research}, 5(1):38--53.

\bibitem[Taves, 1974]{taves1974minimization}
Taves, D.~R. (1974).
\newblock Minimization: a new method of assigning patients to treatment and
  control groups.
\newblock {\em Clinical Pharmacology \& Therapeutics}, 15(5):443--453.

\bibitem[Taves, 2010]{taves2010use}
Taves, D.~R. (2010).
\newblock The use of minimization in clinical trials.
\newblock {\em Contemporary Clinical Trials}, 31(2):180--184.

\bibitem[Toorawa et~al., 2009]{toorawa2009use}
Toorawa, R., Adena, M., Donovan, M., Jones, S., and Conlon, J. (2009).
\newblock Use of simulation to compare the performance of minimization with
  stratified blocked randomization.
\newblock {\em Pharmaceutical Statistics: The Journal of Applied Statistics in
  the Pharmaceutical Industry}, 8(4):264--278.

\bibitem[Vershynin, 2018]{vershynin2018high}
Vershynin, R. (2018).
\newblock {\em High-Dimensional Probability: An Introduction with Applications
  in Data Science}, volume~47.
\newblock Cambridge university press, New York.

\bibitem[Villar and Rosenberger, 2018]{villar2018covariate}
Villar, S.~S. and Rosenberger, W.~F. (2018).
\newblock Covariate-adjusted response-adaptive randomization for multi-arm
  clinical trials using a modified forward looking gittins index rule.
\newblock {\em Biometrics}, 74(1):49--57.

\bibitem[Zelen, 1974]{zelen1974randomization}
Zelen, M. (1974).
\newblock The randomization and stratification of patients to clinical trials.
\newblock {\em Journal of Chronic Diseases}, 27(7):365--375.

\bibitem[Zhang et~al., 2022]{zhang2022covariate}
Zhang, H., Hu, F., and Yin, J. (2022).
\newblock Covariate-adaptive randomization with variable selection in clinical
  trials.
\newblock {\em Stat}, 11(1):e461.

\bibitem[Zhang et~al., 2007]{zhang2007asymptotic}
Zhang, L.-X., Hu, F., Cheung, S.~H., Chan, W.~S., et~al. (2007).
\newblock Asymptotic properties of covariate-adjusted response-adaptive
  designs.
\newblock {\em The Annals of Statistics}, 35(3):1166--1182.

\bibitem[Zhou et~al., 2018]{zhou2018sequential}
Zhou, Q., Ernst, P.~A., Morgan, K.~L., Rubin, D.~B., and Zhang, A. (2018).
\newblock Sequential rerandomization.
\newblock {\em Biometrika}, 105(3):745--752.

\bibitem[Zhu and Zhu, 2023]{ZhuCARA2023}
Zhu, H. and Zhu, H. (2023).
\newblock {Covariate-Adjusted Response-Adaptive Designs Based on Semiparametric
  Approaches}.
\newblock {\em Biometrics}, 79(4):2895--2906.

\end{thebibliography}
\end{document}